\documentclass[a4paper, 12pt]{article}
\usepackage{amssymb, tikz, bbm,bm,color}
\usepackage{latexsym}
\usepackage{enumerate}
\usepackage{graphics}
\usepackage{amsmath}
\usepackage{amssymb}
\usepackage{subfig}
\usepackage{graphicx}
\usepackage{todonotes}
\usepackage{amsmath}
\usepackage{amssymb}
\usepackage{graphicx}
\usepackage{breqn}

\usepackage{tikz}
\usetikzlibrary{arrows, decorations.markings,
  decorations.pathmorphing, shapes, fit, backgrounds, patterns,
  calc,petri,snakes, arrows.meta}
  \usetikzlibrary{positioning}
\usetikzlibrary{decorations.text}
\usetikzlibrary{decorations.pathmorphing}
\tikzset{>=latex}

\definecolor{arrowblue}{RGB}{98,145,224}
\usetikzlibrary{matrix,arrows,fit} 

\tikzset{circarrow/.style={
        *->,
        shorten <=-2pt
    }
}
\tikzstyle{place}=[circle,draw=blue!50,fill=blue!20,thick]
\tikzstyle{placespec}=[circle,draw=blue!50,fill=blue!20,thick, inner sep=0pt,minimum size=6.5mm]
\tikzstyle{place1}=[circle,draw=blue!50,fill=blue!20,thick,inner sep=0pt,minimum size=2mm]
\tikzstyle{place11}=[circle,draw=blue!50,fill=blue!20,thick, inner sep=0pt,minimum size=0mm]
\tikzstyle{placetab}=[circle,draw=blue!50,fill=blue!20,thick, inner sep=0pt,minimum size=2mm]
\tikzstyle{transition}=[rectangle,draw=blue!50,fill=blue!20,thick]
\tikzstyle{vecArrow} = [thick, decoration={markings,mark=at position
   1 with {\arrow[semithick]{open triangle 60}}},
   double distance=1.4pt, shorten >= 5.5pt,
   preaction = {decorate},
   postaction = {draw,line width=1.4pt, white,shorten >= 4.5pt}]

\tikzset{
BasicNode/.style={circle, draw= black!50, fill=colora!40, thin, minimum size
  = 5mm, inner sep = 0mm},
SusCirc/.style={BasicNode,fill=colorS!40},
SusCircBig/.style={BasicNode,fill=colorS!40, minimum size=10mm},
InfCirc/.style={BasicNode,fill=colorI!40},
InfCircBig/.style={BasicNode,fill=colorI!40,minimum size=10mm},
RecCirc/.style={BasicNode,fill=colorR!40},
RecCircBig/.style={BasicNode,fill=colorR!40, minimum size=10mm},
DefCirc/.style={BasicNode}
}

\tikzset{
LargeDot/.style = {circle, draw = black, fill = black, thick, minimum
  size = 1mm, inner sep = 0mm}
}

\tikzstyle{BigCirc}=[circle,draw=black!50,thick,inner sep=0pt,minimum size=0pt]

\tikzstyle{default} = [draw, minimum size = 3em, text width = 4em, text centered]
\tikzstyle{wide}=[draw, minimum size=3em, text width=7.5em, text
centered]
\tikzstyle{shortbox}=[draw, minimum size=2em, text width=4.5em, text centered]
\tikzstyle{bigbox}=[draw, inner sep=20pt,label={[align=right,shift={(-1.5ex,3ex)}]south east:\llap{#1}}]
\tikzstyle{box} = [draw, minimum size=2em, text centered]

\tikzstyle{decay} = [draw, ->, decorate, decoration = {snake,
  amplitude = 0.4mm, segment length=2mm, post length = 2mm, pre length
= 1mm}]

\tikzstyle{ClearCirc}=[circle,draw=black!50,fill=white!20,thick, inner sep=1pt]%a circle

\definecolor{colora}{RGB}{0,115,179}%+}{RGB}{0,100,255}
\definecolor{colorb}{RGB}{230,154,0} %+
\definecolor{colorc}{RGB}{0,154,128} %colorS +
\definecolor{colord}{RGB}{205,10,179}%colorR
\definecolor{colore}{RGB}{255,32,0}%colorI +
\definecolor{colorf}{RGB}{240,228,66}%+
\definecolor{colorg}{RGB}{90,179,230}%+
\definecolor{colorh}{RGB}{205,154,179}

\definecolor{colorS}{RGB}{0,154,128}
\definecolor{colorI}{RGB}{255,32,0}
\definecolor{colorR}{RGB}{205,10,179}
\definecolor{colorE}{RGB}{240,228,66} %exposed

\definecolor{peach}{HTML}{E69F00}
\definecolor{blue}{HTML}{56B4E9}
\definecolor{green}{HTML}{009E73}
\definecolor{yellow}{HTML}{F0E442}
\definecolor{blue2}{HTML}{0072B2}
\definecolor{grey}{HTML}{878787}
\definecolor{pink}{HTML}{CC79A7}

\usepackage{todonotes}

\usepackage[english]{babel}
\usepackage{lineno}
\usepackage{color}

\usepackage[utf8]{inputenc}
\usepackage[english]{babel}
 
\usepackage{hyperref}
\hypersetup{
    colorlinks=true,
    linkcolor=blue,
    filecolor=magenta,      
    urlcolor=cyan,
}
 
\usepackage{enumerate}
\usepackage[english]{babel}
\usepackage{amsmath,amssymb,amsthm}
\usepackage{wasysym}
\usepackage{stmaryrd}
\usepackage{breqn}
\usepackage{lineno}
\newcommand{\be}{\begin{equation}}
\newcommand{\ee}{\end{equation}}
\newcommand{\bea}{\begin{eqnarray}}
\newcommand{\eea}{\end{eqnarray}}

\newcommand{\nod}{\noindent}

\newcommand{\ba}{\begin{array}}
\newcommand{\ea}{\end{array}}
\newcommand{\bc}{\begin{center}}
\newcommand{\ec}{\end{center}}

\newtheorem{prop}{Proposition}

\RequirePackage{color}
\definecolor{nicos-red}{rgb}{0.75,0.0,0.0}
\definecolor{light-gray}{gray}{0.6}
\definecolor{really-light-gray}{gray}{0.8}
\definecolor{sussexg}{rgb}{0.0,0.5,0.5}
\definecolor{sussexp}{rgb}{0.5,0.0,0.5}

\usepackage{nicefrac}
\numberwithin{equation}{section}
\theoremstyle{remark}

\definecolor{darkgreen}{rgb}{0.0,0.5,0.0}
\definecolor{darkblue}{rgb}{0.0,0.0,0.3}
\definecolor{nicosred}{rgb}{0.65,0.1,0.1}
\definecolor{light-gray}{gray}{0.7}
\allowdisplaybreaks[1]

\allowdisplaybreaks

\begin{document}

% The following line gets Overleaf to record the name of the paper - no need to uncomment!
%\title{Master manuscript -- Epidemic threshold in pairwise models for clustered networks: closures and fast correlations}

\begin{center}
\large{\textbf{Epidemic threshold in pairwise models for clustered networks: closures and fast correlations}}\\
\vspace{0.5cm}
\textit{Rosanna C. Barnard $^{1}$, Luc Berthouze $^2$, P\'eter L. Simon  $^3$ \& Istv\'an Z. Kiss $^{1,*}$}
\end{center}
\vspace{0.0cm}
\begin{center}
$^1$ School of Mathematical and Physical Sciences, Department of
Mathematics, University of Sussex, Falmer,
Brighton BN1 9QH, UK\\
$^2$ Centre for Computational Neuroscience and Robotics, University of Sussex,
Falmer, Brighton, BN1 9QH, UK\\
$^3$ Institute of Mathematics, E\"otv\"os Lor\'and University Budapest, and
Numerical Analysis and\\
Large Networks Research Group, Hungarian Academy of Sciences, Hungary
\end{center}
\vspace{2.5cm}

%%%%%%%%%%%%%%%%%%%%%%%%%%%
% Abstract
%%%%%%%%%%%%%%%%%%%%%%%%%%%
\begin{center}
\textbf{Abstract}\\
\end{center}
The epidemic threshold is probably the most studied quantity in the modelling of epidemics on networks. For a large class of networks and dynamics the epidemic threshold is well studied and understood. However, it is less so for clustered networks where theoretical results are mostly limited to idealised networks. In this paper we focus on a class of models known as pairwise models where, to our knowledge, no analytical result for the epidemic threshold exists. We show that by exploiting the presence of fast variables and using some standard techniques from perturbation theory we are able to obtain the epidemic threshold analytically. 
%More precisely, the threshold is obtained as an asymptotic expansion in terms of powers of the clustering coefficient. 
We validate this new threshold by comparing it to the numerical solution of the full system. The agreement is found to be excellent over a wide range of values of the clustering coefficient, transmission rate and average degree of the network. Interestingly, we find that the analytical form of $R_0$ depends on the choice of closure, highlighting the importance of model choice when dealing with real-world epidemics. Nevertheless, we expect that our method will extend to other systems in which fast variables are present. \\

\nod {\bf Keywords:} network, epidemic, pairwise model, clustering, correlation, fast variables\\

\vspace{0.0cm}
\noindent $^{*}$ corresponding author: i.z.kiss@sussex.ac.uk
%Networks are widespread in natural and manmade systems, and as a result network science
%provides an ideal framework for modelling systems ranging from molecular and social to technological.
%This relatively new and very productive research area has seen tremendous growth over the last decade with contributions
%from many different areas starting to crystallise into a rigorous common mathematical language, with
%\textit{control of networks}, understood in the widest sense, emerging as a prominent and active area.
%\[
%5U \sim Unif[0, 1].  X = \frac{1}{\sqrt[\alpha-1]{1 - U}} - 1 
%\]
\newpage
%%%%%%%%%%%%%%%%%%%%%%%%%%%%%%%
\section{Introduction} \label{Intro}
%%%%%%%%%%%%%%%%%%%%%%%%%%%%%%%
Epidemic dynamics on networks, being susceptible-infected-susceptible (SIS), susceptible-infected-recovered/removed (SIR) or otherwise, are often modelled as continuous time Markov chains with discrete but extremely large state spaces of order $m^{N}$, where $m$ is the number of different disease statuses (e.g. $m=2$ for SIS and $m=3$ for SIR) and $N$ is the number of nodes/individuals in the network. This makes the analysis of the resulting exact system almost impossible, except for some specific network topologies such as the fully connected network, networks with considerable structural symmetry or networks with few nodes \cite{kiss2017mathematics,holme2017three} that allow for simplification.

This problem, instead, has been dealt with by focusing on mean-field models where the goal is to derive, often heuristically, a system of ordinary or integro-differential equations that describe (non-Markovian) epidemics for some average quantities, such as the expected number of nodes in various states, the expected number of links in various states or the expected number of star-like structures (focusing on a node and all of its neighbours). These methods usually rely on closures to break the dependency on higher-order moments (e.g. the expected number of nodes in a state depends on the expected number of links in certain states and so on). Such an approach has led to a number of models including heterogeneous or degree-based mean-field~\cite{pastor2001epidemic,pastor2015epidemic}, pairwise~\cite{rand1999correlation,keeling1999effects}, effective-degree~\cite{lindquist2011effective}, edge-based compartmental~\cite{miller2012edge} and message passing~\cite{karrer2010message}, to name a few. These models essentially differ in the choice of variables over which the averaging is done. Perhaps the most compact model with the fewest number of equations is the edge-based compartmental model~\cite{miller2013model} and this works for heterogeneous networks with Markovian SIR epidemics, although extensions of it for arbitrary infection and recovery processes are also possible~\cite{sherborne2018mean}.

Pairwise models have been extremely popular and the very first model for regular networks and SIR epidemics~\cite{rand1999correlation,keeling1999effects} has been generalised to heterogeneous networks~\cite{eames2002modeling}, preferentially mixing networks~\cite{eames2002modeling}, directed~\cite{sharkey2006pair} and weighted networks~\cite{rattana2013class}, adaptive networks~\cite{gross2006epidemic,kiss2012modelling,szabo2016oscillating}, and structured networks~\cite{house2009motif} among others. Perhaps this is due to its relative simplicity and transparency, whereby variables seem to make sense in a straightforward way and a basic understanding of the network and epidemic dynamics coupled with good bookkeeping leads to a valid and analytically tractable model. Pairwise models have been successfully used to derive analytically the epidemic threshold and final epidemic size, with these results mostly limited to networks without clustering. The propensity of contacts to cluster, i.e. that two friends of an individual/node are also friends of each other, is known to lead to many complications, and modelling epidemics on clustered networks using analytically tractable mean-field models is still limited to networks with very specific structural features~\cite{house2009motif,newman2009random,miller2009percolation,miller2009spread,karrer2010random,volz2011effects,ritchie2016beyond}. However, using approaches borrowed from percolation theory~\cite{miller2009spread} and focusing more on the stochastic process itself~\cite{trapman2007analytical}, some results have been obtained. For example, in~\cite{miller2009spread} it was shown that for the susceptible-infected-recovered (SIR) epidemic on clustered networks with heterogeneous degree distributions, the basic reproduction number is given by
\begin{equation}
R_0=\frac{\langle k^2-k\rangle}{\langle k\rangle}T-\frac{2\langle n_{\triangle}\rangle}{\langle k\rangle}T^2+\cdots,
\end{equation}
where $\langle k^i \rangle$ stands for the $i$th moment of the degree distribution, $T$ is the probability of infection spreading across a link connecting an infected to a susceptible node and $\langle n_{\triangle}\rangle$ denotes the average number of triangles that a node belongs to. 
The first positive term corresponds to what the threshold is for configuration-type networks with no clustering. The second term, which is negative, shows that clustering reduces the epidemic threshold when compared to the unclustered case. 

%Furthermore, if the network is regular and we assume that infections and recoveries are Markovian processes with rates $\tau$ and $\gamma$ respectively, giving $T=\tau/(\tau+\gamma)$, $R_0$ above reduces to
%\begin{equation}
%5R_0=\frac{\tau (n-1)}{\tau+\gamma}-(n-%1)\phi\left(\frac{\tau}{\tau+\gamma}\right)^2+\cdots,
%\end{equation}
%where we have used the fact that a global clustering coefficient of $\phi$ translates to a node on average being part of $\frac{1}{2}n(n-1)\phi$ uniquely counted triangles.

For pairwise models, clustering first manifests itself by requiring a different and more complex closure, which makes the analysis of the resulting system, even for regular networks and SIR dynamics, challenging. Furthermore, it turns out that such closure may in fact fail to conserve pair-level relations and may not accurately reflect the early growth of quantities such as closed loops of three with all nodes being infected~\cite{house2010impact}. Such considerations have led to an improved closure being developed in an effort to keep as many true features of the exact epidemic process as possible~\cite{house2010impact}. In this paper we will focus on the classic pairwise model for regular networks with clustering, using both the simplest closure and a variant of the improved closure. We will show that by working with two fast variables corresponding to the correlations that develop during the spread of the epidemic, we can analytically determine the epidemic threshold as an asymptotic expansion in terms of the clustering coefficient. 

The use of fast variables is not completely new and has been used in~\cite{keeling1999effects} and~\cite{eames2008modelling} but the epidemic threshold has only been obtained numerically and it was framed in terms of a growth-rate-based threshold which of course is equivalent to the basic reproduction number at the critical point of the epidemic spread. In~\cite{eames2008modelling} a hybrid pairwise model incorporating random and clustered contacts is considered, with the analysis focused on the growth-rate-based threshold. The authors of \cite{eames2008modelling} managed to derive a number of results, some analytic (the critical clustering coefficient for which an epidemic can take off) and some semi-analytic, and they have shown, in agreement with most studies, that clustering inhibits the spread of the epidemic when compared to an equivalent network without clustering but with equivalent parameter values governing the epidemic process. However, no analytic expression for the threshold was provided. 

More recently, in~\cite{li2018epidemic}, the epidemic threshold in a pairwise model for clustered networks with closures based on the number of links in a motif, rather than nodes, was calculated as
\begin{equation}
R_0=\frac{(n-1)\tau}{\tau+\gamma+\tau\phi}, %\label{eq:disc_threshold_example}
\end{equation}
where $n$ is the average number of links per node, $\phi$ is the global clustering coefficient (see later sections for a formal definition) and $\tau$ and $\gamma$ are the infection and recovery rates, respectively. The expression above can be expanded in terms of $\phi$ to give 
\begin{equation}
R_0=\frac{(n-1)\tau}{\tau+\gamma}\left(\frac{1}{1+\phi\frac{\tau}{\tau+\gamma}}\right)\simeq \frac{(n-1)\tau}{\tau+\gamma}\left( 1-\phi \frac{\tau}{\tau+\gamma}+\cdots\right),
%\label{eq:motif_based_threshold_expanded}
\end{equation}
which again reflects that clustering reduces the epidemic threshold. 

Building on these results, and effectively extending the work in~\cite{keeling1999effects,eames2008modelling}, our paper sets out to present a method to determine the epidemic threshold analytically and apply it in the context of pairwise models with two different closures for clustered networks. The paper is structured as follows. In Section 2 we outline the model with closures for unclustered and clustered networks discussed in Section 3. In Section 4 we briefly review existing results and approaches for the pairwise model with the simple closure and then focus on the correlation structure in terms of fast variables, showing that the epidemic threshold can be expressed via the solution of a cubic polynomial. This key solution is determined numerically and analytically as an asymptotic expansion in terms of the clustering coefficient. In Section 5 we show that our approach works for a compact version of the improved closure, thus validating and generalising our approach.
%is used and we also investigate the sensitivity of the pairwise model when the level of clustering is varied. This is done by systematically comparing the final epidemic size obtained for different values of the clustering coefficient. 
Finally we conclude with a discussion of the results, including comparing the threshold to other known results and touching upon a number of possible extensions.

%%%%%%%%%%%%%%%%%%%%%%
\section{Model formulation} \label{ModelForm}
%%%%%%%%%%%%%%%%%%%%%%%

%%%%%%%%%%%%%%%%%%%%%%
\subsection{The network} \label{ModelForm_network}
%%%%%%%%%%%%%%%%%%%%%%%

We begin by considering a population of $N$ individuals with its contact structure described by an undirected network with adjacency matrix $G=(g_{ij})_{i,j=1,2,\dots, N}$ where $g_{ij}=1$ if nodes $i$ and $j$ are connected and zero otherwise. Self-loops are excluded, so $g_{ii}=0$ and $g_{ij}=g_{ji}$ for all $i,j=1,2, \dots N$. The network is static and regular, such that each individual has exactly $n$ edges or links. The sum over all elements of $G$ is defined as $||G||=\sum_{i,j}g_{ij}$. Hence, the number of doubly counted links in the network is $||G||=nN$. More importantly, using simple matrix operations on $G$, we can calculate the clustering coefficient of the network 
\begin{equation} \label{equation:phi_clustering}
\phi=\frac{trace(G^{3})}{||G^{2}||-trace(G^{2})}, 
\end{equation}
where $trace(G^{3})$ yields six times the number of closed triples or loops of length three (uniquely counted) and $||G^{2}||-trace(G^{2})$, twice the number of triples (open and closed, also uniquely counted).

%%%%%%%%%%%%%%%%%%%%%%
\subsection{SIR dynamics} \label{ModelForm_SIRDynamics}
%%%%%%%%%%%%%%%%%%%%%%%
The standard SIR epidemic dynamics on a network is considered. The dynamics are driven by two processes: (a) infection and (b) recovery from infection. Infection can spread from an infected/infectious node to any of its susceptible neighbours and this is modelled as a Poisson point process with per-link infection rate $\tau$. Infectious nodes recover from infection at constant rate $\gamma$.

%%%%%%%%%%%%%%%%%%%%%%
\subsection{The unclosed pairwise model} \label{ModelForm_UnClosedPWModel}
%%%%%%%%%%%%%%%%%%%%%%%
Let $A_{i}$ equal 1 if the individual at node $i$ is of type $A$ and equal zero otherwise. Then single nodes (singles) of type $A$ can be counted as $[A]=\sum_{i}A_{i}$, pairs of nodes (pairs) of type $A-B$ can be counted as $[AB]=\sum_{i,j}A_{i}B_{j}g_{ij}$ and triples of nodes (triples) of type $A-B-C$ can be counted as $[ABC]=\sum_{i,j,k}A_{i}B_{j}C_{k}g_{ij}g_{jk}$. This method of counting means that pairs are counted once in each direction, so $[AB]=[BA]$, and $[AA]$ is even. Using this notation to keep track of singles, pairs and triples leads to the following system of pairwise equations describing the SIR epidemic on a regular network:

\begin{eqnarray} 
\dot{[S]}&=&-\tau[SI], \label{equations:unclosed_pw_model_one} \\ 
\dot{[I]}&=&\tau[SI] -\gamma[I],\label{equations:unclosed_pw_model_two} \\ 
\dot{[SI]}&=&\tau([SSI]-[ISI]-[SI])-\gamma[SI], \label{equations:unclosed_pw_model_four}\\
\dot{[SS]}&=&-2\tau[SSI], \label{equations:unclosed_pw_model_three} \\ % Rosie swapped the ordering so that \dot{[SI]} comes before \dot{[SS}]}, in agreement with the rest of this paper
\dot{[II]}&=&2\tau([ISI]+[SI])-2\gamma[II]. \label{equations:unclosed_pw_model_end} \end{eqnarray}

We note that equations \eqref{equations:unclosed_pw_model_four}-\eqref{equations:unclosed_pw_model_end} contain triples which are not defined within the entire system of equations \eqref{equations:unclosed_pw_model_one}-\eqref{equations:unclosed_pw_model_end}. The flow between compartments and the rates are illustrated in Fig.~\ref{fig:SIR_PW_flow}. To determine solutions of the system, we must find a way to account for these triples in terms of pairs and singles, a method referred to as \emph{closing the system}.

\begin{figure}
\begin{center}
  \begin{tikzpicture}
%    \node[box, fill=colorS!60] (S) at ( 0,0) {$[S]$};
%    \node[box, fill=colorI!60] (I) at (2.5,0) {$[I]$};
%    \path [->, above] (S.30) edge node {$\tau [SI]$} (I.150);
%    \path [decay] (I.210) -- (S.330) node [ midway, below] {$\gamma [I]$} ;
%
%    \node[box, fill=colorS!60] (SS) at ( 1.25,-3) {$[SS]$};
%    \node[box, fill=colorI!25] (IS) at ( 2.5,-4.25) {$[IS]$};
%    \node[box, fill=colorS!60] (SI) at ( 0,-4.25) {$[SI]$};
%    \node[box, fill=colorI!60] (II) at (1.25,-5.5) {$[II]$};
%    \path [->, above, sloped] (SS.150) edge node {$\tau [SSI]$} (SI.120);
%    \path [->, above, sloped] (SS.30) edge node {$\tau [ISS]$} (IS.60);
%    \path [->, below, sloped] (IS.225) edge node {$\tau [ISI]$} (II.45);
%    \path [->, below, sloped, bend left] (IS.300) edge node {$\tau [IS]$} (II.330);
%    \path [->, below, sloped] (SI.315) edge node {$\tau [ISI]$} (II.135);
%    \path [->, below, sloped, bend right] (SI.240) edge node {$\tau [SI]$} (II.210);
%    \path [decay] (IS.150) -- (SS.300) node [midway, below, sloped] {$\gamma [IS]$};
%    \path [decay] (SI.30) -- (SS.240) node [midway, below, sloped] {$\gamma [SI]$};
%    \path [decay] (II.100) -- (SI.350) node [midway, above, sloped] {$\gamma [II]$};
%    \path [decay] (II.80) -- (IS.190) node [midway, above, sloped] {$\gamma [II]$};

    \node[box, fill=colorS!60] (S1) at (2,-2.8) {$[S]$};
    \node[box, fill=colorI!60] (I1) at (2,-6.3) {$[I]$};
    \node[box, fill= colorR!60] (R1) at (2,-9.8) {$[R]$};
    \path [->, right] (S1) edge node {$\tau [SI]$} (I1);
    \path [decay] (I1) -- (R1) node [midway, right] {$\gamma [I]$};

%    \node[box, fill=colorS!60] (SS) at (7.5,-2) {$[SS]$};   %(7.5,-2)
%    \node[box, fill=colorI!25] (IS) at ( 8.75,-3.25) {$[IS]$};
%    \node[box, fill=colorS!60] (SI) at ( 6.25,-3.25) {$[SI]$};
%    \node[box, fill=colorI!60] (II) at (7.5,-4.5) {$[II]$};
%    \node[box, fill = colorR!25] (RS) at (10,-4.5) {$[RS]$};
%    \node[box, fill = colorS!60] (SR) at (5,-4.5) {$[SR]$};
%    \node[box, fill = colorR!25] (RI) at (8.75,-5.75) {$[RI]$};
%    \node[box, fill = colorI!60] (IR) at (6.25,-5.75) {$[IR]$};
%    \node[box, fill = colorR!60] (RR) at (7.5,-7) {$[RR]$};
     \node[box, fill=colorS!60] (SS) at (10.5,-2.8) {$[SS]$};   %(7.5,-2)
    \node[box, fill=colorI!25] (IS) at (12.25,-4.55) {$[IS]$};
    \node[box, fill=colorS!60] (SI) at (8.75,-4.55) {$[SI]$};
    \node[box, fill=colorI!60] (II) at (10.5,-6.3) {$[II]$};
    \node[box, fill = colorR!25] (RS) at (14,-6.3) {$[RS]$};
    \node[box, fill = colorS!60] (SR) at (7,-6.3) {$[SR]$};
    \node[box, fill = colorR!25] (RI) at (12.25,-8.05) {$[RI]$};
    \node[box, fill = colorI!60] (IR) at (8.75,-8.05) {$[IR]$};
    \node[box, fill = colorR!60] (RR) at (10.5,-9.8) {$[RR]$};
    %%%%%%%%%%%%%%%%%%%%%%%%%%%%
    \path [->, above, sloped] (SS) edge node {$\tau [SSI]$} (SI);
    \path [->, above, sloped] (SS) edge node {$\tau [ISS]$} (IS);
    \path [decay] (SI) -- (SR)  node [midway, above, sloped] {$\gamma [SI]$};
    \path [->,above, sloped] (SI.337) edge node {$\tau [SI] $} (II.113);
    \path [->, above, sloped] (IS.203) edge node {$\tau [IS] $} (II.67);
    \path [->,below, sloped] (SI.293) edge node {$\tau [ISI]$} (II.157);
    \path [->, below, sloped] (IS.247) edge node {$\tau [ISI]$} (II.23);
    \path [decay] (IS) -- (RS)  node [midway, above, sloped] {$\gamma [IS]$};
    \path [->, below, sloped] (SR) edge node {$\tau [ISR]$} (IR) ;
    \path [decay] (II) -- (RI) node [midway, above, sloped] {$\gamma [II]$};
    \path [decay] (II) -- (IR) node [midway, above, sloped] {$\gamma [II]$};
    \path [->, below, sloped] (RS) edge node {$\tau [RSI]$} (RI) ;
    \path [decay] (RI) -- (RR) node [midway, below, sloped] {$\gamma [RI]$};
    \path [decay] (IR) -- (RR) node [midway, below, sloped] {$\gamma [IR]$};
  \end{tikzpicture}
\end{center}
\caption{Flow diagrams showing the flux between compartments of singles (left) and compartments of pairs (right). In the compartments of pairs, straight arrows denote infections coming from within the pair (with a rate depending on a pair) or from outside the pair (with a rate depending on a triple), and curved arrows denote a recovery.  The colour indicates the status of the ``first'' node in the pair.  Symmetry allows us to conclude that some of the variables (see lighter shaded variables on the right hand side of the pairs diagram) must equal their symmetric version (e.g. $[RS]=[SR]$), so we do not need to directly calculate both quantities.}
\label{fig:SIR_PW_flow}
\end{figure}
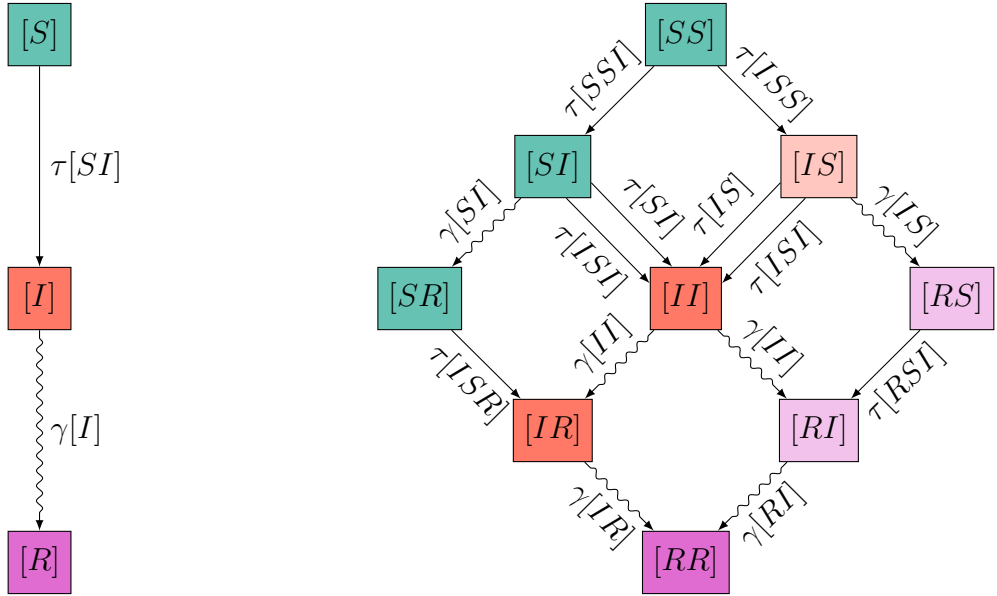

%%%%%%%%%%%%%%%%%%%%%%
\section{Closures} \label{sec:closures}
%%%%%%%%%%%%%%%%%%%%%%

%{\color{red}\textbf{To describe the approach from Rand's chapter in the \cite{} (Peter and Istvan can attempt this)}}

A quick inspection of the unclosed pairwise system (\ref{equations:unclosed_pw_model_one})-(\ref{equations:unclosed_pw_model_end}) reveals that only triples of type $[ASI]$ need closing, with $A\in\{S, I\}$. These triples, as well as triples of type $[RSI]$, are illustrated in Fig.~\ref{fig:ig:clo_around_S_I} for unclustered and clustered networks. %In fact, due to removed/recovered nodes not playing a role in the dynamics and due to their number being %negligible at the start of an epidemic, it is often assumed that $A\in \{S,I\}$.

\begin{figure}
\centering
    % fixed S-i star structure, no clust
        \subfloat[]{
            \begin{tikzpicture}
              \tikzset{localnode/.style={BasicNode}}
%                \node[localnode] (a) at (-1,1) {$N-1$};
                \node[RecCircBig] (lefttop) at ( -1,1.3) {$X_{n-1}$};
                \node[SusCircBig] (middle) at ( 0,0) {{\large $S$}};
                %\node[InfCircBig] (middletop) at ( 0,1.4) {$S$};
                \node[InfCircBig] (righttop) at ( 1,1.3) {$I$};
               % \node[RecCircBig] (middleright) at (1.4,0) {$X_2$};
                \node[RecCircBig] (rightbottom) at (1,-1.3) {$X_1$};
                %\node[RecCircBig] (middlebottom) at ( 0,-1.4) {$X_2$};
                \node[RecCircBig] (leftbottom) at (-1,-1.3) {$X_2$};
                %\coordinate (topcoord) at (-1.2, 0.5);
                %\coordinate (bottomcoord) at (-1.2, -0.5);
                % these are the dots, dots dots
                \coordinate (dot1) at (-1.3,-0.4) {};
                \coordinate (dot2) at (-1.5,0) {};
                \coordinate (dot3) at (-1.3,0.4) {};
                \draw [-] (lefttop) -- (middle);
                \draw [-] (righttop) -- (middle);
               %\draw [-] (middletop) -- (middle);
               \draw [-] (middle) -- (leftbottom);
                \draw [-] (middle) -- (rightbottom);
                \draw [dashed] (dot1) -- (middle);
                                \draw [dashed] (dot2) -- (middle);
                                                \draw [dashed] (dot3) -- (middle);
                %\draw [-] (middlebottom) -- (middle);
                %\draw [-] (middleright) -- (middle);
%                \draw [-] (bottomcoord) -- (middle);
%                \draw [-] (topcoord) -- (middle);
            \end{tikzpicture}
        \label{fig:ig:clo_around_S_I_no_clust}
        }
       \hspace{2cm}
    %     % fixed S-i star structure, no clust
        \subfloat[]{
  \begin{tikzpicture}
              \tikzset{localnode/.style={BasicNode}}
%                \node[localnode] (a) at (-1,1) {$N-1$};
                \node[RecCircBig] (lefttop) at ( -1,1.3) {$X_{n-1}$};
                \node[SusCircBig] (middle) at ( 0,0) {{\large $S$}};
                %\node[InfCircBig] (middletop) at ( 0,1.4) {$S$};
                \node[InfCircBig] (righttop) at ( 1,1.3) {$I$};
               % \node[RecCircBig] (middleright) at (1.4,0) {$X_2$};
                \node[RecCircBig] (rightbottom) at (1,-1.3) {$X_1$};
                %\node[RecCircBig] (middlebottom) at ( 0,-1.4) {$X_2$};
                \node[RecCircBig] (leftbottom) at (-1,-1.3) {$X_2$};
                %\coordinate (topcoord) at (-1.2, 0.5);
                %\coordinate (bottomcoord) at (-1.2, -0.5);
                % these are the dots, dots dots
                \coordinate (dot1) at (-1.3,-0.4) {};
                \coordinate (dot2) at (-1.5,0) {};
                \coordinate (dot3) at (-1.3,0.4) {};
                \draw [-] (lefttop) -- (middle);
                \draw [-] (righttop) -- (middle);
               %\draw [-] (middletop) -- (middle);
               \draw [-] (middle) -- (leftbottom);
                \draw [-] (middle) -- (rightbottom);
                \draw [dashed] (dot1) -- (middle);
                \draw [dashed] (dot2) -- (middle);
                \draw [dashed] (dot3) -- (middle);
                \draw [dashed,->,red] (righttop) to[bend right=30] (lefttop);
                \draw [dashed,->,red] (righttop) to[bend left=30] (rightbottom);
                \draw [dashed,->,red] (righttop) to[bend left=40] (leftbottom);                                                                                                              
                %\draw [-] (middlebottom) -- (middle);
                %\draw [-] (middleright) -- (middle);
%                \draw [-] (bottomcoord) -- (middle);
%                \draw [-] (topcoord) -- (middle);
            \end{tikzpicture}        \label{fig:clo_around_S_I_clust}
        }
    \caption{General setup for a central susceptible node with a given infected neighbour for (a) unclustered and (b) clustered regular networks with degree $n$. Dashed arrows indicate that the infected node may be connected to the other neighbours of the central susceptible node.  Random variables $X_1, X_2, \dots, X_{n-1}$ take values from the set $\{S, I, R\}$.}
    \label{fig:ig:clo_around_S_I}
\end{figure}
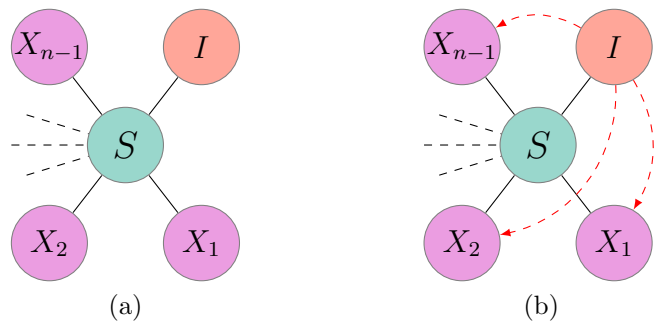

%%%%%%%%%%%%%%%%%%%%%%%%%%%%%%%%
\subsection{Closure for unclustered networks}
%%%%%%%%%%%%%%%%%%%%%%%%%%%%%%%%

First, we consider the situation depicted in Fig~\ref{fig:ig:clo_around_S_I_no_clust}. Several observations can be made. The expected number of $A-S$ type links is $[AS]$ and the total number of links emanating from susceptible nodes counted across the whole network is $n[S]$. Hence, the most straightforward approximation would be to assume that $X_{i}$, with $i=1,2,\dots,n-1$, are independent and identically Bernoulli distributed random variables with a probability of success being equal to $p_{A|S-I}^{uc}=\frac{[AS]}{n[S]}$. Averaging across the whole network leads to
\begin{equation}
[ASI]=[SI](n-1)p_{A|S-I}^{uc}=\frac{n-1}{n}\frac{[AS][SI]}{[S]}.
\label{eq:basic_clo}
\end{equation}
It is important to note that the new, closed system is effectively an approximation of the exact pairwise model and one should question if the closure \eqref{eq:basic_clo} conserves the properties of the stochastic process and of the counting on the network. For example, it is expected that in the closed system the number of nodes is conserved, i.e. $[S]+[I]+[R]=N$. Furthermore, the number of pairs of different types must sum to $nN$. More subtle conditions refer to the conservation of link types at node level ($[SS]+[SI]+[SR]=n[S]$) and pair level ($[SSI]+[ISI]+[RSI]=(n-1)[SI]$), respectively. It turns out that the closure for unclustered networks \eqref{eq:basic_clo} conserves these relations \cite{kiss2017mathematics}. Finally, the validity or appropriateness of closures can be empirically assessed by looking at the initial growth rate of the number of open and closed triples, where the number of open triples comprised of three infectious nodes should grow differently to the number of such closed triples. Of course such subtle tests/comparisons  are usually preceded by direct comparisons between the numerical solution of the closed pairwise system and explicit stochastic network simulations for a range of parameters. Such tests initially focus on prevalence of infection and final epidemic size but may include expected number of pairs. 
 
%%%%%%%%%%%%%%%%%%%%%%%%%%%%%%%%
\subsection{Closures for clustered networks}
%%%%%%%%%%%%%%%%%%%%%%%%%%%%%%%%

%%%%%%%%%%%%%%%%%%%%%%%%%%%%%%%%
\subsubsection{Simple closure}
%%%%%%%%%%%%%%%%%%%%%%%%%%%%%%%%
The presence of closed loops of length three, as illustrated in Fig.~\ref{fig:clo_around_S_I_clust}, introduces some complications. Namely, a neighbour of the central susceptible that is itself connected to an infected neighbour of the central node, is less likely to be susceptible due to the added pressure from the infected neighbour, when compared to the case when the force of infection is distributed evenly, as is the case for the closure for unclustered networks \eqref{eq:basic_clo}. More precisely, the epidemic process on the network displays clear correlations. In \cite{cator2014nodal} it has been shown that the exact $SIS$ and $SIR$ epidemics on networks are non-negatively correlated in the sense that $\mathbf{P}(I_{i}I_{j})\ge \mathbf{P}(I_i)\mathbf{P}(I_j)$. Here, $\mathbf{P}(I_iI_j)$ represents the probability that nodes $i$ and $j$, connected by a link, are both infected, while $\mathbf{P}(I_i)$ stands for the probability of node $i$ being infected. For this result to hold, all processes must be Markovian and infection rates across all links and recovery rates of all nodes have to be fixed a priori. Using the pairwise model for an $SIS$ epidemic on an unclustered network with closure~\eqref{eq:basic_clo}, it has been shown that the same correlation is preserved when averaging at the population level~\cite{kiss2017mathematics}. While the proof has not been executed for the pairwise $SIR$ model, intuitively we expect to find the same correlation structure. Based on these observations, we assume that the correlation structure in exact $SIS$ and $SIR$ epidemics on networks averaged at the population level is maintained. Hence, the inequalities
\begin{equation}
[SI]\le n[S]\frac{[I]}{N}, \,\,\, [II]\ge n[I]\frac{[I]}{N}, \,\,\, \text{and} \,\,\, [SS]\ge n[S]\frac{[S]}{N},
\label{eq:correl_ineq}
\end{equation}
hold, where $[AB]$ and $[A]$ with $A,B\in \{S,I\}$ represent the expected counts of pairs and singles of the corresponding types taken from the exact model, i.e., the continuous time full Markov chain.

Intuitively, this means that as the epidemic spreads on the network, infected nodes are more likely to have neighbours which are themselves infected (either those that infected the node or were infected by it), and at the `front' of the epidemic we would expect to observe a `sea' of susceptible nodes alongside a `front' of links between susceptible and infected nodes that drives the epidemic. Hence, clustering and correlations need to be accounted for and a new $p_{A|S-I}^{c}$ for clustered networks needs to be defined. This has been done in \cite{keeling1999effects} and it relies on a correlation factor, $C_{AB}$, that is able to capture the propensity of nodes of type $A$ and $B$ being neighbouring nodes. This is given by
\begin{equation}
C_{AB}=\frac{[AB]}{n[A]\frac{[B]}{N}},
\label{eq:corr_cab}
\end{equation}
where $A,B\in \{S,I\}$. This effectively compares the expected number of edges of type $[AB]$ to what its value would be if nodes were labelled at random with $[A]$ nodes of type $A$ and $[B]$ nodes of type $B$. If $C_{AB}>1$, then nodes of type $A$ and $B$ are positively correlated, whereas if nodes of type $A$ and $B$ are negatively correlated, $C_{AB}<1$. As expected, $C_{AB}=1$ means that nodes are effectively labelled as type $A$ or $B$ at random. Equation~\eqref{eq:correl_ineq} implies that
\begin{equation}
C_{SI}\le 1, \,\,\, C_{II}\ge 1 \,\,\, \text{and} \,\,\, C_{SS}\ge 1.
\label{eq:correl_ineq_cab}
\end{equation}
We can augment $p_{A|S-I}^{uc}=\frac{[AS]}{n[S]}$ to reflect these observations, leading to $p_{A|S-I}^{c}=\frac{[AS]}{n[S]}C_{AI}$.
%$p_{A|S-I}^{c}=\frac{[AS]}{n[S]}C_{AS}$ was used here previously. Rosie thinks this was a mistake and should instead be C_{AI}, not C_{AS}. Edited on 5th June 2018.
However, before the closure can be expressed, open and closed loops need to be treated separately. In order to do this, we split the closure based on whether the neighbour whose state is to be determined is part of a closed loop of three nodes and thus in direct contact with an infectious node, or not. This leads to
\begin{equation}
p_{A|S-I}^{c}=   \begin{cases}
     p_{A|S-I}^{uc} & \text{with probability}\ (1-\phi),\\
     p_{A|S-I}^{uc}C_{AI}& \text{with probability}\ \phi,
    \end{cases}
\label{eq:clo_clust_simple}
\end{equation}
where $\phi$ is defined in equation \eqref{equation:phi_clustering}. With this in mind, the closure can be derived by averaging equation \eqref{eq:basic_clo} over the unclustered and clustered parts of the network. This leads to
\begin{eqnarray}
[ASI]&=&(1-\phi)(n-1)[SI]p_{A|S-I}^{uc}+\phi(n-1)[SI]p_{A|S-I}^{uc}C_{AI}\\ &=&\frac{(n-1)}{n}\frac{[AS][SI]}{[S]}\left((1-\phi)+\phi\frac{N[AI]}{n[A][I]}\right). \label{eq:clust_simple_closure}
\end{eqnarray}
Framing $p_{A|S-I}^{uc}$ and $p_{A|S-I}^{c}$ more generally and independently of the network type, i.e. simply considering $p_A$, the following statement holds:
\begin{prop} 
Consider a closure of the following form $[ASI]=(n-1)[SI]p_{A}$. If $\sum_{A}p_{A}=1$, where $A$ is taken over all possible states, then
$\sum_{A}[ASI]=(n-1)[SI]$.
\end{prop}
\begin{proof}
$\sum_{A}[ASI]=(n-1)[SI]\sum_{A}p_{A}=(n-1)[SI].$
\end{proof}

%%%%%%%%%%%%%%%%%%%%%%%%%%%%%%%%
\subsubsection{Improved closure}
%%%%%%%%%%%%%%%%%%%%%%%%%%%%%%%%
We note that while $p_{A|S-I}^{uc}$ satisfies the above proposition, $p_{A|S-I}^{c}$ does not. In particular, we find
\begin{align*}
\sum_{A}[ASI]&=\sum_{A}(n-1)[SI]p_{A|S-I}^{uc}=\sum_{A}(n-1)[SI]\frac{[AS]}{n[S]}\\
&=\frac{(n-1)[SI]}{n[S]}\sum_{A}[AS]=\frac{(n-1)[SI]}{n[S]}n[S]=(n-1)[SI].
\end{align*}
However, for the clustered part of the network this is not the case. We find that
\begin{align*}
\sum_{A}[ASI]&=\sum_{A}(n-1)[SI]p_{A|S-I}^{c}=\sum_{A}(n-1)[SI]\frac{[AS]}{n[S]}\frac{N[AI]}{n[A][I]}\\
&=\frac{(n-1)N[SI]}{n^2[S][I]}\sum_{A}\frac{[AS][AI]}{[A]},
\end{align*}
which does not result in the desired $(n-1)[SI]$. This can be corrected in a straightforward way by defining 
\begin{equation}
p_{A|S-I}^{c_{new}}=   \begin{cases}
     p_{A|S-I}^{uc} & \text{with probability}\ (1-\phi),\\
     \frac{p_{A|S-I}^{c}}{\sum_{a}p_{a|S-I}^{c}} & \text{with probability}\ \phi.
    \end{cases}
\label{eq:clo_clust_improved}
\end{equation}
Hence we can now write
\begin{align*}
\sum_{A}[ASI]&=\sum_{A}\left((1-\phi)[ASI]+\phi[ASI]\right)\\&=(1-\phi)(n-1)[SI]\sum_{A}p_{A|S-I}^{uc}+\phi(n-1)[SI]\sum_{A}p_{A|S-I}^{c_{new}}\\
&=(1-\phi)(n-1)[SI]\sum_{A}\frac{[AS]}{n[S]}+\phi(n-1)[SI]\sum_{A}\frac{p_{A|S-I}^{c}}{\sum_{a}p_{a|S-I}^{c}}\\
&=(1-\phi)(n-1)[SI]\frac{1}{n[S]}\sum_{A}[AS]+\phi(n-1)[SI]\\
&=(1-\phi)(n-1)[SI]+\phi(n-1)[SI]\\&=(n-1)[SI],
\end{align*}
as required. It is informative to investigate the relationship between the various probability models that lead to different closures. This is summarised in the following proposition.
\begin{prop}
For closures applied across the clustered part of the network and assuming that the number of nodes in state $R$ is negligible, it follows that 
\begin{equation}
p_{S|S-I}^{c_{new}}=\frac{[SS][I]}{[SS][I]+[II][S]} ,\,\,\,  p_{S|S-I}^{c}=\frac{[SS]}{n[S]}\frac{N[SI]}{n[S][I]},\,\,\, p_{S|S-I}^{uc}=\frac{[SS]}{n[S]},
\end{equation}
and
\begin{equation}
p_{S|S-I}^{c} \le p_{S|S-I}^{uc} \,\,\, \text{and} \,\,\, p_{S|S-I}^{c_{new}} \le p_{S|S-I}^{uc}.
\label{eq:ineq_for_p_s}
\end{equation}
\label{prop:inequ_closures}
\end{prop}
\begin{proof}
All three probabilities follow from their definitions and assuming that $A\in\{S,I\}$.  
%In the presence of clustering this leads to
%\begin{equation}
%p_{S|S-I}^{c_{new}}=\frac{[SS][I]}{[SS][I]+[II][S]} \,\,\, \text{and} \,\,\, p_{I|S-I}^{c_{new}}=\frac{[II][S]}{[SS][I]+[II][S]}.
%\label{eq:clo_clust_improved}
%\end{equation}
Since $S-I$ links are negatively correlated~\eqref{eq:correl_ineq}, it follows that $C_{SI}=\frac{N[SI]}{n[S][I]}\le 1$ and as a result
\begin{equation}
p_{S|S-I}^{c}=\frac{[SS]}{n[S]}C_{SI}\le\frac{[SS]}{n[S]} =p_{S|S-I}^{uc}.
\end{equation}  
%Let us start from the observation that 
%\begin{align*}
%[SS][I]+[II][S] \ge n[S]\frac{[S]}{N}[I]+n[I]\frac{[I]}{N}[S]=\frac{n[S][I]}{N}\left([S]+[I]\right)=n[S][I]\ge N[SI],\\
%\end{align*}
%which follows from equation~\eqref{eq:correl_ineq}.
%Taking into account that susceptible and infected nodes are negatively correlated, i.e. $N[SI]\le n [S][I]$~\eqref{eq:correl_ineq}, and the inequality above leads to
%\begin{align*}
%N^2[SI]^2 \le n[S][I] n[S][I] \le n[S][I] \left( [SS][I]+[II][S]\right) \iff \\ 
%n[S][I] n[S][I] [SS]\le [SS] N[SI] \left( [SS][I]+[II][S]\right) \iff \\
%p_{S|S-I}^{c_{new}}=\frac{[SS][I]}{[SS][I]+[II][S]} \le \frac{[SS]}{n[S]}\frac{N[SI]}{n[S][I]}=p_{S|S-I}^{c},
%\end{align*}
%and hence
%\begin{equation*}
%p_{S|S-I}^{c_{new}} \le  p_{S|S-I}^{c} \le p_{S|S-I}^{uc}.
%\end{equation*}

While $p_{S|S-I}^{c}$ has a natural interpretation (it is a simple discounted variant of the probability from the unclustered network case and takes into account the observation that if the neighbour of a central susceptible node is connected to one of the infected neighbours of the same node then it is less likely that the node in question is susceptible), the interpretation of $p_{S|S-I}^{c_{new}}$ is less obvious. A close inspection reveals that $p_{S|S-I}^{c_{new}}$ can be rewritten as
\begin{equation}
p_{S|S-I}^{c_{new}}=\frac{[SS][I]}{[SS][I]+[II][S]}=\frac{[SS]}{[SS]+[II]\frac{[S]}{[I]}}.
\label{eq:pS_cnew}
\end{equation}
However, combining $[SI] \le n[S]\frac{[I]}{N}$ with $[I] \le \frac{N}{n}\frac{[II]}{[I]}$, as given in equation~\eqref{eq:correl_ineq}, leads to $[SI] \le [II]\frac{[S]}{[I]}$. Finally, using the relation $[SI] \le [II]\frac{[S]}{[I]}$ in equation~\eqref{eq:pS_cnew} yields
\begin{equation} \label{eq:p_c_new_ineq}
p_{S|S-I}^{c_{new}}=\frac{[SS]}{[SS]+[II]\frac{[S]}{[I]}} \le \frac{[SS]}{[SS]+[SI]}=\frac{[SS]}{n[S]}=p_{S|S-I}^{uc}.
\end{equation}
Equation \eqref{eq:p_c_new_ineq} illustrates that as expected $p_{S|S-I}^{c_{new}} \le p_{S|S-I}^{uc}$. Again, this simply shows that for clustered networks and for the setup in Fig.~\ref{fig:ig:clo_around_S_I}, it is less likely to find neighbours who are susceptible compared with the unclustered network case. 
%Compared to the case of the simple closure for clustered networks we note that while 
%$p_{S|S-I}^{c} \le p_{S|S-I}^{uc}$ and $p_{S|S-I}^{c_{new}} \le p_{S|S-I}^{uc}$, the way how the probability is discounted is different. {\color{red} \textbf{Can we find an inequality between these?}}.
\end{proof}

Taking into account the new way of defining $p_{A|S-I}^{c_{new}}$, the improved closure yields
\begin{align}
[ASI]&=(1-\phi)[ASI]+\phi[ASI] \notag\\
&= (1-\phi)(n-1)[SI]\frac{[AS]}{n[S]}+\phi (n-1)[SI]\frac{\frac{[AS]}{n[S]}C_{AI}}{\sum_{a}p_{a|S-I}^{c}}\notag\\
&= (1-\phi)\frac{(n-1)}{n}\frac{[AS][SI]}{[S]}+\phi (n-1)[SI]\frac{\frac{[AS]}{n[S]}\frac{N[AI]}{n[A][I]}}{\sum_{a}\frac{[aS]}{n[S]}\frac{N[aI]}{n[a][I]}}\notag\\
&=(1-\phi)\frac{(n-1)}{n}\frac{[AS][SI]}{[S]}+\phi (n-1)\frac{[AS][SI][IA]}{[A]\sum_{a}\frac{[aS][aI]}{[a]}}\notag\\
&=(n-1)\left((1-\phi)\frac{[AS][SI]}{n[S]}+\phi \frac{[AS][SI][IA]}{[A]\sum_{a}[aS][aI]/[a]}\right).
\label{eq:improved_closure}
\end{align}

We finally note that the closures rely heavily on the assumption of how the states of the neighbours are distributed, and the assumption of independent and identically Bernoulli-distributed variables is a strong one. 
%Closures also effectively assume that neighbours of a node are either in state $S$ or $I$, which of course ignores the real possibility of having neighbours who have recovered from infection or are removed. 
For clustered networks in particular, we have illustrated different ways of incorporating correlations induced by closed cycles of length three. Despite these seemingly strong assumptions, it is known that the pairwise model for unclustered networks is equivalent to the edge-based compartmental equivalent on configuration networks \cite{miller2014epidemic,kiss2017mathematics} and the latter has been shown to be the limiting system of the stochastic network epidemic model \cite{decreusefond2012large,janson2014law}. For clustered networks we are not aware of such results.

%%%%%%%%%%%%%%%%%%%%%%%%%%%%%%%%
\section{Results for the pairwise model with the simple closure} \label{section:Results_simp_clo}
%%%%%%%%%%%%%%%%%%%%%%%%%%%%%%%%

%%%%%%%%%%%%%%%%%%%%%%%%%%%%%%%%
\subsection{Background} \label{subsec:background}
%%%%%%%%%%%%%%%%%%%%%%%%%%%%%%%%

Using the simple closure for clustered networks \eqref{eq:clust_simple_closure}, and writing $\xi=\frac{(n-1)}{n}$, we obtain the following closed pairwise model equations describing an SIR epidemic on a clustered regular network of $N$ individuals with degree $n$: \begin{align} \label{eq:pw_clust_simple_closure_1} \dot{[S]}&=-\tau[SI], \\ \label{eq:pw_clust_simple_closure_2} \dot{[I]}&=\tau[SI]-\gamma[I], \\ \label{eq:pw_clust_simple_closure_3} \dot{[SI]}&=-(\tau+\gamma)[SI]+\tau\xi\frac{[SS][SI]}{[S]}\left((1-\phi)+\phi\frac{N[SI]}{n[S][I]}\right)-\tau\xi\frac{[SI]^{2}}{[S]}\left((1-\phi)+\phi\frac{N[II]}{n[I]^{2}}\right), \\ \dot{[SS]}&=-2\tau\xi\frac{[SS][SI]}{[S]}\left((1-\phi)+\phi\frac{N[SI]}{n[S][I]}\right), \\ \dot{[II]}&=2\tau[SI]-2\gamma[II]+2\tau\xi\frac{[SI]^{2}}{[S]}\left((1-\phi)+\phi\frac{N[II]}{n[I]^{2}}\right). \label{eq:pw_clust_simple_closure_end} \end{align}

For model equations \eqref{eq:pw_clust_simple_closure_1}-\eqref{eq:pw_clust_simple_closure_end}, in \cite{keeling1999effects} the basic reproductive ratio ($R_{0}$) is considered. Starting from the evolution equations of infectious individuals leads to \begin{align*} \dot{[I]}&=\tau[SI]-\gamma[I] \\ &=\left(\frac{\beta[S]}{N}C_{SI}-\gamma\right)[I], \end{align*} where $C_{SI}$ is defined in equation \eqref{eq:corr_cab}. Taking into account that $\tau n=\beta$ and that initially $[S]\simeq N$, in~\cite{keeling1999effects} it is claimed that $R_{0}=C_{SI}\beta/\gamma$. It is important to note that this $R_0$ is not the classical $R_0$ in the sense of being the expected number of new infections produced by a typical infectious individual when introduced in a fully susceptible population. Rather it can be thought of as a growth-rate-based threshold, and has the same properties as the classical $R_0$ when both are exactly one. In what follows, we will simply refer to it as $R$~\cite{eames2008modelling,kiss2012modelling}.

Thus in order to determine $R$ explicitly, the authors in \cite{keeling1999effects} consider the early behaviour of $C_{SI}$ and find that this variable is given by the ordinary differential equation (ODE) \begin{equation} \dot{C_{SI}}=-\tau\left(C_{SI}+C_{SI}^{2}-n\xi(C_{SI}-C_{SI}^{2})(1-\phi)+n\xi C_{SI}^{2}\phi\frac{[I]C_{II}}{N}\right). \label{eq:differential_for_C_SI} \end{equation} The ODE above, \eqref{eq:differential_for_C_SI}, however depends on the behaviour of $[I]C_{II}/N$ and in~\cite{keeling1999effects} it was found that \begin{equation} \frac{[I]C_{II}}{N}\longrightarrow\frac{2\tau C_{SI}}{\gamma+\beta C_{SI}-2\xi\beta C_{SI}^{2}\phi}. \label{eq:cii_no_clust} \end{equation} Considering the quasi-equilibrium of $C_{SI}$, referred to as $C_{SI}^{*}$, in equation~\eqref{eq:differential_for_C_SI} together with the expression for $[I]C_{II}/N$ in equation~\eqref{eq:cii_no_clust}, one finds that $C_{SI}^{*}$ is given by \begin{equation} 1+C_{SI}^{*}-n\xi(1-C_{SI}^{*})(1-\phi)+\frac{2\tau n\xi\phi {C_{SI}^{*}}^{2}}{\gamma+\beta C_{SI}^{*}-2\xi\beta {C_{SI}^{*}}^{2}\phi}=0. \label{eq:solve_C_SI_for_thresh_clust} \end{equation}
Hence, $R$ can be calculated as $C_{SI}^{*}\beta/\gamma$, at least numerically. Variables such as $C_{SI}$ and $C_{II}$ describe the correlations between the states of neighbouring nodes on the network as the epidemic unfolds and these have been studied numerically in~\cite{keeling1999effects}. 

For model equations \eqref{eq:pw_clust_simple_closure_1}-\eqref{eq:pw_clust_simple_closure_end} and when there is no clustering present in the network structure (thus $\phi=0$), a further simplification of equation \eqref{eq:solve_C_SI_for_thresh_clust} can be achieved~\cite{keeling1999effects}. To determine $R=C_{SI}^{*}\beta/\gamma$ in this case, simply solve \begin{equation} 1+C_{SI}^{*}-n\xi(1-C_{SI}^{*})=0 \end{equation} to find $C_{SI}^{*}=\frac{n-2}{n}$ and thus $R=\frac{(n-2)\tau}{\gamma}$.

Unfortunately when $\phi\ne 0$, according to our knowledge, the quasi-equilibrium values can only be determined numerically via equation \eqref{eq:solve_C_SI_for_thresh_clust}. In what follows, we show that by working with two new variables, $\alpha=[SI]/[I]$ and $\delta=[II]/[I]$, which are still closely linked to the correlations formed during the spreading process, it is possible to obtain the epidemic threshold as the solution of a cubic equation and, more importantly, we show that this can be obtained as asymptotic expansion in powers of $\phi$.

%%%%%%%%%%%%%%%%%%%%%%%%%%%%%%%%
\subsection{Epidemic threshold} \label{sec:ep_threshold}
%%%%%%%%%%%%%%%%%%%%%%%%%%%%%%%%

Consider the initial phase of an infection invading an entirely susceptible population in the pairwise model, described by equations \eqref{eq:pw_clust_simple_closure_1}-\eqref{eq:pw_clust_simple_closure_end}. We find that 
\begin{equation} \dot{[I]}=\tau[SI]-\gamma[I]=\gamma[I]\left(\frac{\tau[SI]}{\gamma[I]}-1\right). 
\label{eq:new_threshold}
\end{equation} We know the quantity $\gamma[I]$ remains non-negative regardless of time in the epidemic process, and we choose to consider the threshold in terms of $\frac{[SI]}{[I]}$. This leads to $R=\frac{\tau[SI]}{\gamma[I]}$. When $R>1$ an epidemic will occur, and when $R<1$ the epidemic will die out. Although we know the values of $\tau$ and $\gamma$, to determine if an epidemic will occur \emph{a priori}, we require further knowledge about the quantity $\frac{[SI]}{[I]}$ at some initial time close to $t=0$. While this is similar to the approach taken in~\cite{keeling1999effects}, we focus on variables such as $\frac{[SI]}{[I]}$ and $\frac{[II]}{[I]}$, and we motivate our choice below. The problem of finding the epidemic threshold can be dealt with in at least two more different but equivalent ways. First, one can carry out a simple linear stability analysis of the disease-free steady state and this is shown in Appendices~\ref{sec:lin_stab_anal_simp_clo} and \ref{sec:lin_stab_anal_comp_imp_clo}. Second, the threshold can also be computed as the largest eigenvalue of the next generation matrix, see section~\ref{sec:discussion}. However, in both cases, the variables $[SI]/[I]$ and $[II]/[I]$ turn out to play a key role and their values for small times need to determined.

% following para deleted by Istvan:

% The typical assumption that both $[SI]$ and $[I]$ are equal to zero at time zero leads to an ill-defined quantity (dividing zero by zero). To circumvent this, we must compute a differential equation for $\frac{[SI]}{[I]}$, assume that it is a fast variable which converges to its equilibrium value faster than the epidemic process, and solve to find a steady state value for $\frac{[SI]}{[I]}$, which can be plugged into $T_{r}$ to determine if an epidemic is likely to occur. In section X we compute a differential equation for $\frac{[SI]}{[I]}$.

%%%%%%%%%%%%%%%%%%%%%%%%%%%%%%%%
\subsection{Fast variables with the simple closure} \label{sec:correlation_struct_and_fast_vars}
%%%%%%%%%%%%%%%%%%%%%%%%%%%%%%%%

To circumvent the problem of the ill-defined variables above, we exploit the fact that $\alpha=\frac{[SI]}{[I]}$ and $\delta=\frac{[II]}{[I]}$ are fast variables when compared to the time course of the epidemic. Fig.~\ref{fig:Fast_var} shows clearly that $\alpha$ and $\delta$ are fast compared to the epidemic process and that they quickly converge to a quasi-equilibrium. Hence, at early times $\alpha$ and $\delta$ attain their quasi-equilibrium values, and these are the values that can be used to compute the epidemic threshold. 
\begin{figure}
\centering
\subfloat[]{\includegraphics[width=0.40\textwidth]{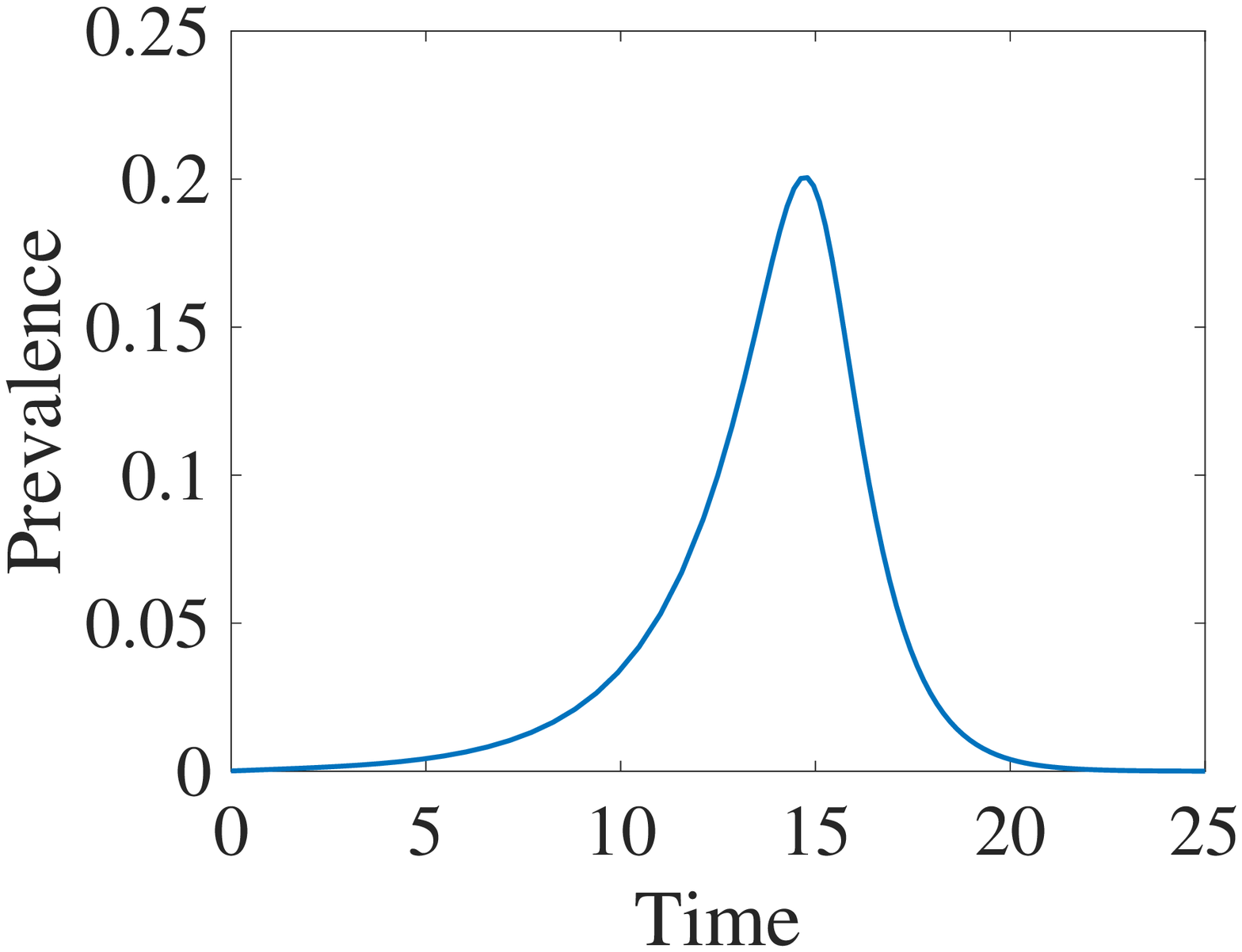}
 \label{fig:Fast_var_I_simp_clo}
    }
\hspace{1.0cm}
\subfloat[]{\includegraphics[width=0.40\textwidth]{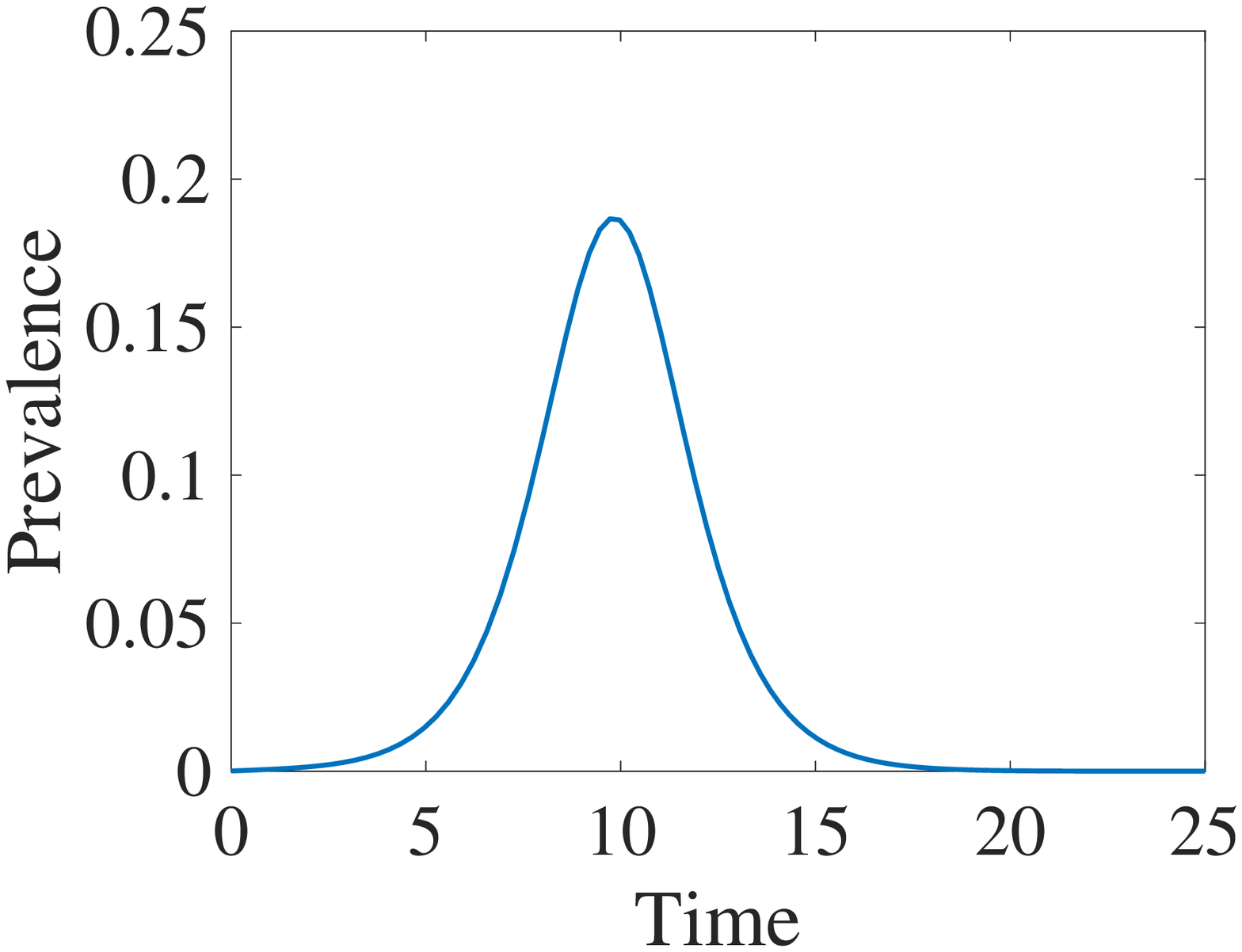}
 \label{fig:Fast_var_I_TH_full}
    }\hfill
\subfloat[]{\includegraphics[width=0.40\textwidth]{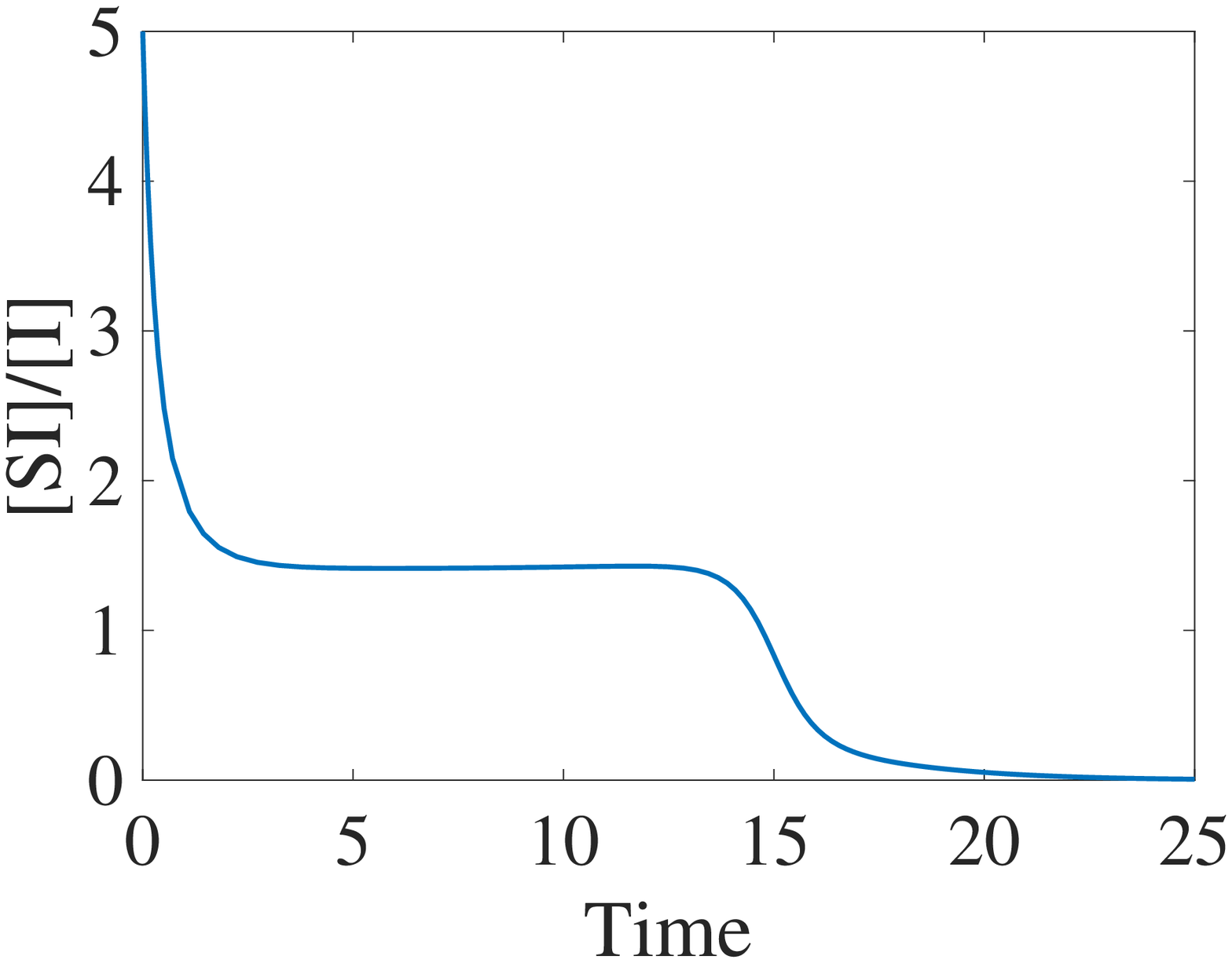}
\label{fig:Fast_var_SI_by_I_simp_clo}
        }
\hspace{1.0cm}
\subfloat[]{\includegraphics[width=0.40\textwidth]{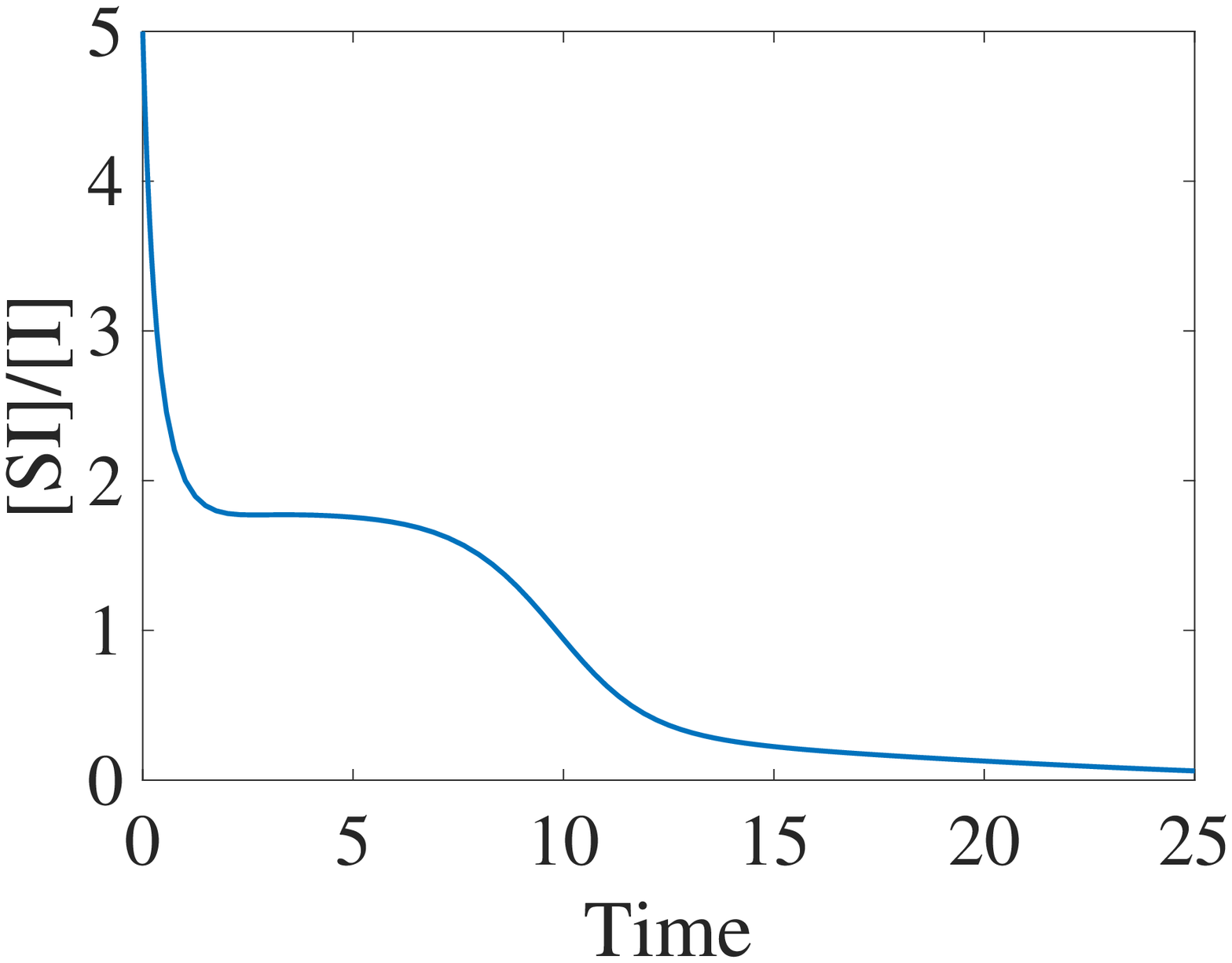}
\label{fig:Fast_var_SI_by_I_TH_full}
        }\hfill
\subfloat[]     {\includegraphics[width=0.40\textwidth]{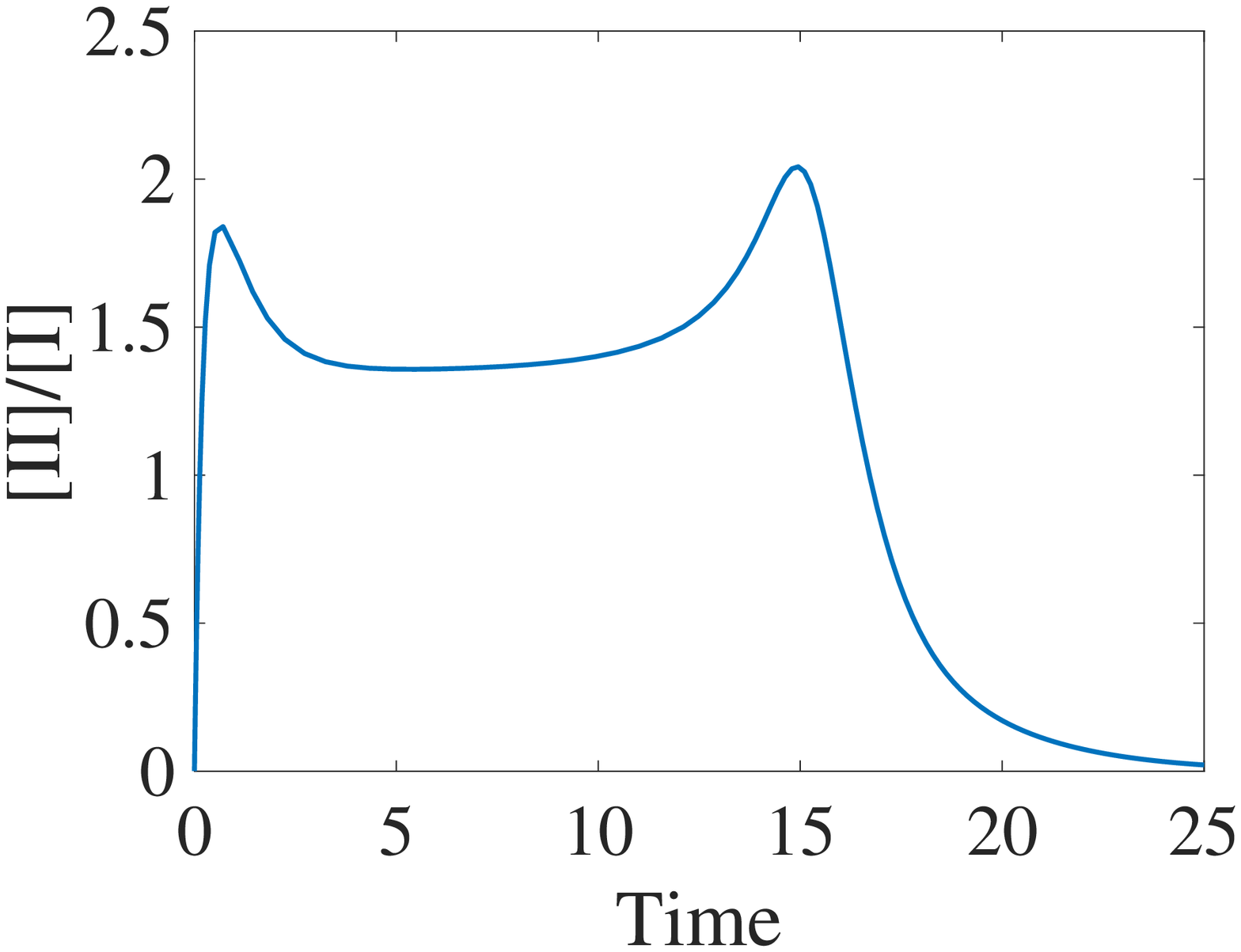}
\label{fig:Fast_var_II_by_I_simp_clo}
        }
        \hspace{1.0cm}
        \subfloat[]     {\includegraphics[width=0.40\textwidth]{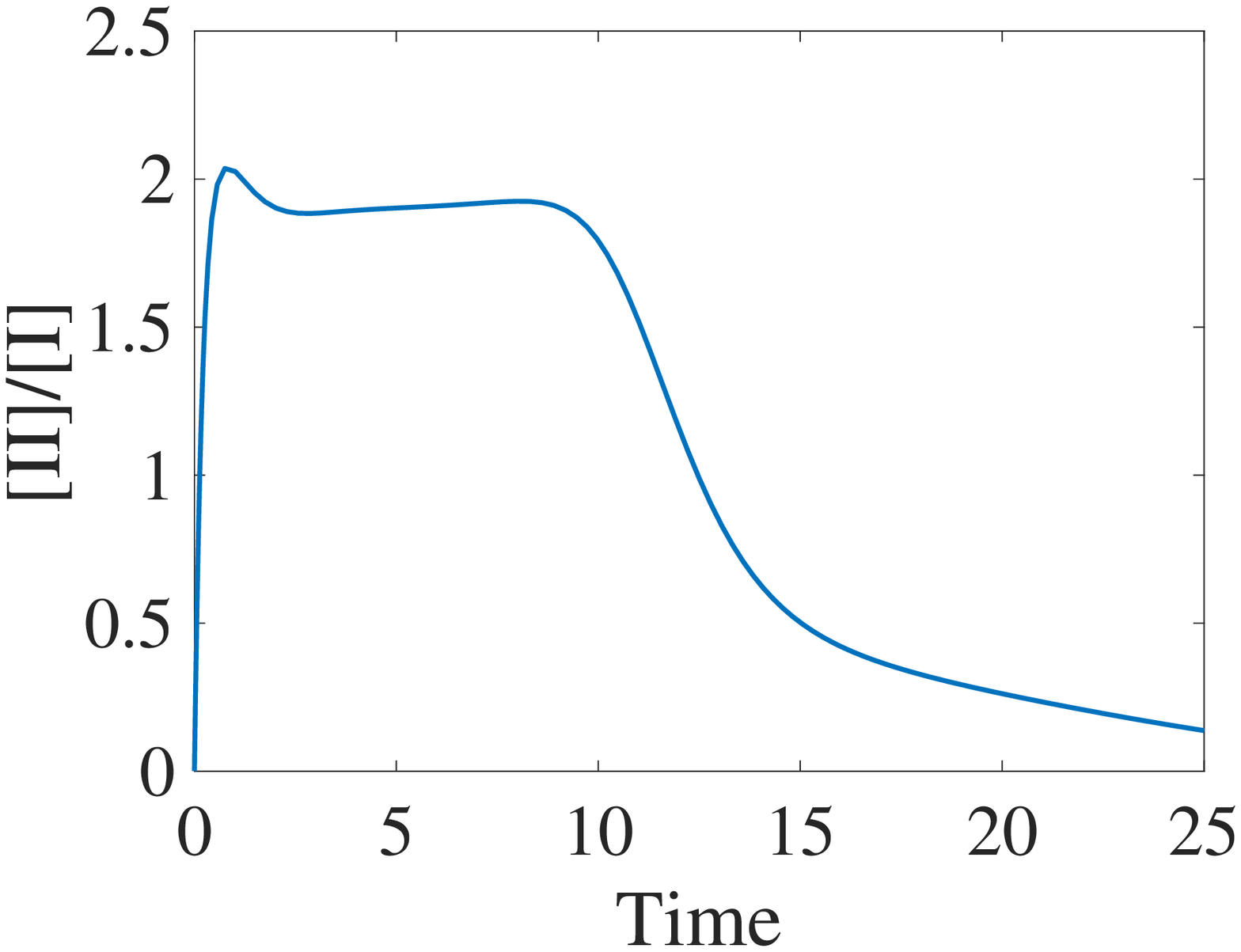}
\label{fig:Fast_var_II_by_I_TH_full}
        }
 \caption{Illustration of the dynamics of prevalence, $[I]/N$, over time ((a)-(b)), compared to that of $\alpha=\frac{[SI]}{[I]}$ ((c)-(d)) and $\delta=\frac{[II]}{[I]}$ ((e)-(f)) for the pairwise model with the simple (left column) and the improved (right column) closures. Parameter values are $N=10000$, $n=5$, $\phi=0.5$ and $\tau=\gamma=1$.}
  \label{fig:Fast_var}
\end{figure}

We continue by deriving differential equations for the variables $\alpha=\frac{[SI]}{[I]}$ and $\delta=\frac{[II]}{[I]}$. Differentiating $\alpha$ and $\delta$ and using equations~\eqref{eq:pw_clust_simple_closure_1}-\eqref{eq:pw_clust_simple_closure_end} leads to
\begin{eqnarray} \label{eq:diff_for_alpha} \frac{d\alpha}{dt}&=&-\tau\alpha+\tau\xi n(1-\phi)\alpha+\tau\xi\phi\alpha^{2}-\tau\xi\frac{1}{n}\phi\alpha^{2}\delta-\tau\alpha^{2}, \\ \frac{d\delta}{dt}&=&2\tau\alpha-\gamma\delta+2\tau\xi\frac{1}{n}\phi\alpha^{2}\delta-\tau\alpha\delta. \label{eq:diff_for_delta} \end{eqnarray}
The detailed derivation for equations \eqref{eq:diff_for_alpha} and \eqref{eq:diff_for_delta} can be found in Appendix~\ref{sec:correlation_struct_and_fast_vars_appendix}.

%%%%%%%%%%%%%%%%%%%%%%%%%%%%%%%%
\subsubsection{Fast variables without clustering}
%%%%%%%%%%%%%%%%%%%%%%%%%%%%%%%%

When clustering is negligible and hence $\phi=0$, we find that \begin{eqnarray} \label{eq:diff_for_alpha_no_clust} \frac{d\alpha}{dt}&=&-\tau\alpha+\tau\xi n\alpha-\tau\alpha^{2}, \\ \frac{d\delta}{dt}&=&2\tau\alpha-\gamma\delta-\tau\alpha\delta, \label{eq:diff_for_delta_no_clust} \end{eqnarray} where $\xi=\frac{(n-1)}{n}$. The steady states of the system \eqref{eq:diff_for_alpha_no_clust}-\eqref{eq:diff_for_delta_no_clust} are given by $(\alpha_{1}^{*},\delta_{1}^{*})=(0,0)$ and  $(\alpha_{2}^{*},\delta_{2}^{*})=\left((n-2),\frac{2\tau(n-2)}{\gamma+\tau(n-2)}\right)$. Based on equation~\eqref{eq:new_threshold}, it follows that $R=\frac{\tau \alpha_{2}^{*}}{\gamma}=\frac{\tau (n-2)}{\gamma}.$
%$(\alpha^{*}_{2},\delta^{*}_{2})=(n-2,)$determine steady states of equations \eqref{eq:diff_for_alpha_no_clust} and \eqref{eq:diff_for_delta_no_clust} analytically, set $\frac{d\alpha}{dt}=\frac{d\delta}{dt}=0$ and solve to find steady state values $\alpha^{*}$ and $\delta^{*}$. We begin by setting $\frac{d\alpha}{dt}=0$ to obtain \begin{eqnarray*} -\tau\alpha+\tau\xi n\alpha-\tau\alpha^{2}&=&0 \\ \tau\alpha\left(-1+\xi n-\alpha\right)&=&0 \\ \tau\alpha\left(-1+(n-1)-\alpha\right)&=&0 \\ \tau\alpha\left(n-2-\alpha\right)&=&0, \end{eqnarray*} where $\xi=\frac{(n-1)}{n}$, implying that $\alpha^{*}_{1}=0$ (trivial steady state) and $\alpha^{*}_{2}=n-2$. Rearranging $\frac{d\delta}{dt}=0$ gives $\delta=\frac{2\tau\alpha}{\gamma+\tau\alpha}$. Thus $\alpha^{*}_{1}=0 \implies \delta^{*}_{1}=0$ and $\alpha^{*}_{2}=n-2 \implies \delta^{*}_{2}=\frac{2\tau(n-2)}{\gamma+\tau(n-2)}$. In summary, when clustering is negligible and $\phi=0$, we have a trivial steady state at the point $(\alpha^{*}_{1},\delta^{*}_{1})=(0,0)$ and a non-trivial steady state occurring at $(\alpha^{*}_{2},\delta^{*}_{2})=(n-2,\frac{2\tau(n-2)}{\gamma+\tau(n-2)})$. 

%%%%%%%%%%%%%%%%%%%%%%%%%%%%%%%%
\subsubsection{Fast variables with clustering}
%%%%%%%%%%%%%%%%%%%%%%%%%%%%%%%%

When clustering is present in the network, the differential equations for $\alpha$ and $\delta$ are more complex and thus steady states are harder to compute. Firstly, we set equation \eqref{eq:diff_for_alpha} equal to zero and rearrange to isolate $\delta$, finding \begin{equation} \delta=\frac{-1+\xi n(1-\phi)+\xi\phi\alpha-\alpha}{\xi\frac{1}{n}\phi\alpha}. \label{eq:expression_1_for_delta*} \end{equation} 
Plugging equation \eqref{eq:expression_1_for_delta*} into equation~\eqref{eq:diff_for_delta} leads to the following cubic equation in $\alpha$:
\begin{align} 
&(2\tau\xi\phi(1-\xi\phi))\alpha^{3}+(\tau\xi n\phi-2\tau\xi^{2}n\phi(1-\phi)-\tau n)\alpha^{2}\notag\\
&+(-n(\tau+\gamma)+\tau\xi n^{2}(1-\phi)+\gamma\xi n\phi)\alpha+(\gamma\xi n^{2}(1-\phi)-\gamma n)=0. \label{eq:cubic_eqn_for_alpha} \end{align} 
The solution of the cubic equation \eqref{eq:cubic_eqn_for_alpha} provides the steady state(s) of system~\eqref{eq:diff_for_alpha}-\eqref{eq:diff_for_delta}, and allows the computation of the threshold via the formula $R^{c}=\frac{\tau \alpha^{*}}{\gamma}$. We note that the steady state in $\alpha$ has to be biologically plausible. $\alpha=\frac{[SI]}{[I]}$ restricts the steady state to be positive and to be less than $n$, since the average number of susceptible neighbours averaged over all infected nodes needs to be less than the average degree.
%the form $(\alpha^{*},\delta^{*})$, solve equation \eqref{eq:cubic_eqn_for_alpha} to find $\alpha^{*}$ and use equation \eqref{eq:expression_1_for_delta*} or equation \eqref{eq:expression_2_for_delta*} to compute $\delta^{*}$.

%%%%%%%%%%%%%%%%%%%%%%%%%%%%%%%%
\subsection{Asymptotic expansion of the epidemic threshold} \label{sec:asymptotic_expansion_1}
%%%%%%%%%%%%%%%%%%%%%%%%%%%%%%%%

The case of $\phi \ne 0$ can be regarded as a perturbation of the case with no clustering and we thus set out to find $\alpha$ using a perturbation method. More precisely, we seek to find the roots of the cubic polynomial, given in equation~\eqref{eq:cubic_eqn_for_alpha}, in terms of an asymptotic expansion in powers of $\phi$, that is \begin{equation} \alpha=\alpha_{0}+\phi\alpha_{1}+\phi^{2}\alpha_{2}+\cdots.\label{eq:substitution_for_alpha} \end{equation} Plugging \eqref{eq:substitution_for_alpha} into equation~\eqref{eq:cubic_eqn_for_alpha} leads to \begin{dmath} \label{eq:asymptotic_expansion} 2\tau\xi\phi(1-\xi\phi)(\alpha_{0}+\phi\alpha_{1}+\phi^{2}\alpha_{2}+\cdots)^{3}+(\tau\xi n\phi-2\tau\xi^{2}n\phi(1-\phi)-\tau n)(\alpha_{0}+\phi\alpha_{1}+\phi^{2}\alpha_{2}+\cdots)^{2}+(-n(\tau+\gamma)+\tau\xi n^{2}(1-\phi)+\gamma\xi n\phi)(\alpha_{0}+\phi\alpha_{1}+\phi^{2}\alpha_{2}+\cdots)+(\gamma\xi n^{2}(1-\phi)-\gamma n)=0. \end{dmath} Collecting terms of order $\phi^{0}$ in \eqref{eq:asymptotic_expansion} and after some algebra we find that $\alpha_{0}$ satisfies the equation below: \begin{equation}% -\tau n\alpha_{0}^{2}-n(\tau+\gamma)\alpha_{0}+\tau\xi n^{2}\alpha_{0}+\gamma\xi n^{2}-\gamma n&=&0 \\ -\tau n\alpha_{0}^{2}+(-n(\tau+\gamma)+\tau\xi n^{2})\alpha_{0}+\gamma\xi n^{2}-\gamma n&=&0 \\ 
n(\alpha_0-(n-2))(\tau \alpha_0-\gamma )=0.
%-\tau n\alpha_{0}^{2}+(-n(\tau+\gamma)+\tau(n-%1)n)\alpha_{0}+\gamma\xi n^{2}-\gamma n=0. 
\label{eq:quadratic_for_alpha_0} \end{equation} 
Hence, $\alpha_{0}=(n-2)$. As expected, this corresponds to the unclustered case. Collecting terms of order $\phi$ in \eqref{eq:asymptotic_expansion}, we find a polynomial in terms of $\alpha_{0}$ and $\alpha_{1}$: \begin{dmath} 2\tau\xi\alpha_{0}^{3}+(\tau\xi n-2\tau\xi^{2}n)\alpha_{0}^{2}+(\gamma\xi n-\tau\xi n^{2})\alpha_{0}-2\tau n\alpha_{0}\alpha_{1}+(\tau\xi n^{2}-n(\tau+\gamma))\alpha_{1}-\gamma\xi n^{2}=0 \label{eq:polynomial_for_alpha_0_1} \end{dmath}. Equation \eqref{eq:polynomial_for_alpha_0_1} leads to 
%Using that Since we know $\alpha_{0}=(n-2)$, we can use \eqref{eq:polynomial_for_alpha_0_1} to determine that 
\begin{equation*} \alpha_{1}=\frac{\gamma\xi n^{2}-2\tau\xi\alpha_{0}^{3}+(2\tau\xi^{2}n-\tau\xi n)\alpha_{0}^{2}+(\tau\xi n^{2}-\gamma\xi n)\alpha_{0}}{\tau\xi n^{2}-n(\tau+\gamma)-2\tau n\alpha_{0}}, \end{equation*} which after substituting $\alpha_{0}=(n-2)$ and $\xi=\frac{(n-1)}{n}$ yields
\begin{equation} \label{eq:alpha_1} \alpha_{1}=\frac{-2(n-1)}{n^{2}}\left(\frac{2\tau(n-1)(n-2)+\gamma n}{\tau(n-2)+\gamma}\right). \end{equation} 
To summarise, we have determined the first two coefficients $\alpha_{0}$ and $\alpha_{1}$ of the asymptotic expansion \eqref{eq:substitution_for_alpha} which solves the cubic equation \eqref{eq:cubic_eqn_for_alpha}. Hence, the true solution is approximated by the following expansion: \begin{equation} \alpha=(n-2)-\phi\frac{2(n-1)}{n^{2}}\left(\frac{2\tau(n-1)(n-2)+\gamma n}{\tau(n-2)+\gamma}\right)+\cdots. \label{eq:approximation_for_alpha}\end{equation}
We make several remarks. First, the epidemic threshold will be given by $R^{c}=\tau \alpha/\gamma$. Second, the coefficient of the first order correction of $\alpha$ can be rearranged in terms of $R=\frac{\tau (n-2)}{\gamma}$, the threshold for the case of unclustered networks, leading to
\begin{equation}
R^{c}=R-\phi a\frac{\tau}{\gamma}\left(\frac{aR+1}{R+1}\right),
\end{equation}
where $a=2(n-1)/n$.

%{\color{red} Can we say anything about, is it worth saying, about what happens when $n$, $\tau$, $\gamma$ or their combinations goes to $\infty$?}
Finally, it is clear that due to the first order correction being negative, we have that 
\begin{equation}
R^{c}=R-\phi a\frac{\tau}{\gamma}\left(\frac{aR+1}{R+1}\right) \le R=\frac{\tau (n-2)}{\gamma}.
\end{equation} 

The goodness of the estimate for $\alpha$ \eqref{eq:approximation_for_alpha} is tested by comparing it to the numerical solution of the cubic equation~\eqref{eq:cubic_eqn_for_alpha}. This is done in Fig.~\ref{fig:simp_clo_thres} for five different values of the clustering coefficient. The asymptotic approximation performs well and only breaks down for values of clustering larger than $\simeq 0.3$. From the same figure it is clear that higher values of clustering push the critical $R^{c}=1$ curve to higher values of $\tau$ and $n$. Hence, in the presence of clustering a viable epidemic requires either a denser network or a higher transmission rate, noting that the transmission rate and the recovery rate $\gamma$ are not strictly independent. 
%\vspace{5cm}
%{\color{red} Rosie, please type up a summary of the asymptotic expansion. See my scanned notes on Dropbox.}

%%%%%%%%%%%%%%%%%%%%%%%%%%%%%%%%
 \subsection{Numerical examples}
%%%%%%%%%%%%%%%%%%%%%%%%%%%%%%%%
In the previous section we have demonstrated that for the pairwise model with the simplest closure for clustered networks, the determination of the epidemic threshold involves the solution of a cubic equation. While this can be done numerically, we presented an asymptotic approximation of the solution in powers terms of powers of the clustering coefficient $\phi$. In Fig.~\ref{fig:simp_clo_thres} we present a systematic test of the newly determined threshold by comparing the threshold based on the numerical solution of the cubic equation \eqref{eq:cubic_eqn_for_alpha} (continuous line in the ($\tau,n,0$) plane), the asymptotic approximation of the solution to the cubic equation \eqref{eq:approximation_for_alpha} (dashed line and markers - $\circ$) and the numerical solution of the full ODE system corresponding to the closed pairwise model \eqref{eq:pw_clust_simple_closure_1}-\eqref{eq:pw_clust_simple_closure_end}.

The agreement between the explicit numerical solution of the closed pairwise system and threshold based on the numerical solution of the cubic equation is excellent for all clustering values and other parameter combinations. Moreover, the agreement of these results with the threshold based on the asymptotic approximation is also excellent and remains valid for values of $0\le \phi \le 0.3$. Our numerical tests confirm that our analytical results are correct. The initial conditions for the closed pairwise systems were set in the following way: $[I](0)=I_0=1$, $[S](0)=N-I_0=S_0$, $[SI](0)=nI_0\frac{S_0}{N}$, $[SS](0)=nS_0\frac{S_0}{N}$ and $[II](0)=nI_0\frac{I_0}{N}$. The ODEs were run for a sufficiently long time ($T_{max}=1000$) to ensure that the epidemic died out. It is worth noting that the correct numerical solution of the cubic equation can be chosen by keeping in mind that $0\le \alpha=\frac{[SI]}{[I]}\le n$.

\begin{figure}
\centering
%\subfloat[]
{\includegraphics[scale=0.25]{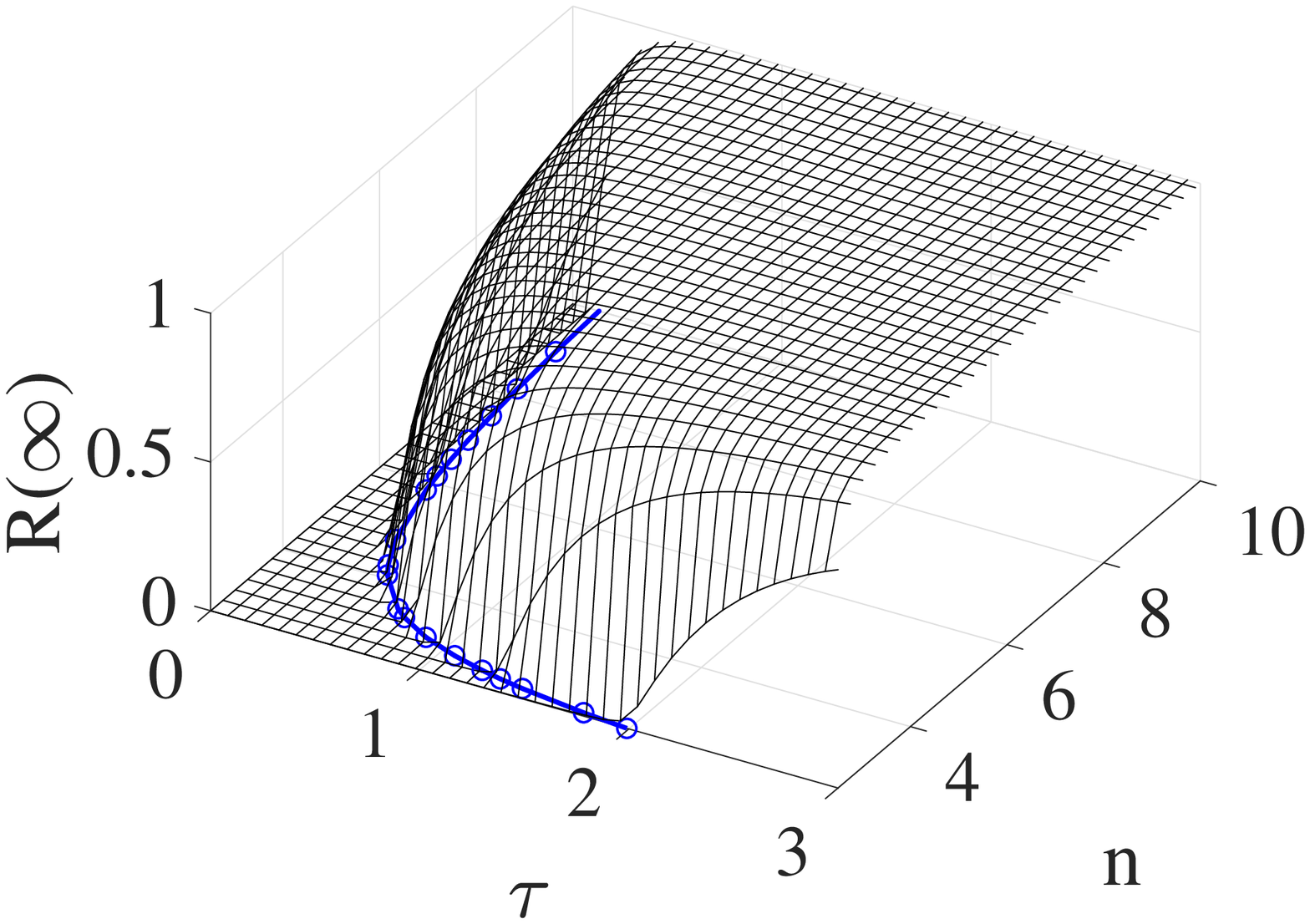}
% \label{fig:simp_clo_thres_c0}
    }
\hspace{1.0cm}
%\subfloat[]
{\includegraphics[scale=0.25]{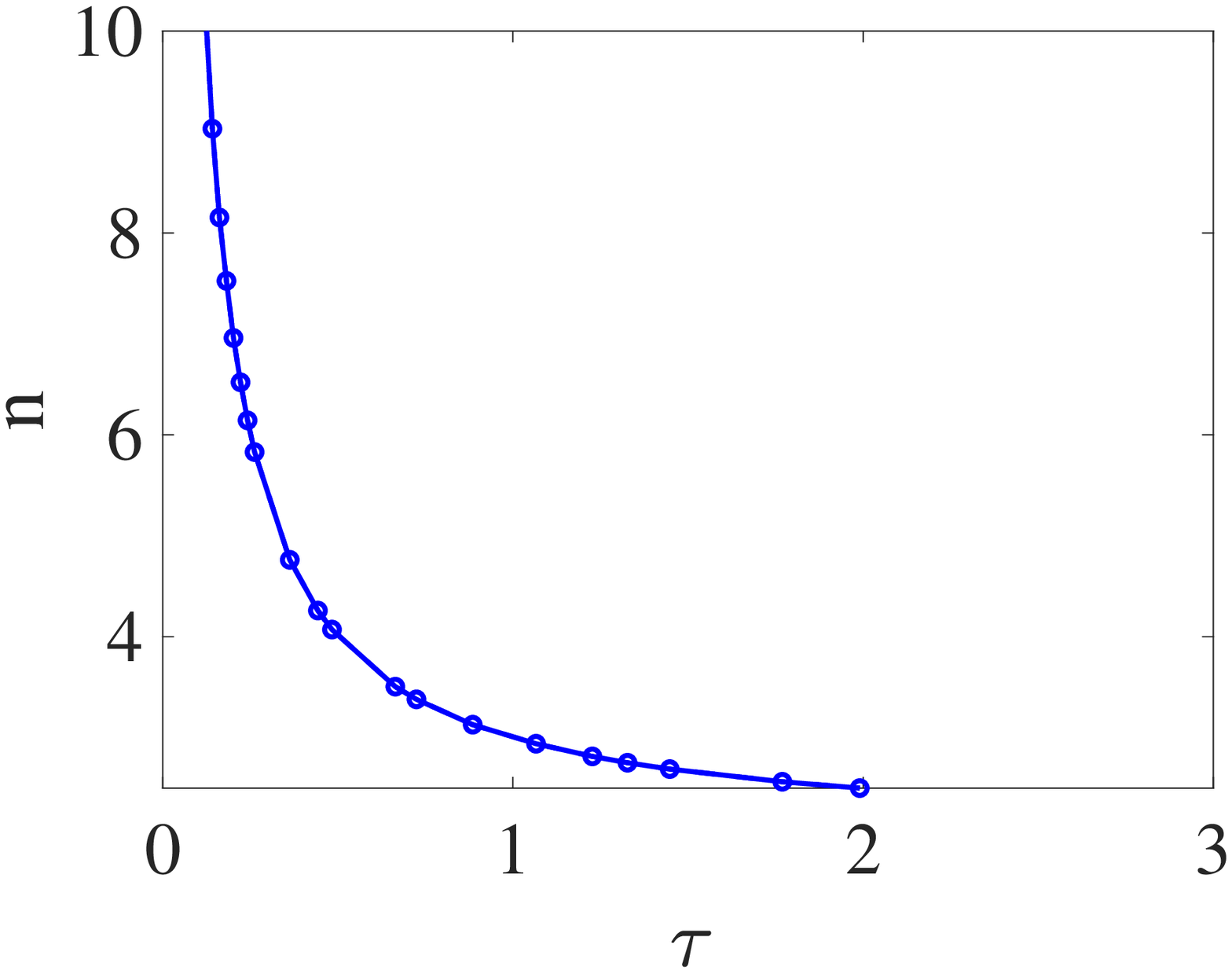}
% \label{fig:simp_clo_thres_c0p1}
    }%\hfill
%\subfloat[]
{\includegraphics[scale=0.25]{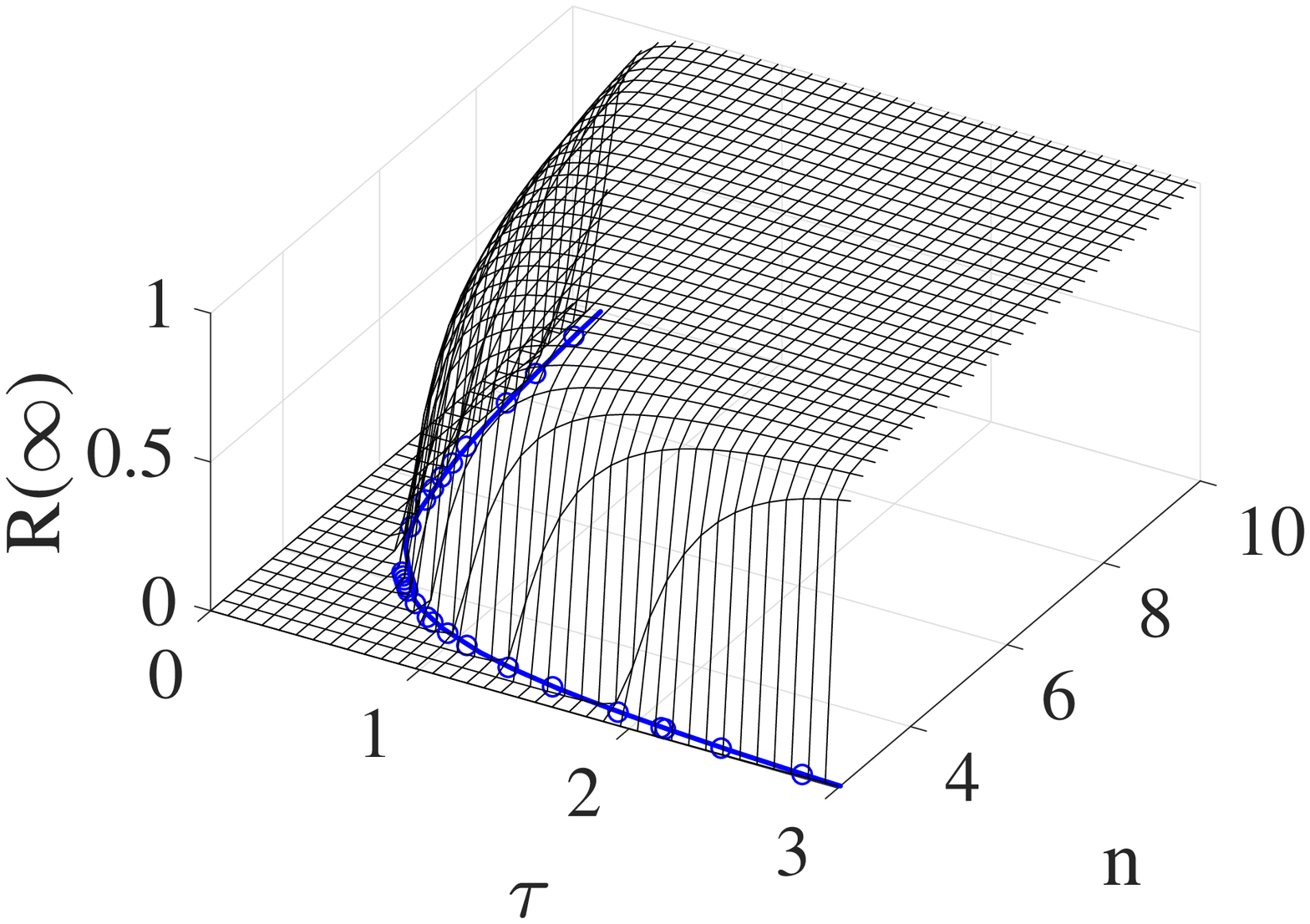}
%\label{fig:simp_clo_thres_c0p2}
        }
\hspace{1.0cm}
%\subfloat[]
{\includegraphics[scale=0.25]{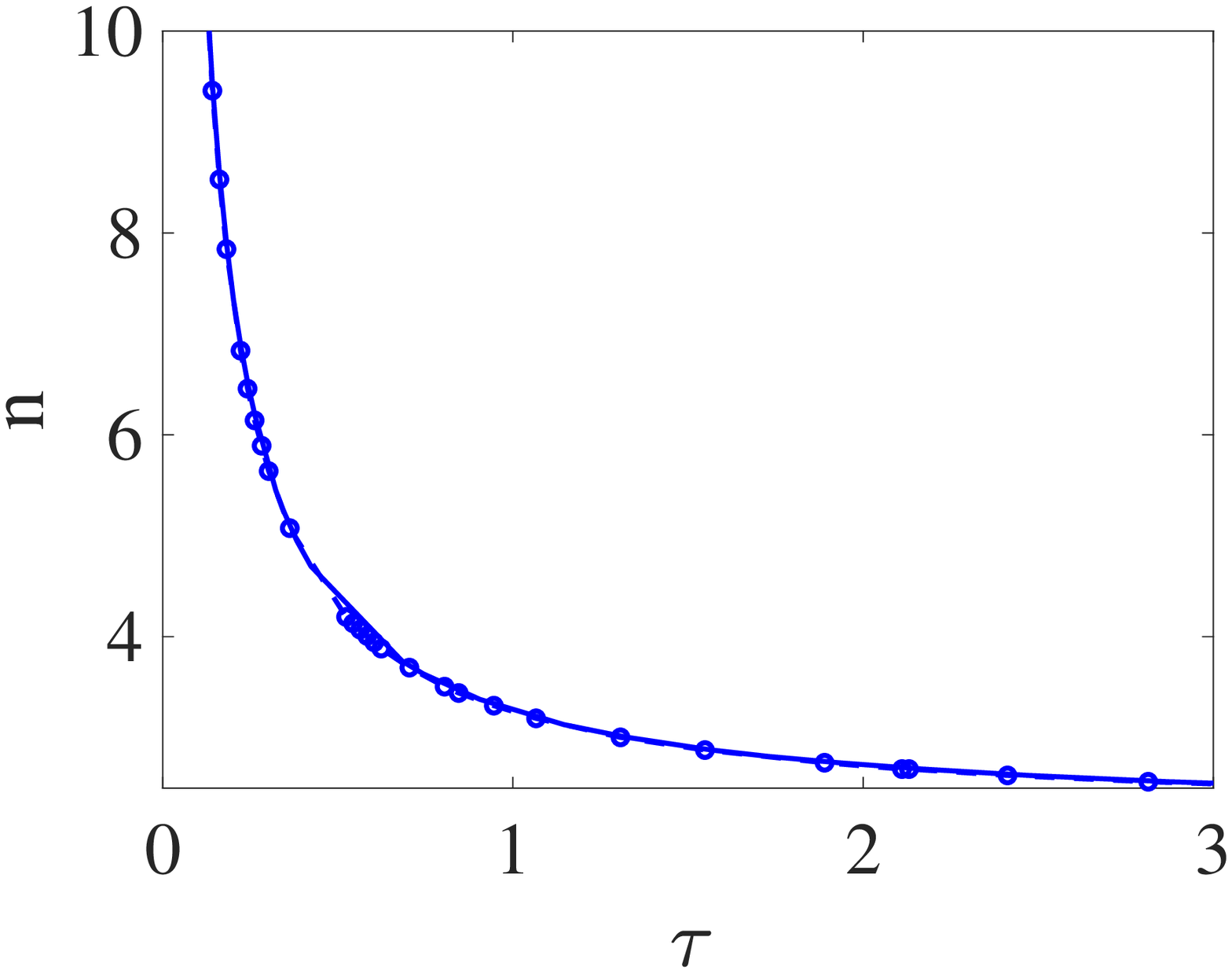}
%\label{fig:simp_clo_thres_c0p3.eps}
  }%\hfill
%\subfloat[]     
{\includegraphics[scale=0.25]{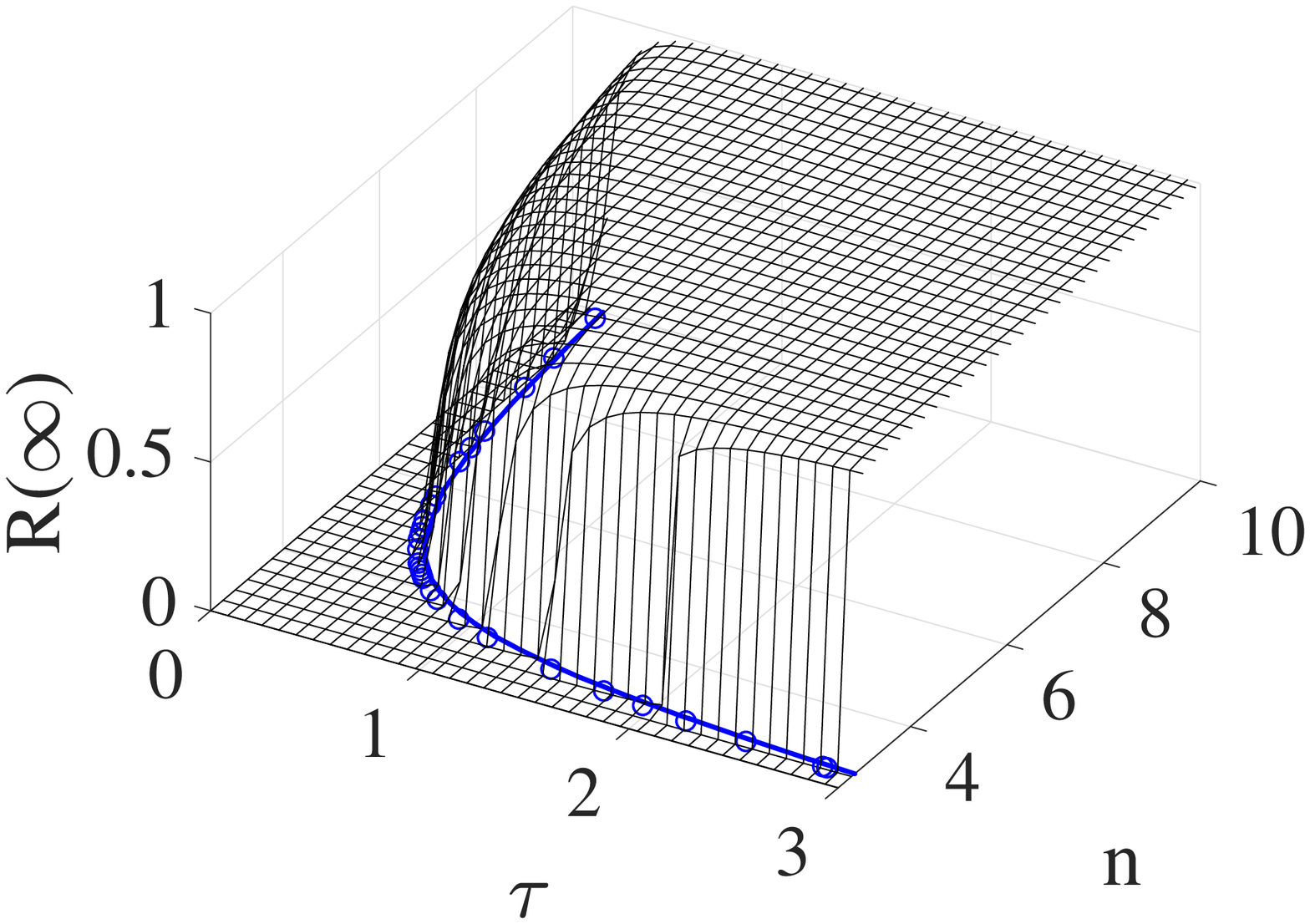}
%\label{fig:simp_clo_thres_c0p4}
        }
      \hspace{1.0cm}
 %       \subfloat[]     
 {\includegraphics[scale=0.25]{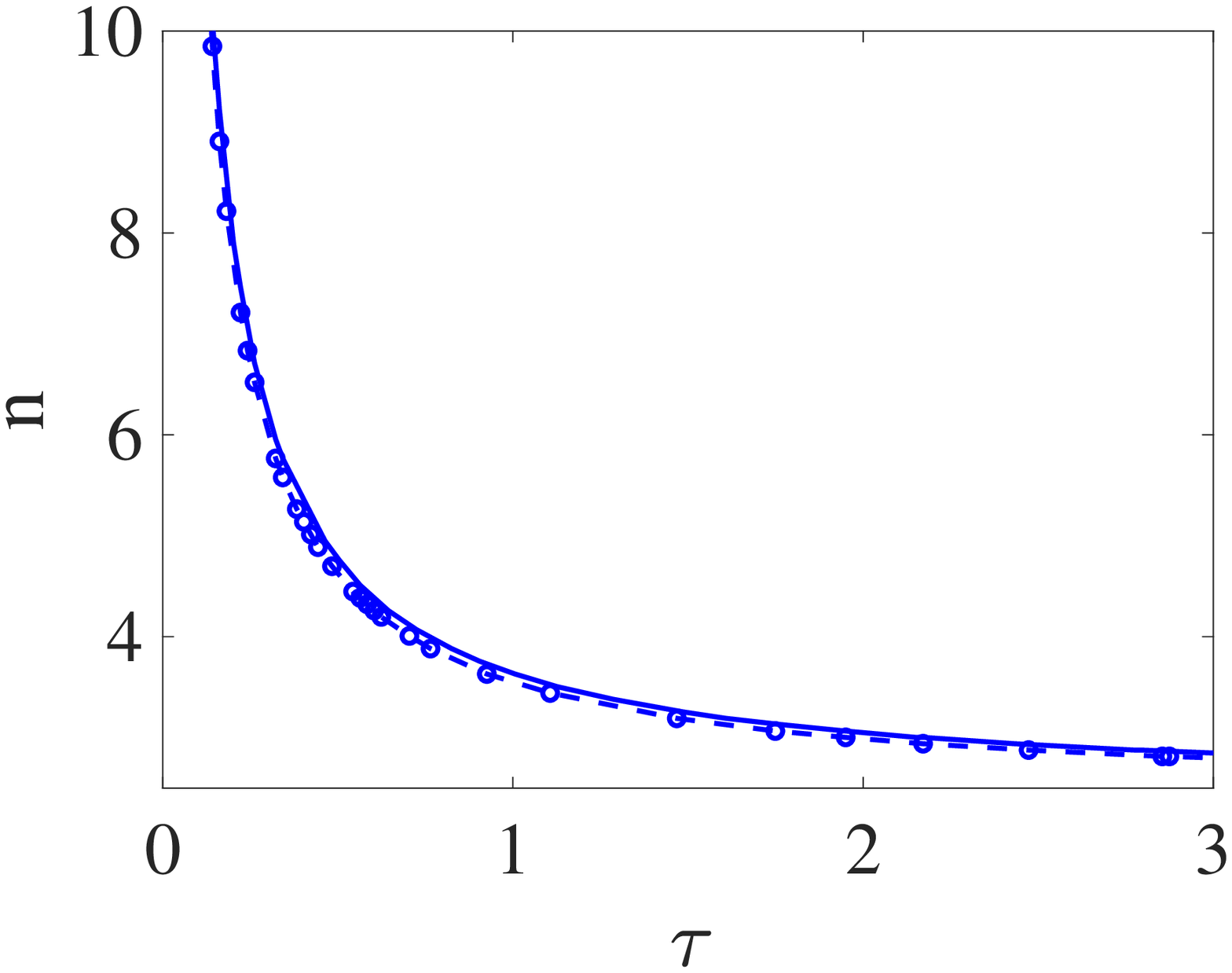}
%\label{fig:simp_clo_thres_c0p5}
        }%\hfill
        %\subfloat[]
{\includegraphics[scale=0.25]{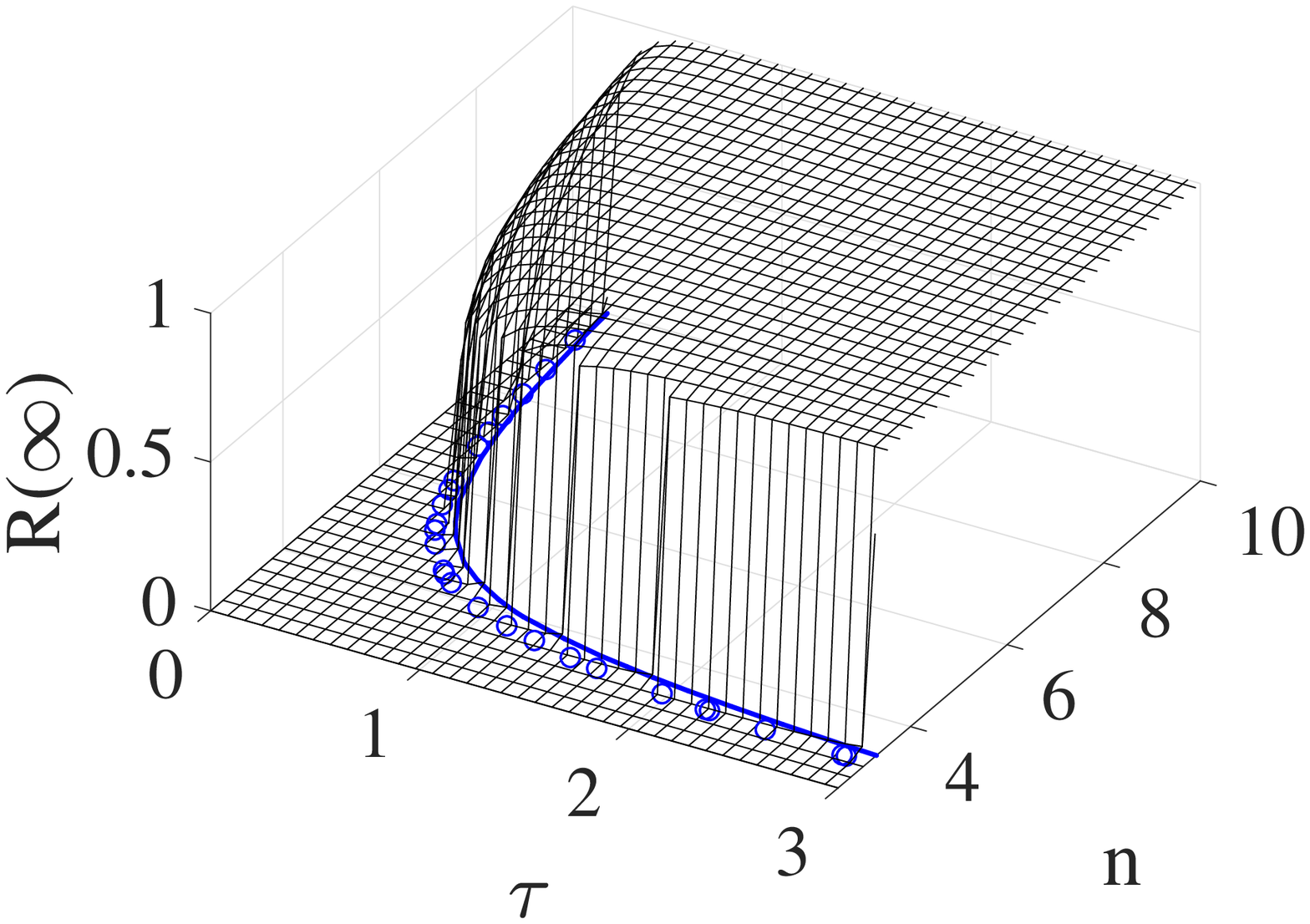}
%\label{fig:simp_clo_thres_c0p2}
        }
\hspace{1.0cm}
%\subfloat[]
{\includegraphics[scale=0.25]{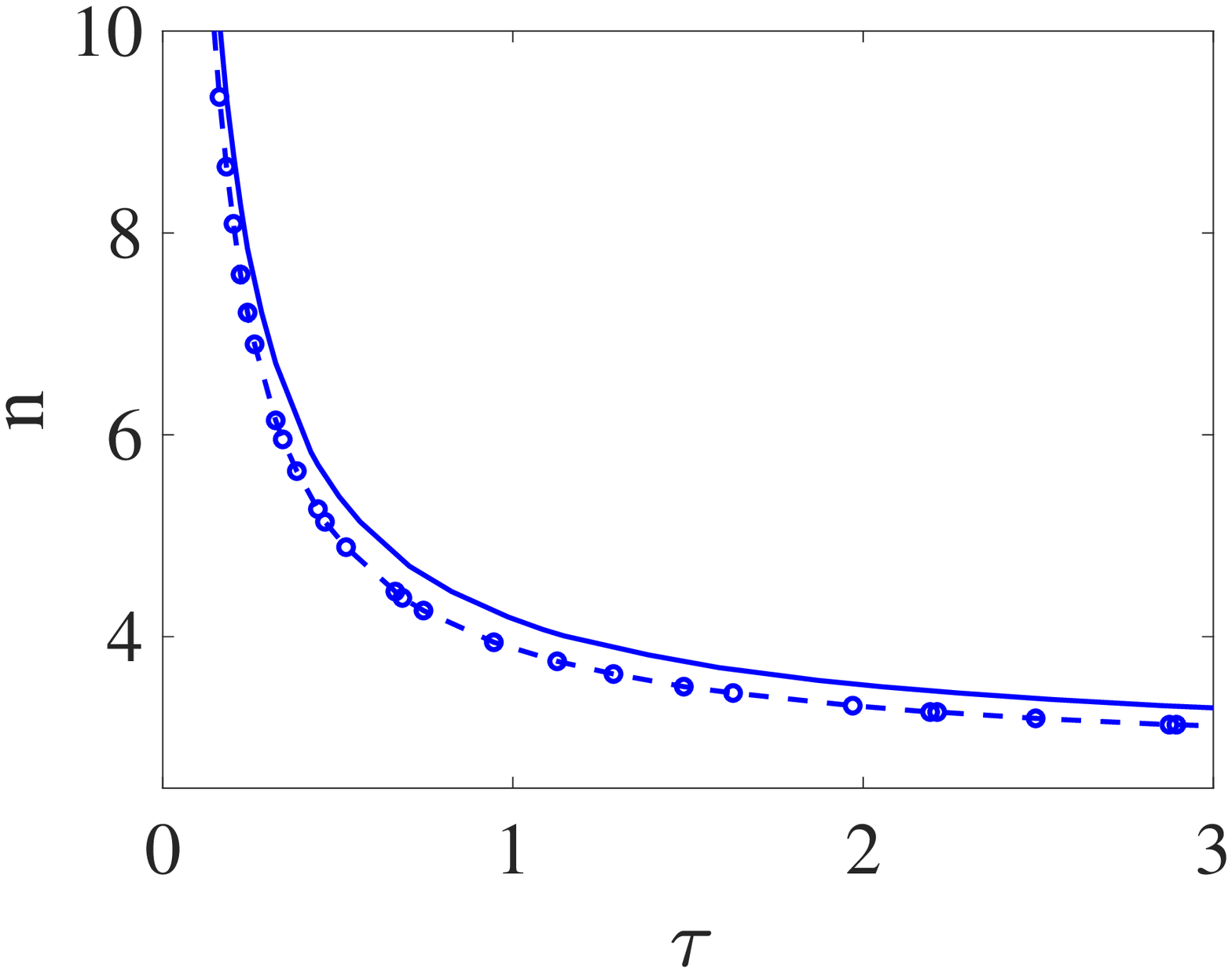}
%\label{fig:simp_clo_thres_c0p3.eps}
        }
        %\hfill
%\subfloat[]     
{\includegraphics[scale=0.25]{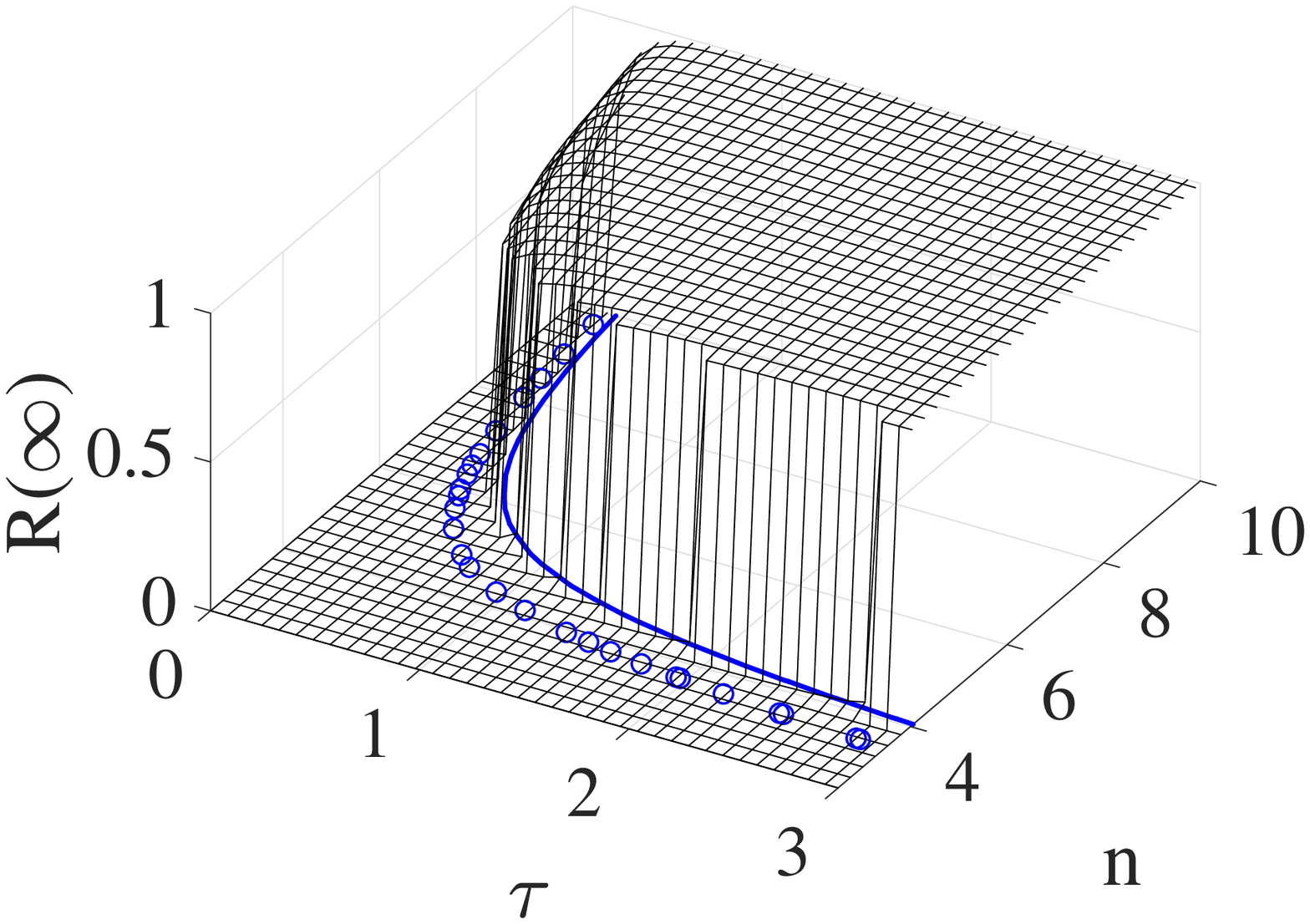}
%\label{fig:simp_clo_thres_c0p4}
        }
       \hspace{1.0cm}
%        \subfloat[]     
{\includegraphics[scale=0.25]{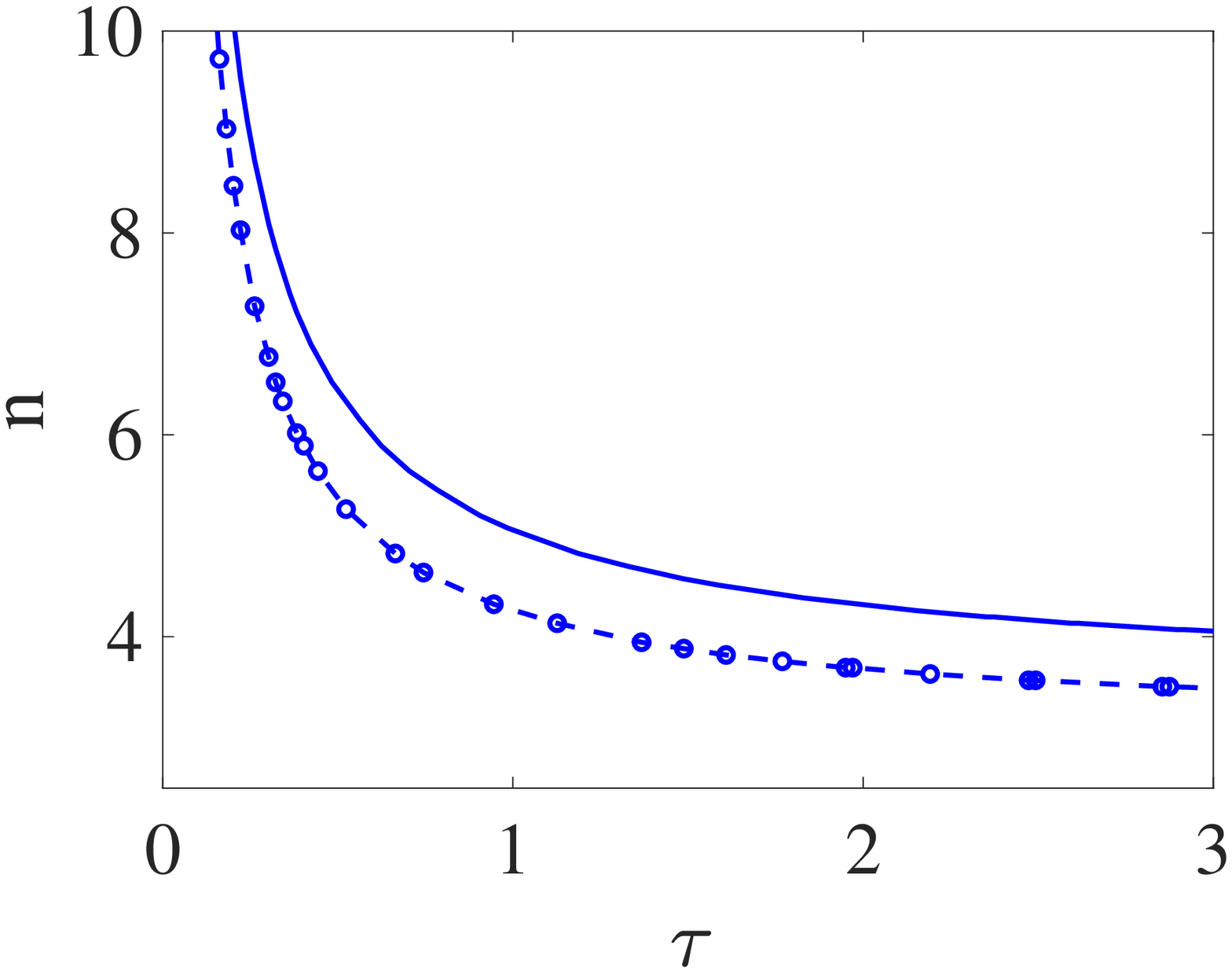}
%\label{fig:simp_clo_thres_c0p5}
        }
\caption{Assessing the validity of the epidemic threshold based on the asymptotic approximation \eqref{eq:approximation_for_alpha} (dashed line and markers - $\circ$) by comparing it to the epidemic threshold based on the numerical solution of the cubic equation \eqref{eq:cubic_eqn_for_alpha} (continuous lines). In the right hand column we compare both threshold curves in the $(\tau,n,0)$ plane. In the left hand column both curves are compared to the final epidemic size based on numerical integration of the pairwise model equations with the simple closure. Parameter values are $N=10000$, $\gamma=1$ and from top to bottom the clustering coefficients are $\phi=0, 0.15, 0.3, 0.45, 0.6$. 
%OLDER CAPTION ---- Illustration of the epidemic threshold based on the numerical solution of the cubic equation (continuous lines) and that of asymptotic approximation (dashed line and markers - $\circ$). Both are compared to final epidemic size based on the numerical integration of the pairwise model with the simple closure. Parameter values are $N=10000$, $\gamma=1$, the values of clustering from top to bottom are $\phi=0, 0.15, 0.3, 0.45, 0.6$. For clarity, the right column illustrates the threshold curves in 2D, that is being in the $(\tau,n,0)$ plane.
}
  \label{fig:simp_clo_thres}
\end{figure}

%%%%%%%%%%%%%%%%%%%%%%%%%%%%%%%%
\section{Results for the pairwise model with the compact improved closure} \label{section:Results_comp_imp}
%%%%%%%%%%%%%%%%%%%%%%%%%%%%%%%%

Starting from the improved closure~\eqref{eq:improved_closure} but in line with Proposition~\ref{prop:inequ_closures}, we adapt the closure so that the term responsible for the approximation on the clustered part of the network does not consider variables, singles or pairs involving the recovered/removed class. This leads to the new closure
 \begin{equation} \label{eq:compact_improved_closure}
 [ASI]=(n-1)\left((1-\phi)\frac{[AS][SI]}{n[S]}+\phi\frac{[AS][SI][IA]}{[A]\left(\frac{[SS][SI]}{[S]}+\frac{[SI][II]}{[I]}\right)}\right),
 \end{equation}
which we refer to as the compact improved closure. Plugging equation \eqref{eq:compact_improved_closure} into the exact system~\eqref{equations:unclosed_pw_model_one}-\eqref{equations:unclosed_pw_model_end} leads to a self-consistent system that is written out in full in Appendix~\ref{sec:correlation_struct_and_fast_vars_reduced_improved_closure}. 

In line with our procedure so far, we aim to find the epidemic threshold of this new pairwise system with the compact improved closure. It turns out that the approach used for the pairwise system with the simple closure is applicable to this case, and the steps and results are summarised below.

%%%%%%%%%%%%%%%%%%%%%%%%%%%%%%
\subsection{Fast variables with the compact improved closure}
%%%%%%%%%%%%%%%%%%%%%%%%%%%%%%

As we have shown before, finding the threshold relies on finding the quasi-equilibrium of $\alpha=\frac{[SI]}{[I]}$. In Appendix~\ref{sec:correlation_struct_and_fast_vars_reduced_improved_closure} we show that this requires knowledge about the behaviour of the $\delta=\frac{[II]}{[I]}$ variable and indeed a system of differential equations involving these two variables can be derived. This system is given below \begin{eqnarray} \frac{d\alpha}{dt}&=&-\tau\alpha-\tau\alpha^{2}+\tau(n-1)\left((1-\phi)\alpha+\phi\alpha\left(\frac{n-\delta}{n+\delta}\right)\right), \label{eq:alpha_comp_imp}\\ \frac{d\delta}{dt}&=&2\tau\alpha-\gamma\delta+2\tau(n-1)\left(\frac{\phi\alpha\delta}{n+\delta}\right)-\tau\alpha\delta. \label{eq:delta_comp_imp}\end{eqnarray}
As previously, the steady state of this system is of interest and apart from the trivial $(\alpha^{*},\delta^{*})=(0,0)$ steady state, the quasi-equilibrium can be found by first expressing $\delta$ as a function of $\alpha$. This can be done by setting equation~\eqref{eq:alpha_comp_imp} equal to zero and rearranging, leading to \begin{equation} \alpha=(n-2)-(n-1)\phi\frac{2\delta}{n+\delta}. \label{eq:alpha_star} \end{equation} Plugging equation \eqref{eq:alpha_star} into equation~\eqref{eq:delta_comp_imp} and collecting powers of $\delta$ leads to the following cubic equation \begin{align} &(-A-B)\delta^{3}+(-n(n-2)-A^2-2nB)\delta^{2}\notag\\ &+(-n(n-2)A+2nA-n^2B)\delta+2n^2(n-2)=0,\label{eq:cubic_delta_comp_imp} \end{align} where $A=(n-2)-2\phi(n-1)$ and $B=\gamma/\tau$. It is worth noting that in this case it is easier to work with $\delta$, but any results can be converted in terms of $\alpha$ which is the main variable of interest. %Details of the calculations are given in  Appendix~\ref{sec:correlation_struct_and_fast_vars_reduced_improved_closure}.

%%%%%%%%%%%%%%%%%%%%%%%%%%%%%%
\subsection{Asymptotic expansion of the epidemic threshold}
%%%%%%%%%%%%%%%%%%%%%%%%%%%%%%

As in Section \ref{sec:asymptotic_expansion_1}, we require the roots of the cubic polynomial given in equation \eqref{eq:cubic_delta_comp_imp}. To do so, we express $\delta$ as an asymptotic expansion in powers of $\phi$. We substitute \begin{equation} \delta=\delta_{0}+\delta_{1}\phi+\delta_{2}\phi^{2}+\cdots . \label{eq:expansion_for_delta} \end{equation} Plugging the expansion for $\delta$ \eqref{eq:expansion_for_delta} into equation \eqref{eq:cubic_delta_comp_imp} leads to \begin{dmath} (-A-B)(\delta_{0}+\delta_{1}\phi+\delta_{2}\phi^{2}+\cdots)^{3}+(-n(n-2)-A^{2}-2nB)(\delta_{0}+\delta_{1}\phi+\delta_{2}\phi^{2}+\cdots)^{2}+(-n(n-2)A+2nA-n^{2}B)(\delta_{0}+\delta_{1}\phi+\delta_{2}\phi^{2}+\cdots)+2n^{2}(n-2)=0. \label{eq:asymptotic_expansion_for_delta} \end{dmath} Alternatively, substituting \eqref{eq:alpha_star} into the differential equation for $\delta$ \eqref{eq:delta_comp_imp}, setting the expression equal to zero and rearranging leads to \begin{equation} \gamma\delta(n+\delta)^{2}=\tau[(n-2)(n+\delta)-2\phi(n-1)\delta][(2-\delta)(n+\delta)+2\phi(n-1)\delta]. \label{eq:alternative_expansion_for_delta} \end{equation} Substituting \eqref{eq:expansion_for_delta} into \eqref{eq:alternative_expansion_for_delta} and collecting terms of order $\phi^{0}$ yields \begin{eqnarray} \gamma\delta_{0}(n+\delta_{0})^{2}&=&\tau[(n-2)(n+\delta_{0})][(2-\delta_{0})(n+\delta_{0})] \\ \gamma\delta_{0}&=&\tau(n-2)(2-\delta_{0}) \\ \delta_{0}(\gamma+\tau(n-2))&=&2\tau(n-2) \\ \delta_{0}&=&\frac{2\tau(n-2)}{\gamma+\tau(n-2)}. \label{eq:expression_for_delta_0} \end{eqnarray} Following the same process to collect terms of order $\phi^{1}$, we find \begin{dmath} \gamma\delta_{1}[(n+\delta_{0})^{2}+2(n+\delta_{0})\delta_{0}]=\tau(n-2)(n+\delta_{0})[\delta_{1}(2-n-2\delta_{0})+2(n-1)\delta_{0}]+\tau(2-\delta_{0})(n+\delta_{0})[(n-2)\delta_{1}-2(n-1)\delta_{0}], \end{dmath} which can be rearranged to yield \begin{equation} \delta_{1}=\frac{2\tau(n-1)\delta_{0}(n-4+\delta_{0})}{\gamma(n+3\delta_{0})+\tau(n-2)(n+3\delta_{0}-4)} \label{eq:expression_for_delta_1}, \end{equation} with $\delta_{0}$ defined in \eqref{eq:expression_for_delta_0}. In summary, we have determined the first two coefficients $\delta_{0}$ and $\delta_{1}$ of the asymptotic expansion for $\delta$ given in equation \eqref{eq:expansion_for_delta}. Hence, the true solution is approximated by the following expression: \begin{dmath} \delta=\frac{2\tau(n-2)}{\gamma+\tau(n-2)}+\frac{2\tau(n-1)\delta_{0}(n-4+\delta_{0})\phi}{\gamma(n+3\delta_{0})+\tau(n-2)(n+3\delta_{0}-4)}+\cdots. \label{eq:final_expansion_for_delta} \end{dmath} Finally, we are able to plug \eqref{eq:final_expansion_for_delta} into the quasi-equilibrium point for $\alpha$, given in equation \eqref{eq:alpha_star}, to obtain \begin{equation} \alpha=(n-2)-2(n-1)\phi\frac{\delta_{0}}{n+\delta_{0}}+\mathcal{O}(\phi^{2}), \end{equation} which can be rearranged to find \begin{equation} \alpha=(n-2)-\phi\frac{4\tau(n-1)(n-2)}{\tau(n+2)(n-2)+\gamma n}. \label{eq:expression_for_alpha_2} \end{equation} The expression for $\alpha$ \eqref{eq:expression_for_alpha_2} can be used to determine the epidemic threshold as follows
\begin{equation}
R^{cci}=\frac{\tau\alpha}{\gamma}=\frac{(n-2)\tau}{\gamma}-\phi\frac{\tau}{\gamma}\left(\frac{4\tau(n-1)(n-2)}{\tau(n+2)(n-2)+\gamma n}\right).
\label{eq:threshold_cci}
\end{equation}
It is straightforward to see that again $R^{cci}\le R$, with clustering making the spread of the epidemic less likely.
%Rearranging \eqref{eq:alpha_star} to isolate $\delta$ means we can calculate an equivalent quasi-equilibrium point for $\delta$ in terms of $\alpha$: \begin{equation} \delta=\frac{n(n-\alpha-2)}{\alpha+2(n-1)\phi+2-n}. \end{equation}

%%%%%%%%%%%%%%%%%%%%%%%%%%%%%%
\subsection{Numerical examples}
%%%%%%%%%%%%%%%%%%%%%%%%%%%%%%

In Fig.~\ref{fig:comp_imp_clo_thres} we repeat the systematic test of comparing the epidemic threshold generated via the numerical solution of the cubic equation \eqref{eq:cubic_delta_comp_imp}, the epidemic threshold generated by the asymptotic expansion \eqref{eq:threshold_cci} and the numerical value of the final epidemic size predicted by pairwise model with the compact improved closure, over a wide range of $(\tau,n)$ values. Several observations can be made. First, it is clear that higher values of clustering push the location of threshold to higher $\tau$ and $n$ values, meaning that the limiting effect of clustering on the epidemic spread can only be overcome if either the value of the transmission rate or average degree increases. Second, the agreement between the threshold  based on the numerical solution of the cubic equation~\eqref{eq:cubic_delta_comp_imp} and the asymptotic expansion~\eqref{eq:final_expansion_for_delta} is excellent over a wide range of $\phi$ values. In fact, in this case the agreement is excellent for $0\le \phi \le 0.45$, with only small deviations even for $\phi=0.6$. The agreement between the numerical solution of the pairwise model and the threshold based on the numerical solution of the cubic 
equation \eqref{eq:cubic_delta_comp_imp} remains excellent across all parameter values.

\begin{figure}
\centering
%\subfloat[]
{\includegraphics[scale=0.25]{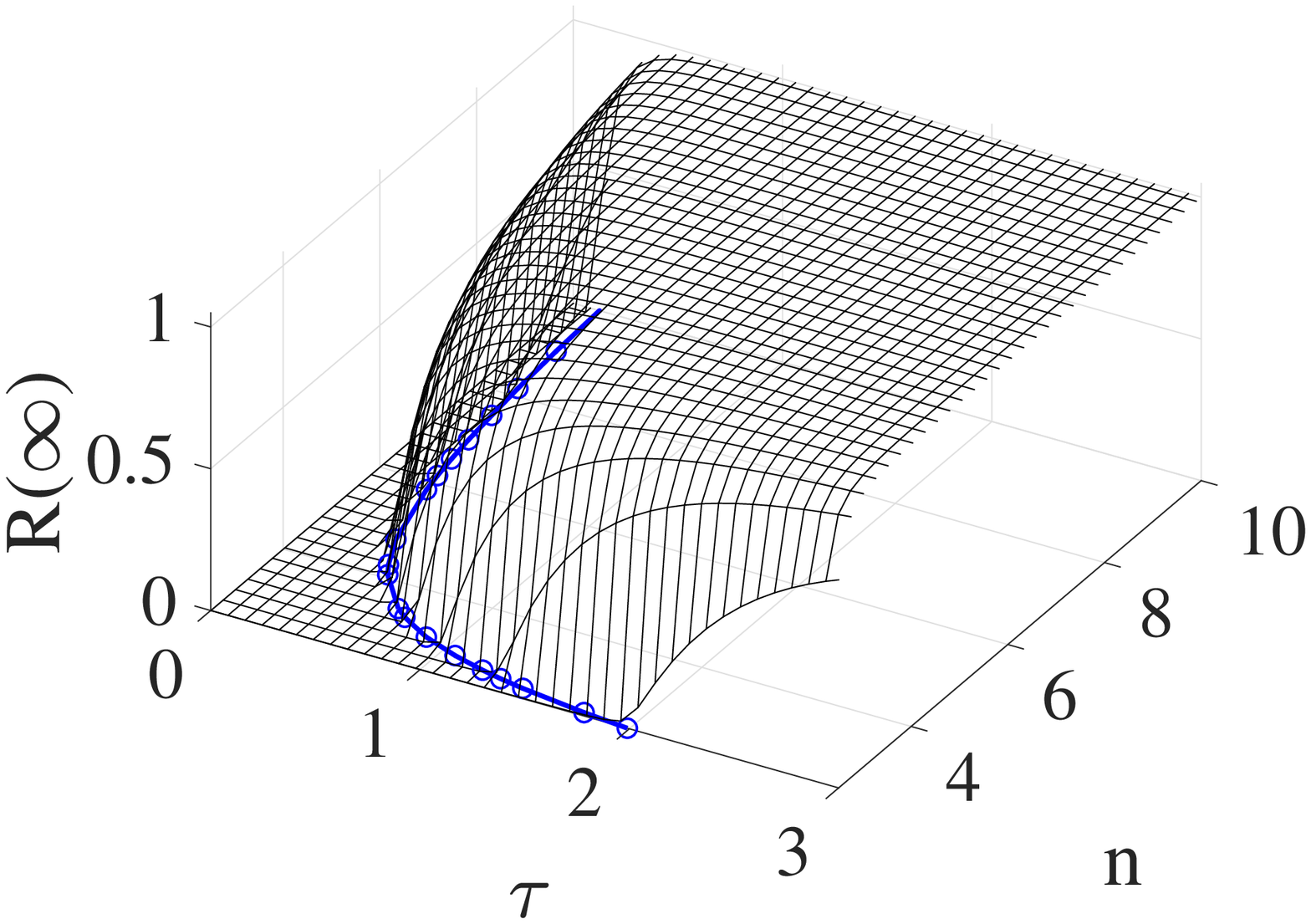}
% \label{fig:simp_clo_thres_c0}
    }
\hspace{1.0cm}
%\subfloat[]
{\includegraphics[scale=0.25]{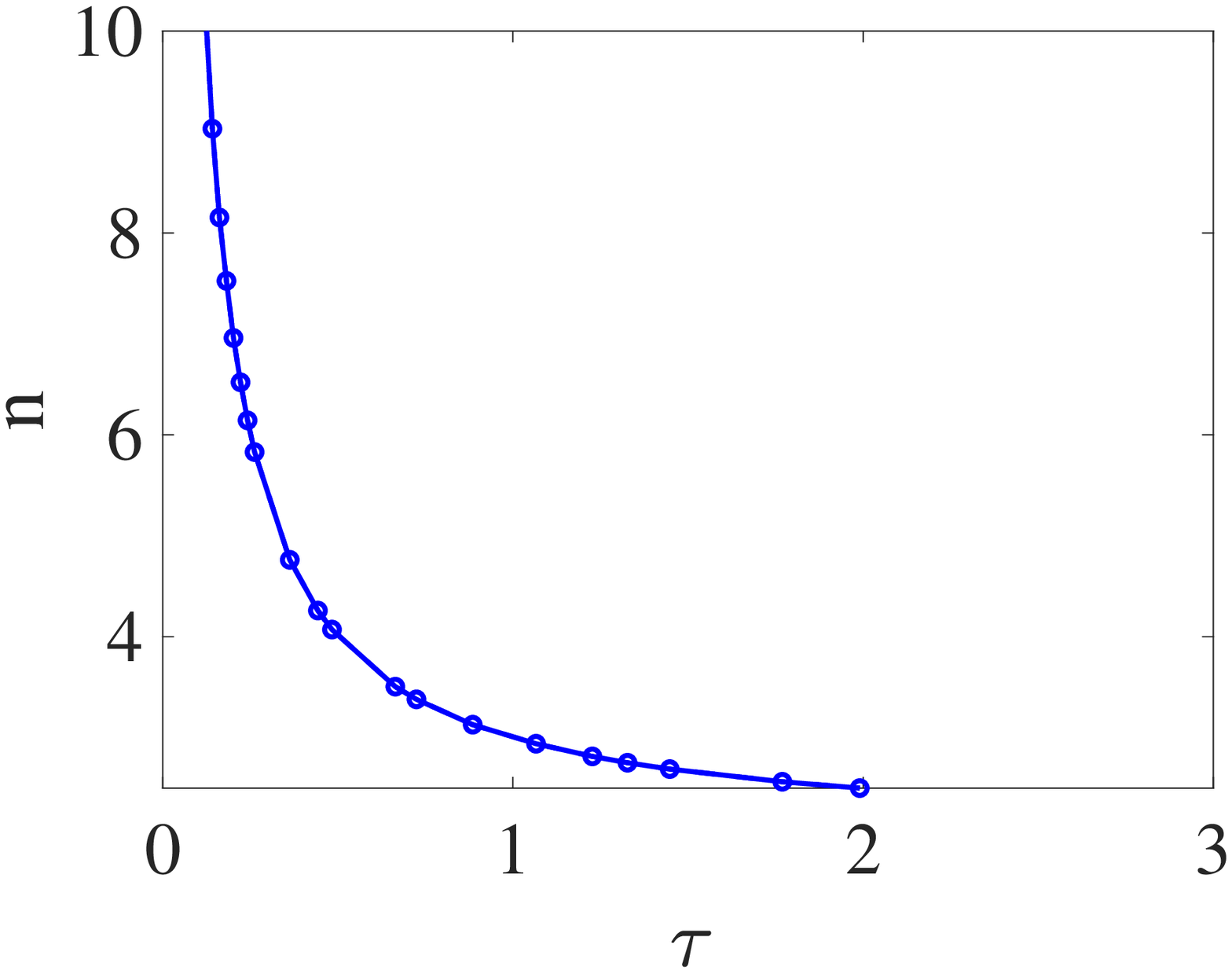}
% \label{fig:simp_clo_thres_c0p1}
    }%\hfill
%\subfloat[]
{\includegraphics[scale=0.25]{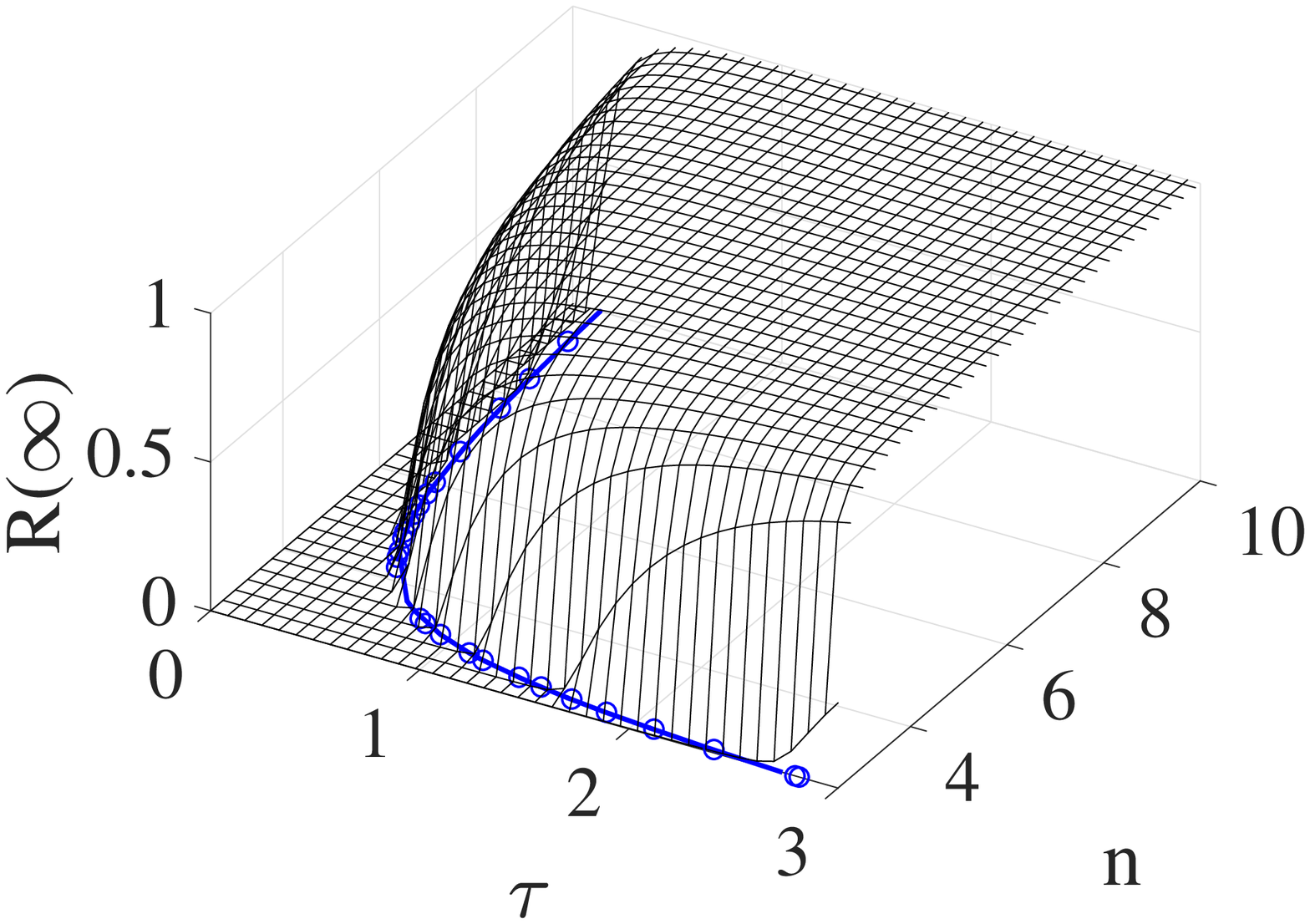}
%\label{fig:simp_clo_thres_c0p2}
        }
\hspace{1.0cm}
%\subfloat[]
{\includegraphics[scale=0.25]{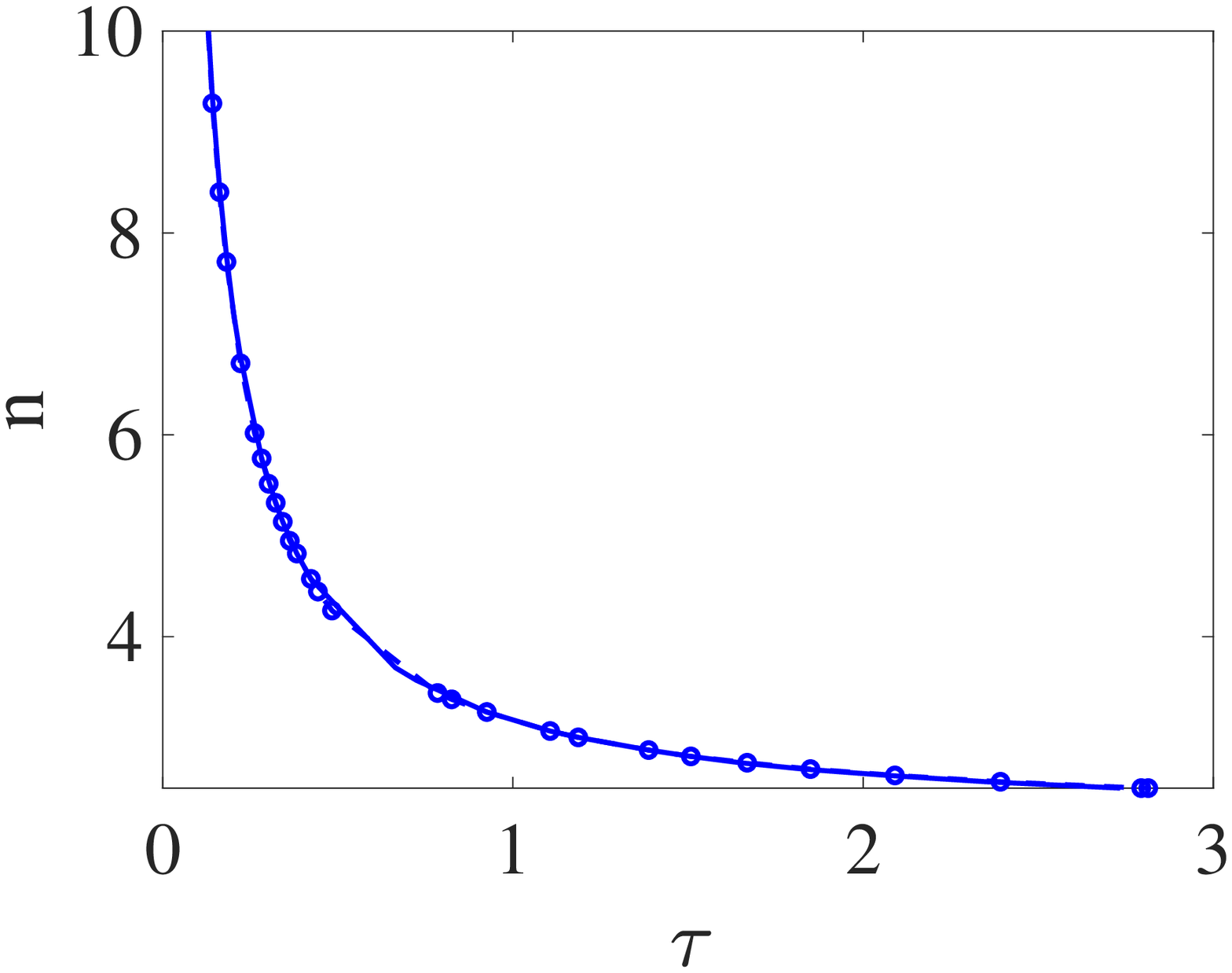}
%\label{fig:simp_clo_thres_c0p3.eps}
  }%\hfill
%\subfloat[]     
{\includegraphics[scale=0.25]{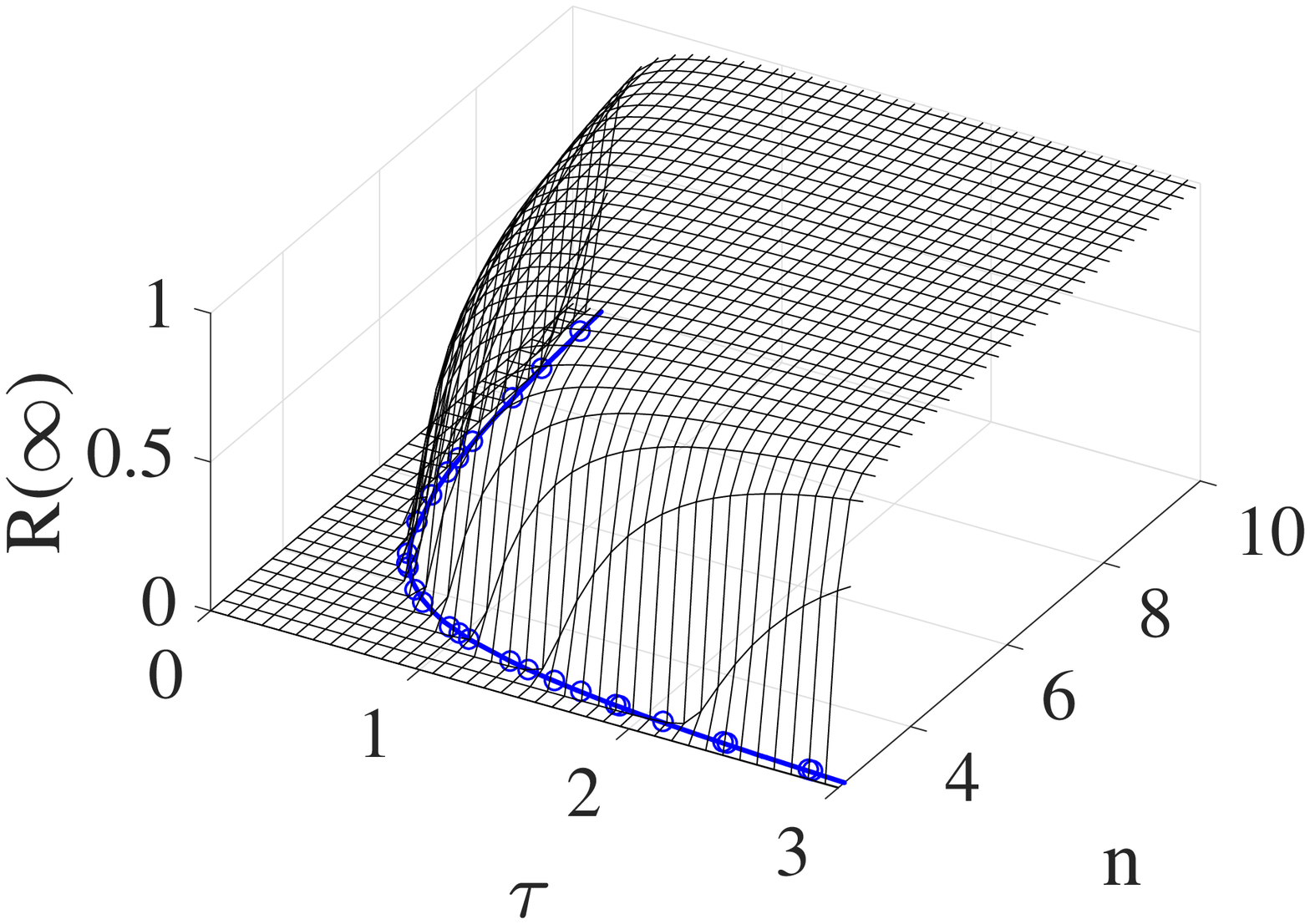}
%\label{fig:simp_clo_thres_c0p4}
        }
       \hspace{1.0cm}
 %       \subfloat[]     
 {\includegraphics[scale=0.25]{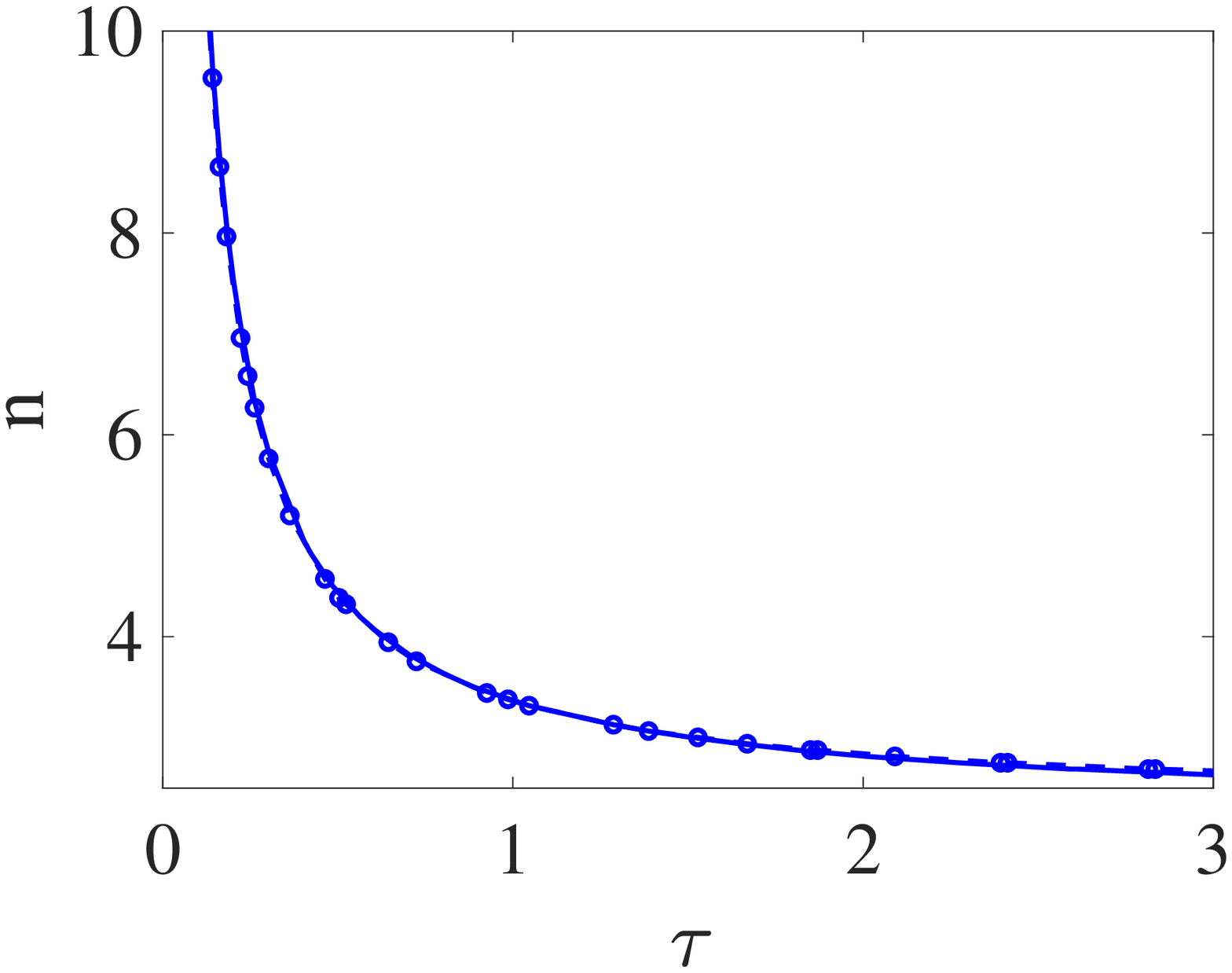}
%\label{fig:simp_clo_thres_c0p5}
        }%\hfill
        %\subfloat[]
{\includegraphics[scale=0.25]{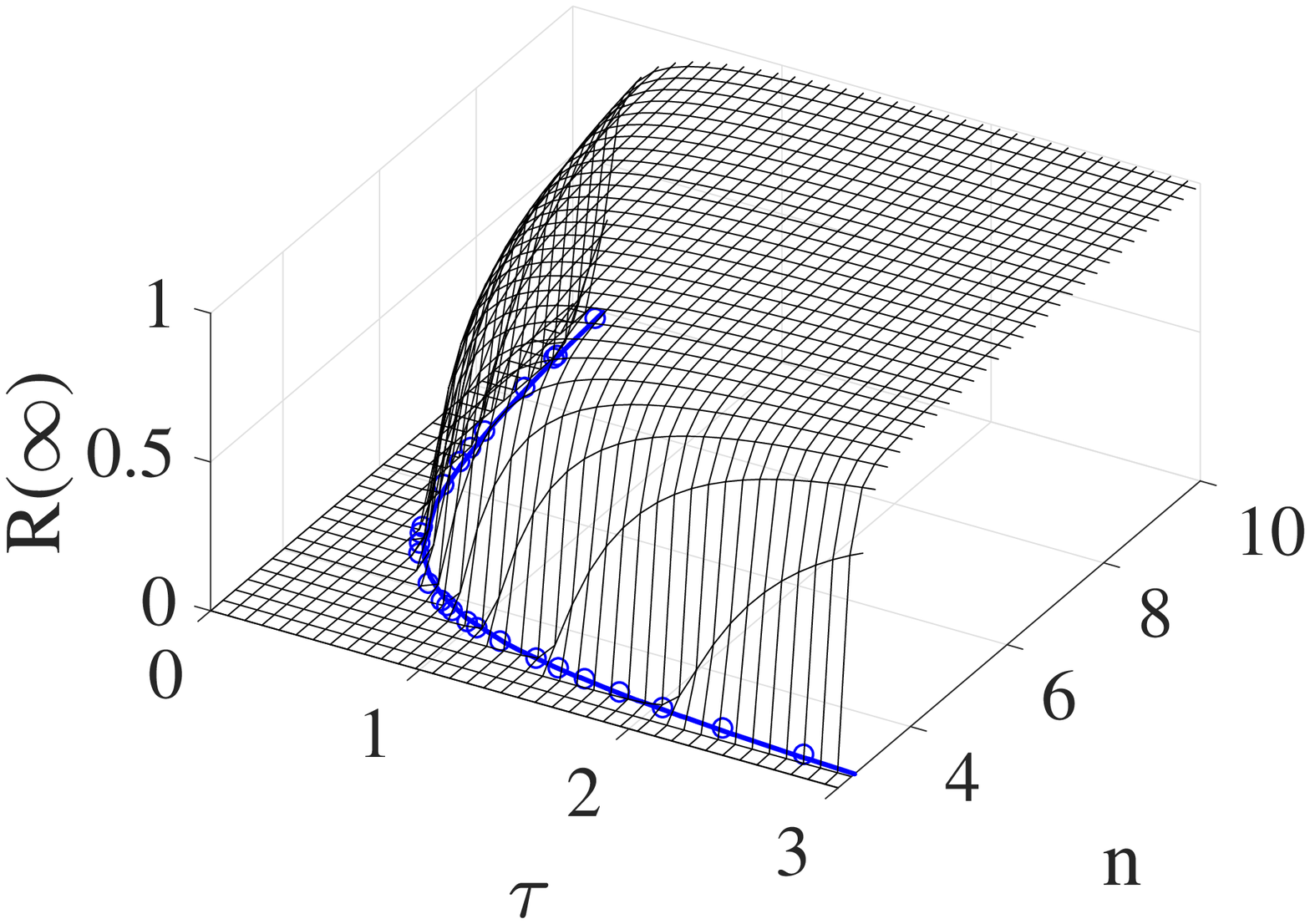}
%\label{fig:simp_clo_thres_c0p2}
        }
\hspace{1.0cm}
%\subfloat[]
{\includegraphics[scale=0.25]{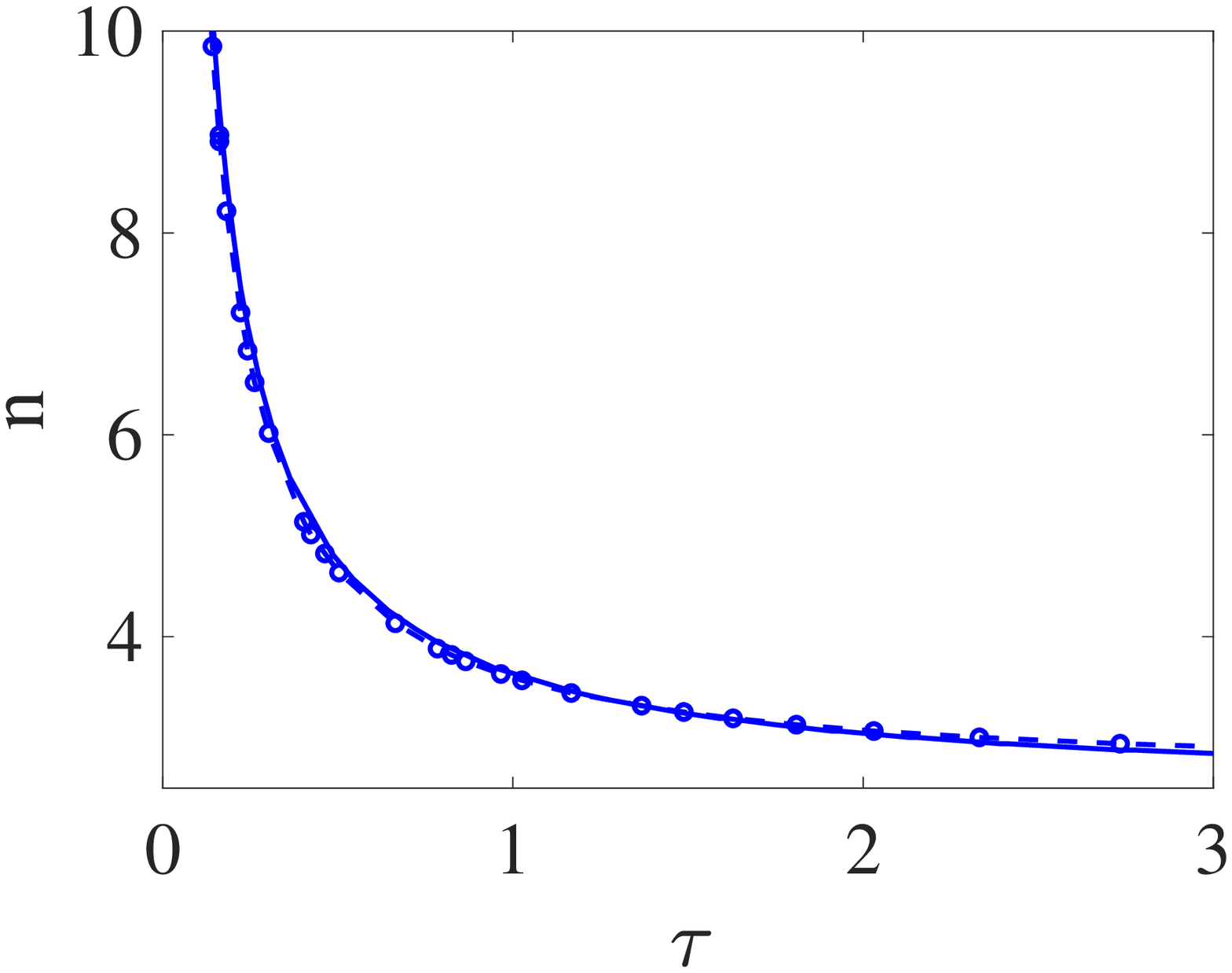}
%\label{fig:simp_clo_thres_c0p3.eps}
        }
        %\hfill
%\subfloat[]     
{\includegraphics[scale=0.25]{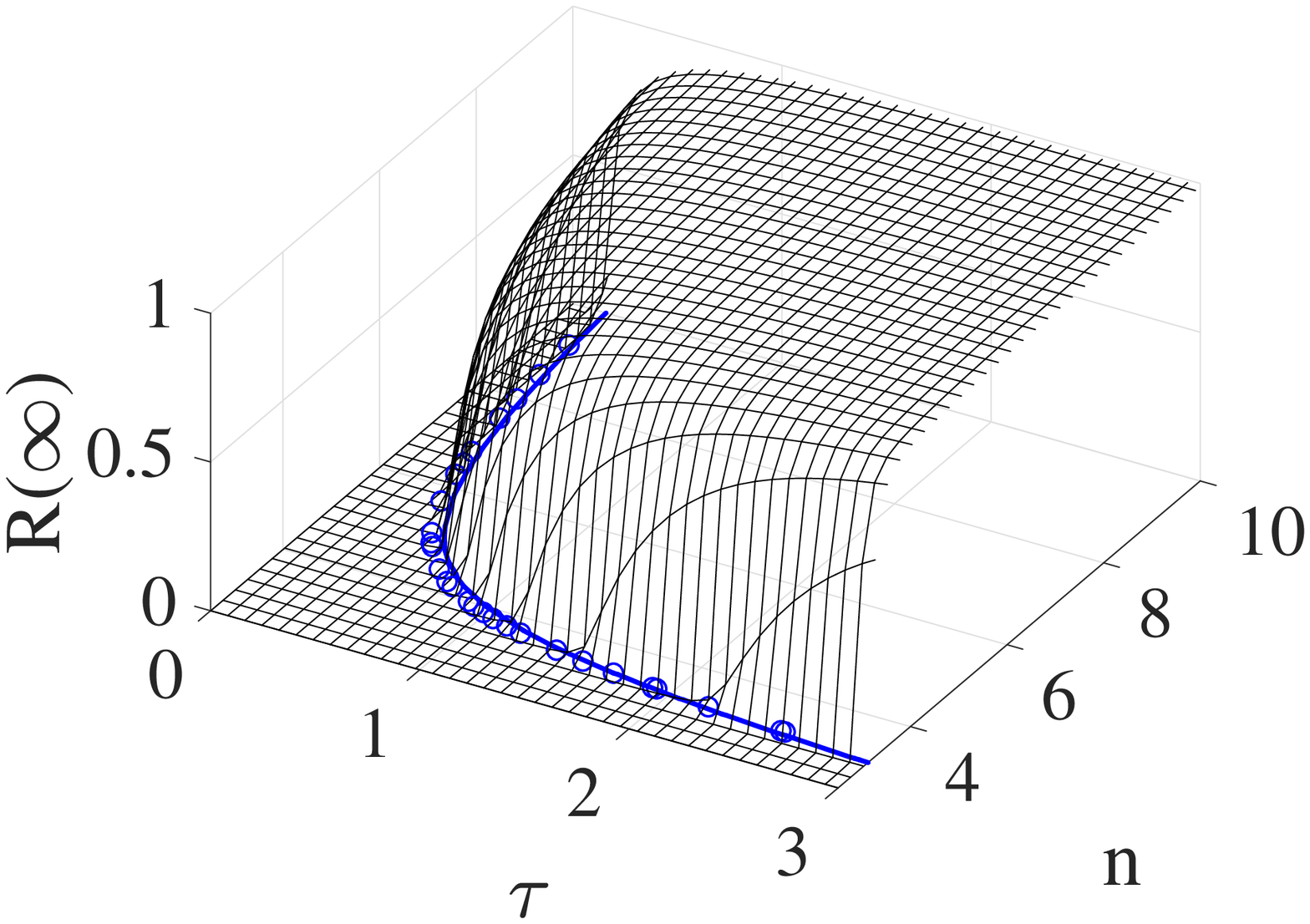}
%\label{fig:simp_clo_thres_c0p4}
        }
        \hspace{1.0cm}
%        \subfloat[]     
{\includegraphics[scale=0.25]{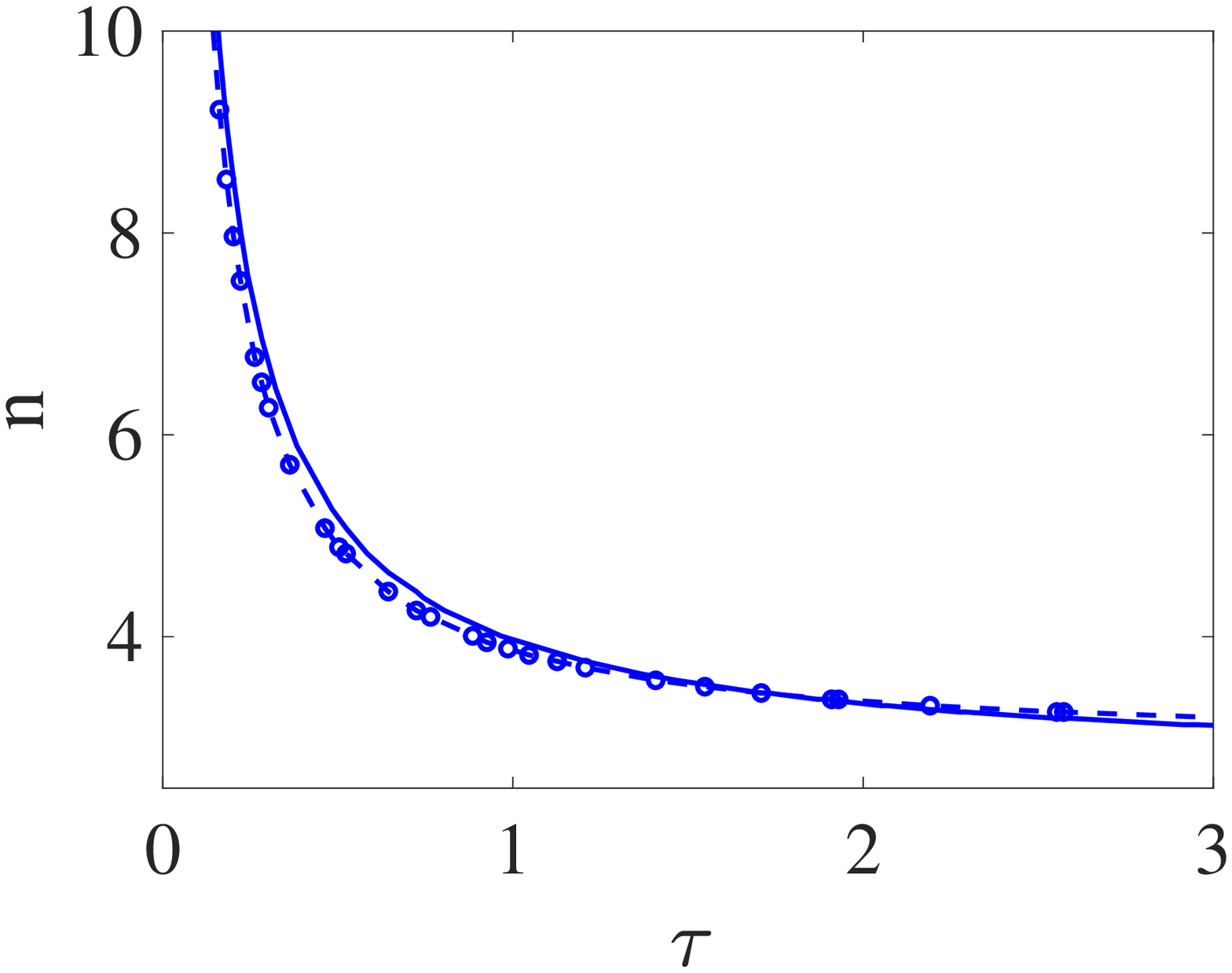}
%\label{fig:simp_clo_thres_c0p5}
        }
\caption{Assessing the validity of the epidemic threshold based on the asymptotic expansion \eqref{eq:final_expansion_for_delta} (dashed line and markers - $\circ$) by comparing it to the epidemic threshold based on the numerical solution of the cubic equation \eqref{eq:cubic_delta_comp_imp} (continuous lines). In the right hand column we compare both threshold curves in the $(\tau,n,0)$ plane. In the left hand column both curves are compared to the final epidemic size based on numerical integration of the pairwise model equations with the compact improved closure. Parameter values are $N=10000$, $\gamma=1$ and from top to bottom the clustering coefficients are $\phi=0, 0.15, 0.3, 0.45, 0.6$.
%OLDER CAPTION --- Illustration of the epidemic threshold based on the numerical solution of the cubic equation (continuous lines) and that of asymptotic approximation (dashed line and markers - $\circ$). Both are compared to final epidemic size based on the numerical integration of the pairwise model with the improved closure not including any variable involving the removed class. Parameter values are $N=10000$, $\gamma=1$, the values of clustering from top to bottom are $\phi=0, 0.15, 0.3, 0.45, 0.6$. For clarity, the right column illustrates the threshold curves in 2D, that is being in the $(\tau,n,0)$ plane.
}
  \label{fig:comp_imp_clo_thres}
\end{figure}

%%%%%%%%%%%%%%%%%%%%%%%%%%%%%%%%%%%%%%%%%%
\section{Discussion}
\label{sec:discussion}
%%%%%%%%%%%%%%%%%%%%%%%%%%%%%%%%%%%%%%%%%%

In this paper we set out to obtain an analytic epidemic threshold using pairwise models but for clustered networks. For the unclustered case this problem has been solved previously~\cite{keeling1999effects}. Furthermore, in~\cite{keeling1999effects} it was shown that one way to approach the computation of the threshold is to exploit the presence of fast variables. In particular, working out the quasi-steady state of the fast variables allowed the authors to determine the epidemic threshold analytically. However, this was done only for the case when the network is unclustered. Here, we went one step further and showed that the quasi-equilibrium can be found as an asymptotic expansion in powers of the clustering coefficient. Prior to this new result we re-derived known closures by providing extra intuition for the assumptions underlying them as well as for the motivation for deriving them. 

Exploiting the presence of fast variables and combining this with elements of perturbation theory allowed us to compute the epidemic threshold for the pairwise model with two different closures that take clustering into account. Our results are in line with the findings of \cite{li2018epidemic} and \cite{miller2009spread}. In~\cite{li2018epidemic}, the epidemic threshold in a pairwise model for clustered networks with closure based on the number of links in a motif, rather than nodes, was calculated as
\begin{equation}
R_0=\frac{(n-1)\tau}{\tau+\gamma+\tau\phi}. \label{eq:disc_threshold_example}
\end{equation}
Equation \eqref{eq:disc_threshold_example} can be expanded in terms of $\phi$ to give 
\begin{equation}
R_0=\frac{(n-1)\tau}{\tau+\gamma}\left(\frac{1}{1+\phi\frac{\tau}{\tau+\gamma}}\right)\simeq \frac{(n-1)\tau}{\tau+\gamma}\left( 1-\phi \frac{\tau}{\tau+\gamma}+\cdots\right),
\label{eq:motif_based_threshold_expanded}
\end{equation}
which again reflects our finding that clustering reduces the epidemic threshold. 

Similarly but for clustered networks with heterogeneous degree distributions, in~\cite{miller2009spread} it was found that
\begin{equation}
R_0=\frac{\langle k^2-k\rangle}{\langle k\rangle}T-\frac{2\langle n_{\triangle}\rangle}{\langle k\rangle}T^2+\cdots,
\end{equation}
where $\langle k^i \rangle$ stands for the $i$th moment of the degree distribution, $T$ is the probability of infection spreading across a link connecting an infected to a susceptible node and $\langle n_{\triangle}\rangle$ denotes the average number of triangles that a node belongs to. The expression above again shows that clustering reduces the epidemic threshold when compared to the unclustered case. Furthermore, if the network is regular and we assume that infections and recoveries are Markovian processes with rates $\tau$ and $\gamma$ respectively, giving $T=\tau/(\tau+\gamma)$, $R_0$ above reduces to
\begin{equation}
R_0=\frac{\tau (n-1)}{\tau+\gamma}-(n-1)\phi\left(\frac{\tau}{\tau+\gamma}\right)^2+\cdots,
\end{equation}
where we have used the fact that a global clustering coefficient of $\phi$ translates to a node on average being part of $\frac{1}{2}n(n-1)\phi$ uniquely counted triangles. This in turn coincides with equation~\eqref{eq:motif_based_threshold_expanded}, and this is perhaps unexpected since the first expression was obtained based on a new type of closure for pairwise models while the other expression was based on percolation theory type arguments. In~\cite{trapman2007analytical}, specific networks with household structure were used to investigate the effects of clustering and infectious period distribution on a modified version of $R_0$ referred to as $R_{*}$, and lower and upper bounds for the value of this quantity were found.

Our analysis confirms that clustering starves the spreading epidemic of susceptible neighbours and that the epidemic is less likely to spread if the networks are clustered, all other parameters being equal. More importantly, the epidemic threshold is model-dependent and the pairwise model with the compact improved closure leads more readily to epidemic outbreaks when compared to the pairwise model with the simple closure, see Figs.~\ref{fig:simp_clo_thres}-\ref{fig:comp_imp_clo_thres}. While this ordering is true for the parameters used in this paper, it is easy to show that this relation can change if  parameters are tuned accordingly. For example, looking at the limit of $\gamma \rightarrow 0$ (or $\tau/\gamma$ large limit), the two epidemic thresholds are the same if 
\begin{equation}
\frac{2(n-1)}{n^{2}}\left(\frac{2(n-1)(n-2)}{(n-2)}\right)=\frac{4(n-1)(n-2)}{(n-2)(n+2)}.
\end{equation}
After some simple algebra this reduces to $n=2$. Hence, if the $\tau/\gamma$ ratio is large we will essentially have that (i) if $n>2$ then $R_{0}^{c}<R_{0}^{cci}$, and (ii) if $n<2$ then $R_{0}^{c}>R_{0}^{cci}$. This highlights the difficulty of determining the epidemic threshold and emphasises the importance of model choice when modelling real-world epidemics.

%Start with a summary of what we have done explaning that we have shown that exploting the presence of fast variables and coupling this with perturbation methods we were able for the first time, according to out knowledge, to derive an analytic expression for the epidemic threshold using pariwise models for clustered network with two different type of clsoures.

The analysis of the pairwise model with the full improved closure is still outstanding and will be the subject of a separate research paper. In this case, we expect that additional fast-variables need to be identified. Intuition tells us that $\frac{[SI]}{[I]}$ and $\frac{[II]}{[I]}$ may need to be extended to include $\frac{[SR]}{[R]}$ and $\frac{[RI]}{[I]}$.

The computation of the true $R_0$ for pairwise models can be attempted by considering the next generation matrix approach~\cite{van2002reproduction}. Looking at the pairwise model with the simplest closure and ordering the variables involved in the spreading process as: $[I]$,$[SI]$, the generation of new infectious cases at the the disease-free steady state is given by
\begin{equation}
F = 
 \begin{pmatrix}
  0 & \tau \\
  0 & \tau(n-1)(1-\phi)+\tau\xi\phi\alpha
 \end{pmatrix},
\end{equation}
where the lower right term is obtained from equation~\eqref{eq:pw_clust_simple_closure_3} by looking at the rate of growth of $[SI]$ in terms of $[SI]$ itself, that is
$$\dot{[SI]}=+\tau\xi\frac{[SS]}{[S]}\left((1-\phi)+\phi\frac{N}{n}\frac{[SI]}{[I]}\right)[SI]\simeq \left(\tau(n-1)(1-\phi)+\tau\xi\phi\alpha\right)[SI].$$
Now all other transfers between compartments are summarised in the $V$ matrix, which is given below
\begin{equation}
V = 
 \begin{pmatrix}
  \gamma & 0 \\
  0 & (\tau+\gamma)+\tau\frac{\xi}{n}\alpha\delta\phi
 \end{pmatrix},
\end{equation}
where the lower right term describes the rate at which $[SI]$ pairs move out of this compartment. This is obtained from  equation~\eqref{eq:pw_clust_simple_closure_3} by looking at the rate at which $[SI]$ pairs are depleted as shown below
$$\dot{[SI]}=-\left((\tau+\gamma)+\tau\xi\frac{[SI]}{[S]}(1-\phi)+\phi\tau\xi\frac{[SI]}{[S]}\frac{N[II]}{n[I]^{2}}\right)[SI]\simeq -((\tau+\gamma)+\tau\frac{\xi}{n}\alpha\delta\phi)[SI].$$
Now $R_0$ is given by the leading eigenvalue of $FV^{-1}$, which turns out to be
\begin{equation}
R_0=\frac{\tau n (n-1)-\tau(n-1)(n-\alpha)\phi}{n(\tau+\gamma)+\tau\xi\alpha\delta\phi}.
\label{eq:fv_R_0}
\end{equation}
Obviously, this seems like a rather complicated expression since the quasi-equilibrium values for $\alpha$ and $\delta$ are needed. These are only available as asymptotic expansions in powers of $\phi$. Nevertheless, for $\phi=0$, $R_0=\frac{\tau(n-1)}{\tau+\gamma}$, which agrees perfectly with the two results quoted above. Now writing $R_0=r_0+\phi r_1$, $\alpha=\alpha_0+\phi \alpha_1$ and $\delta=\delta_0+\phi \delta_1$, and plugging these into equation~\eqref{eq:fv_R_0}, leads to
$$r_0=\frac{\tau(n-1)}{\tau+\gamma}\,\,\,\text{and}\,\,\, r_1=-\frac{\tau^2(n-1)}{(\tau+\gamma)^2}\left[\frac{2(\tau+\gamma)}{n\tau}+\frac{(n-1)}{n}\alpha_0\delta_0\right].$$
While the first term in the expansion for $R_0$ agrees with the results quoted above, the second term seems less unlikely to be equivalent to those shown above. This same approach can be used to compute $R_0$ when the compact improved closure is used. We believe that comparing these different ways of computing the epidemic threshold can contribute to reconciling different methods and will lead to more clarity and transparency between various modelling approaches.     
%Our study highlights that the threshiold is model dependent, in this case closure dependent, and we have shown that for the parmeter values used here the pairwise model with the compact improved clsoure leads to epidemic outbreaks more readuly compared to the model using the simple closure. 

The ODE systems for the fast variables are worth investigating in more detail. We expect that these systems will exhibit a number of steady states, some stable and some unstable. Namely, we expect the quasi-steady states to be unstable and the trivial zero steady states to be stable. However, numerical solutions of the cubic polynomials show that other equilibria exist. It will also be worthwhile to compare different models in order to identify the impact of clustering on epidemics by mapping out regions in the parameter space where its effect is strongest. It is known that when the network is dense the effect of clustering is limited and the same holds when the transmission/recovery rates are high/low, respectively. Of course there remains the issue of accounting for degree heterogeneity and this has been addressed to some extent by using percolation type approaches. Using the kind of approach that we presented in this paper may be extended to degree-heterogeneous clustered networks, but this will require more sophisticated models such as effective-degree, or compact/super-compact pairwise models~\cite{simon2015super}. These will no doubt lead to more complex systems which are more challenging to analyse. However, we hope that the results of this paper may encourage other researchers to consider and tackle the challenges posed by modelling epidemic dynamics on clustered networks with heterogeneous degree distributions. Finally, it would be worthwhile to test our findings against explicit stochastic network simulations. This was beyond the scope of the present work, whose focus was on exploiting the presence of fast variables and the use of perturbation analysis to determine the epidemic threshold analytically.

%{\color{red}\textbf{Here we could consider how different PW for network with and without clustering is and try to repeat the analysis for the PW model with the Improved closure (Rosie can attempt this)}}

%%%%%%%%%%%%%%%%%%%%%%%%%%%%%%%%%%%%%%%%%%
\section*{Acknowledgments}
%%%%%%%%%%%%%%%%%%%%%%%%%%%%%%%%%%%%%%%%%%

Rosanna C Barnard acknowledges funding for her PhD studies from the Engineering and Physical Sciences Research Council (EP/M506667/1). P\'eter L. Simon acknowledges support from Hungarian Scientific Research Fund, OTKA, (grant no. 115926).

%%%%%%%%%%%%%%%%%%%%%%%%%%%%%%%%%%%%%%%%%%
\appendix
%%%%%%%%%%%%%%%%%%%%%%%%%%%%%%%%%%%%%%%%%%

%%%%%%%%%%%%%%%%%%%%%%%%%%%%%%%%
\section{Derivation of evolution equations for the fast variables with simple closure} \label{sec:correlation_struct_and_fast_vars_appendix}
%%%%%%%%%%%%%%%%%%%%%%%%%%%%%%%%

We begin by computing differential equations for the variables $\alpha:=\frac{[SI]}{[I]}$ and $\delta:=\frac{[II]}{[I]}$. Differentiating $\alpha=\frac{[SI]}{[I]}$ gives \begin{eqnarray*} \frac{d\alpha}{dt}&=&\frac{[SI]'[I]-[SI][I]'}{[I]^{2}} \\ &=&\frac{[SI]'}{[I]}-\frac{[SI][I]'}{[I]^{2}}, \end{eqnarray*} and substituting $[SI]'$ from equation \eqref{eq:pw_clust_simple_closure_3} and $[I]'$ from equation \eqref{eq:pw_clust_simple_closure_2}, we obtain \begin{dmath*} \frac{d\alpha}{dt}=-(\tau+\gamma)\frac{[SI]}{[I]}+\tau\xi\frac{[SS][SI]}{[S][I]}\left((1-\phi)+\phi\frac{N[SI]}{n[S][I]}\right)-\tau\xi\frac{[SI]^{2}}{[S][I]}\left((1-\phi)+\phi\frac{N[II]}{n[I]^{2}}\right)-\tau\frac{[SI]^{2}}{[I]^{2}}+\gamma\frac{[SI]}{[I]}. \end{dmath*} Replacing all $\frac{[SI]}{[I]}$ terms by $\alpha$ and all $\frac{[II]}{[I]}$ terms by $\delta$ gives \begin{align*} \frac{d\alpha}{dt}&=-(\tau+\gamma)\alpha+\tau\xi\frac{[SS]}{[S]}\alpha\left((1-\phi)+\phi\frac{N}{n[S]}\alpha\right)-\tau\xi\frac{[SI]}{[S]}\alpha\left((1-\phi)+\phi\frac{N}{n[I]}\delta\right)-\tau\alpha^{2}+\gamma\alpha \\ &=-\tau\alpha+\tau\xi\frac{[SS]}{[S]}\alpha\left((1-\phi)+\phi\frac{N}{n[S]}\alpha\right)-\tau\xi\frac{[SI]}{[S]}\alpha\left((1-\phi)+\phi\frac{N}{n[I]}\delta\right)-\tau\alpha^{2} \\ &=-\tau\alpha+\tau\xi\frac{[SS]}{[S]}(1-\phi)\alpha+\tau\xi\phi\frac{N[SS]}{n[S]^{2}}\alpha^{2}-\tau\xi\frac{[SI]}{[S]}(1-\phi)\alpha-\tau\xi\frac{N[SI]}{n[S][I]}\phi\alpha\delta-\tau\alpha^{2} \\ &=-\tau\alpha+\tau\xi\frac{[SS]}{[S]}(1-\phi)\alpha+\tau\xi\phi\frac{N[SS]}{n[S]^{2}}\alpha^{2}-\tau\xi\frac{[SI]}{[S]}(1-\phi)\alpha-\tau\xi\frac{N}{n[S]}\phi\alpha^{2}\delta-\tau\alpha^{2} \end{align*} In section \ref{sec:ep_threshold} we considered an epidemic threshold and a condition for stability of the disease-free steady state. In both cases, we consider the state of the system at time zero. At time zero we assume that $[S]=N$,  $[SS]=nN$, and $[SI]=0$, therefore we substitute these values into the differential equation for $\alpha$ to obtain \begin{align*} \frac{d\alpha}{dt}&=-\tau\alpha+\tau\xi\frac{nN}{N}(1-\phi)\alpha+\tau\xi\phi\frac{NnN}{nN^{2}}\alpha^{2}-\tau\xi\frac{N}{nN}\phi\alpha^{2}\delta-\tau\alpha^{2} \\ &=-\tau\alpha+\tau\xi n(1-\phi)\alpha+\tau\xi\phi\alpha^{2}-\tau\xi\frac{1}{n}\phi\alpha^{2}\delta-\tau\alpha^{2}, \end{align*} where $\xi=\frac{(n-1)}{n}$. Differentiating $\delta=\frac{[II]}{[I]}$ gives \begin{eqnarray*} \frac{d\delta}{dt}&=&\frac{[II]'[I]-[II][I]'}{[I]^{2}} \\ &=&\frac{[II]'}{[I]}-\frac{[II][I]'}{[I]^{2}}, \end{eqnarray*} and substituting $[II]'$ from equation \eqref{eq:pw_clust_simple_closure_end} and $[I]'$ from equation \eqref{eq:pw_clust_simple_closure_2}, we obtain \begin{dmath*} \frac{d\delta}{dt}=2\tau\frac{[SI]}{[I]}-2\gamma\frac{[II]}{[I]}+2\tau\xi\frac{[SI]^{2}}{[S][I]}\left((1-\phi)+\phi\frac{N[II]}{n[I]^{2}}\right)-\tau\frac{[SI][II]}{[I]^{2}}+\gamma\frac{[II]}{[I]}. \end{dmath*} Replacing all $\frac{[SI]}{[I]}$ terms by $\alpha$ and all $\frac{[II]}{[I]}$ terms by $\delta$ gives \begin{align*} \frac{d\delta}{dt}&=2\tau\alpha-2\gamma\delta+2\tau\xi\frac{[SI]}{[S]}\alpha\left((1-\phi)+\phi\frac{N}{n[I]}\delta\right)-\tau\alpha\delta+\gamma\delta \\ &=2\tau\alpha-\gamma\delta+2\tau\xi\frac{[SI]}{[S]}\alpha\left((1-\phi)+\phi\frac{N}{n[I]}\delta\right)-\tau\alpha\delta \\ &=2\tau\alpha-\gamma\delta+2\tau\xi\frac{[SI]}{[S]}(1-\phi)\alpha+2\tau\xi\frac{N[SI]}{n[S][I]}\phi\alpha\delta-\tau\alpha\delta \\ &=2\tau\alpha-\gamma\delta+2\tau\xi\frac{[SI]}{[S]}(1-\phi)\alpha+2\tau\xi\frac{N}{n[S]}\phi\alpha^{2}\delta-\tau\alpha\delta. \end{align*} At time zero we assume that $[S]=N$ and $[SI]=0$. We substitute these values into the differential equation for $\delta$ to obtain \begin{align*} \frac{d\delta}{dt}&=2\tau\alpha-\gamma\delta+2\tau\xi\frac{1}{n}\phi\alpha^{2}\delta-\tau\alpha\delta. \end{align*} Combining the differential equations for both $\alpha=\frac{[SI]}{[I]}$ and $\delta=\frac{[II]}{[I]}$, we have \begin{eqnarray} \frac{d\alpha}{dt}&=&-\tau\alpha+\tau\xi n(1-\phi)\alpha+\tau\xi\phi\alpha^{2}-\tau\xi\frac{1}{n}\phi\alpha^{2}\delta-\tau\alpha^{2}, \\ \frac{d\delta}{dt}&=&2\tau\alpha-\gamma\delta+2\tau\xi\frac{1}{n}\phi\alpha^{2}\delta-\tau\alpha\delta. \end{eqnarray}

%%%%%%%%%%%%%%%%%%%%%%%%%%%%%%%%%%%%%%%%%%%%%%%%%%%%
\section{Derivation of evolution equations for the fast variables with the compact improved closure}
\label{sec:correlation_struct_and_fast_vars_reduced_improved_closure}
%%%%%%%%%%%%%%%%%%%%%%%%%%%%%%%%%%%%%%%%%%%%%%%%%%%%

Using the improved closure \eqref{eq:improved_closure} in line with Proposition \ref{prop:inequ_closures}, which we refer to as the reduced improved closure, we find that \begin{align} [ASI]&=(n-1)\left((1-\phi)\frac{[AS][SI]}{n[S]}+\phi\frac{[AS][SI][IA]}{[A]\sum_{a}[aS][aI]/[a]}\right) \\ &=(n-1)\left((1-\phi)\frac{[AS][SI]}{n[S]}+\phi\frac{[AS][SI][IA]}{[A]\left(\frac{[SS][SI]}{[S]}+\frac{[SI][II]}{[I]}\right)}\right). \label{eq:reduced_improved_closure} \end{align} Using equation \eqref{eq:reduced_improved_closure} to close the original pairwise equations \eqref{equations:unclosed_pw_model_one}-\eqref{equations:unclosed_pw_model_end}, we obtain the following system of equations: \begin{eqnarray} \dot{[S]}&=&-\tau[SI] \\ \dot{[I]}&=&\tau[SI]-\gamma[I] \label{eq:dot_I_reduced_imp_close} \end{eqnarray} \begin{dmath} \label{eq:dot_SI_reduced_imp_close} \dot{[SI]}=-(\tau+\gamma)[SI]+\tau(n-1)\left((1-\phi)\frac{[SS][SI]}{n[S]}+\phi\frac{[I][SS][SI]}{[I][SS]+[S][II]}\right)-\tau(n-1)\left((1-\phi)\frac{[SI]^{2}}{n[S]}+\phi\frac{[S][SI][II]}{[I][SS]+[S][II]}\right) \end{dmath} \begin{dmath} \dot{[SS]}=-2\tau(n-1)\left((1-\phi)\frac{[SS][SI]}{n[S]}+\phi\frac{[I][SS][SI]}{[I][SS]+[S][II]}\right) \end{dmath} \begin{dmath} \dot{[II]}=2\tau[SI]-2\gamma[II]+2\tau(n-1)\left((1-\phi)\frac{[SI]^{2}}{n[S]}+\phi\frac{[S][SI][II]}{[I][SS]+[S][II]}\right). \label{eq:dot_II_reduced_imp_close} \end{dmath} 

As we have shown in the main body of the paper, the computation of the threshold requires a system of differential equations for the fast variables $\alpha=[SI]/[I]$ and $\delta=[II]/[I]$. We find \begin{dmath*} \frac{d\alpha}{dt}=\frac{[SI]'}{[I]}-\frac{[SI][I]'}{[I]^{2}} \end{dmath*} and substituting $[SI]'$ from equation \eqref{eq:dot_SI_reduced_imp_close} and $[I]'$ from equation \eqref{eq:dot_I_reduced_imp_close}, we obtain \begin{dmath} \frac{d\alpha}{dt}=-(\tau+\gamma)\frac{[SI]}{[I]}+\tau(n-1)\left((1-\phi)\frac{[SS][SI]}{n[S][I]}+\phi\frac{[SS][SI]}{[I][SS]+[S][II]}\right)-\tau(n-1)\left((1-\phi)\frac{[SI]^{2}}{n[S][I]}+\phi\frac{[S][SI][II]}{[I]^{2}[SS]+[S][I][II]}\right)-\tau\frac{[SI]^{2}}{[I]^{2}}+\gamma\frac{[SI]}{[I]}. \end{dmath} Replacing all $\frac{[SI]}{[I]}$ terms by $\alpha$ and all $\frac{[II]}{[I]}$ terms by $\delta$ gives \begin{dmath} \frac{d\alpha}{dt}=-(\tau+\gamma)\alpha+\tau(n-1)\left((1-\phi)\frac{[SS]}{n[S]}\alpha+\phi\alpha\frac{[SS]}{[SS]+[S]\delta}\right)-\tau(n-1)\left((1-\phi)\frac{[SI]}{n[S]}\alpha+\phi\alpha\delta\frac{[S]}{[SS]+[S]\delta}\right)-\tau\alpha^{2}+\gamma\alpha, \end{dmath} and evaluating $\frac{d\alpha}{dt}$ at the disease-free steady state $([S],[I],[SI],[SS],[II])=(N,0,0,nN,0)$ \eqref{eq:disease_free_st_state} gives \begin{dmath} \frac{d\alpha}{dt}=-(\tau+\gamma)\alpha+\tau(n-1)\left((1-\phi)\alpha+\phi\alpha\frac{nN}{nN+N\delta}\right)-\tau(n-1)\left(\phi\alpha\delta\frac{N}{nN+N\delta}\right)-\tau\alpha^{2}+\gamma\alpha. \end{dmath} After simplification we find that \begin{eqnarray} \frac{d\alpha}{dt}&=&-\tau\alpha+\tau(n-1)\left((1-\phi)\alpha+\phi\alpha\frac{n}{n+\delta}-\phi\alpha\delta\frac{1}{n+\delta}\right)-\tau\alpha^{2} \\ &=&-\tau\alpha+\tau(n-1)\left((1-\phi)\alpha+\phi\alpha\left(\frac{n-\delta}{n+\delta}\right)\right)-\tau\alpha^{2}. \end{eqnarray} Differentiating $\delta=\frac{[II]}{[I]}$ gives \begin{dmath*} \frac{d\delta}{dt}=\frac{[II]'}{[I]}-\frac{[II][I]'}{[I]^{2}}, \end{dmath*} and substituting $[II]'$ from equation \eqref{eq:dot_II_reduced_imp_close} and $[I]'$ from equation \eqref{eq:dot_I_reduced_imp_close}, we obtain \begin{dmath} \frac{d\delta}{dt}=2\tau\frac{[SI]}{[I]}-2\gamma\frac{[II]}{[I]}+2\tau(n-1)\left((1-\phi)\frac{[SI]^{2}}{n[S][I]}+\phi\frac{[S][SI][II]}{[I]^{2}[SS]+[S][I][II]}\right)-\tau\frac{[SI][II]}{[I]^{2}}+\gamma\frac{[II]}{[I]}. \end{dmath} Replacing all $\frac{[SI]}{[I]}$ terms by $\alpha$ and all $\frac{[II]}{[I]}$ terms by $\delta$ gives \begin{eqnarray*} \frac{d\delta}{dt}&=&2\tau\alpha-2\gamma\delta+2\tau(n-1)\left((1-\phi)\frac{[SI]}{n[S]}\alpha+\phi\alpha\delta\frac{[S]}{[SS]+[S]\delta}\right)-\tau\alpha\delta+\gamma\delta \\ &=&2\tau\alpha-\gamma\delta+2\tau(n-1)\left((1-\phi)\frac{[SI]}{n[S]}\alpha+\phi\alpha\delta\frac{[S]}{[SS]+[S]\delta}\right)-\tau\alpha\delta, \end{eqnarray*} and evaluating $\frac{d\delta}{dt}$ at the disease-free steady state \eqref{eq:disease_free_st_state} gives \begin{eqnarray} \frac{d\delta}{dt}&=&2\tau\alpha-\gamma\delta+2\tau(n-1)\left(\phi\alpha\delta\frac{N}{nN+N\delta}\right)-\tau\alpha\delta \\ &=&2\tau\alpha-\gamma\delta+2\tau(n-1)\left(\phi\alpha\delta\frac{1}{n+\delta}\right)-\tau\alpha\delta. \end{eqnarray} Combining the differential equations for both $\alpha=\frac{[SI]}{[I]}$ and $\delta=\frac{[II]}{[I]}$, we have \begin{eqnarray} \frac{d\alpha}{dt}&=&-\tau\alpha+\tau(n-1)\left((1-\phi)\alpha+\phi\alpha\left(\frac{n-\delta}{n+\delta}\right)\right)-\tau\alpha^{2} \\ \frac{d\delta}{dt}&=&2\tau\alpha-\gamma\delta+2\tau(n-1)\left(\frac{\phi\alpha\delta}{n+\delta}\right)-\tau\alpha\delta. \end{eqnarray}

%%%%%%%%%%%%%%%%%%%%%%%%%%%%%%%%%%%%%%%%%%%%%%%%%%%%
\section{Standard linear-stability analysis for the case of the simple closure}
\label{sec:lin_stab_anal_simp_clo}
%%%%%%%%%%%%%%%%%%%%%%%%%%%%%%%%%%%%%%%%%%%%%%%%%%%%

An alternative way to determine the epidemic threshold is to consider the  stability of the disease-free steady state \begin{equation} ([S],[I],[SI],[SS],[II])=(N,0,0,nN,0). \label{eq:disease_free_st_state} \end{equation} When the disease-free steady state is stable, the system will always end up at the disease-free steady state and thus no epidemic will occur. When the disease-free steady state becomes unstable, there exists (at least) a second steady state whereby an epidemic will occur and $[S]$ will no longer be equal to $N$. To determine a stability condition for the disease-free steady state \eqref{eq:disease_free_st_state}, we must compute the Jacobian matrix $J$ of the system \eqref{eq:pw_clust_simple_closure_1}-\eqref{eq:pw_clust_simple_closure_end}, evaluated at the disease-free steady state, and solve to find its eigenvalues. 

By computing partial derivatives of each differential equation \eqref{eq:pw_clust_simple_closure_1}-\eqref{eq:pw_clust_simple_closure_end} with respect to each model variable $[S]$, $[I]$, $[SI]$, $[SS]$ and $[II]$, and evaluating each expression at the disease-free steady state \eqref{eq:disease_free_st_state}, we obtain \begin{equation} \label{eq:jacobian_disease_free_st_st} J_{df}=\begin{pmatrix} 0 & 0 & -\tau & 0 & 0 \\ 0 & -\gamma & \tau & 0 & 0 \\ 0 & \frac{\partial\dot{[SI]}}{\partial[I]} & \frac{\partial\dot{[SI]}}{\partial[SI]} & 0 & \frac{\partial\dot{[SI]}}{\partial[II]} \\ 0 & \frac{\partial\dot{[SS]}}{\partial[I]} & \frac{\partial\dot{[SS]}}{\partial[SI]} & 0 & 0 \\ 0 & \frac{\partial\dot{[II]}}{\partial[I]} & \frac{\partial\dot{[II]}}{\partial[SI]} & 0 & \frac{\partial\dot{[II]}}{\partial[II]} \end{pmatrix}, \end{equation} with $\frac{\partial\dot{[SI]}}{\partial[I]}=\tau\xi\phi\left(\frac{2[SI]^{2}[II]}{n[I]^{3}}-\frac{[SI]^{2}}{[I]^{2}}\right)$, $\frac{\partial\dot{[SI]}}{\partial[SI]}=-(\tau+\gamma)+\tau\xi(1-\phi)n+2\tau\xi\phi\left(\frac{[SI]}{[I]}-\frac{[SI][II]}{n[I]^{2}}\right)$, $\frac{\partial\dot{[SI]}}{\partial[II]}=-\tau\xi\phi\frac{[SI]^{2}}{n[I]^{2}}$, $\frac{\partial\dot{[SS]}}{\partial[I]}=2\tau\xi\phi\frac{[SI]^{2}}{[I]^{2}}$, $\frac{\partial\dot{[SS]}}{\partial[SI]}=-2\tau\xi(1-\phi)n-4\tau\xi\phi\frac{[SI]}{[I]}$, $\frac{\partial\dot{[II]}}{\partial[I]}=-4\tau\xi\phi\frac{[SI]^{2}[II]}{n[I]^{3}}$, $\frac{\partial\dot{[II]}}{\partial[SI]}=2\tau+4\tau\xi\phi\frac{[SI][II]}{n[I]^{2}}$ and $\frac{\partial\dot{[II]}}{\partial[II]}=-2\gamma+2\tau\xi\phi\frac{[SI]^{2}}{n[I]^{2}}$ all containing variables $\frac{[SI]}{[I]}$ and $\frac{[II]}{[I]}$. The zero entries in $J_{df}$ reflect the true values that the respective partial derivatives attain at the disease-free equilibrium. However, the majority of the non-zero matrix entries involve $\frac{[SI]}{[I]}$ and $\frac{[II]}{[I]}$. Since $[I]=[SI]=[II]=0$ at the disease-free steady state, both of these quantities are ill-defined. Hence, not all entries of the Jacobian can be evaluated at the equilibrium. This issue prevents the computation of the eigenvalues of $J_{df}$ and thus the value of the epidemic threshold. In order to progress, we need to determine the correct values for  $\alpha=\frac{[SI]}{[I]}$ and $\delta=\frac{[II]}{[I]}$. We note that the correct value of $\alpha=\frac{[SI]}{[I]}$ is also required in equation~\eqref{eq:new_threshold}, and the threshold cannot be computed without it.

In fact, using only $\phi=0$, the Jacobian at the disease-free steady state \eqref{eq:disease_free_st_state} becomes \begin{equation} J_{df\_no\_clust}=\begin{pmatrix} 0 & 0 & -\tau & 0 & 0 \\ 0 & -\gamma & \tau & 0 & 0 \\ 0 & 0 & -\gamma+\tau(n-2) & 0 & 0 \\ 0 & 0 & -2\tau(n-1) & 0 & 0 \\ 0 & 0 & 2\tau & 0 & -2\gamma \end{pmatrix}. \end{equation} %and thus we do not require the steady state values for $\alpha=\frac{[SI]}{[I]}$ and $\delta=\frac{[II]}{[I]}$ to determine stability conditions for the disease-free steady state. Calculating the determinant of $(J_{df\_no\_clust}-\lambda I)$, we find \begin{eqnarray*} \begin{vmatrix} -\lambda & 0 & -\tau & 0 & 0 \\ 0 & -\gamma-\lambda & \tau & 0 & 0 \\ 0 & 0 & -\gamma+\tau(n-2)-\lambda & 0 & 0 \\ 0 & 0 & -2\tau(n-1) & -\lambda & 0 \\ 0 & 0 & 2\tau & 0 & -2\gamma-\lambda \end{vmatrix} \\ (-\lambda)\begin{vmatrix} -\gamma-\lambda & \tau & 0 & 0 \\ 0 & -\gamma+\tau(n-2)-\lambda & 0 & 0 \\ 0 & -2\tau(n-1) & -\lambda & 0 \\ 0 & 2\tau & 0 & -2\gamma-\lambda \end{vmatrix}-\tau\begin{vmatrix} 0 & -\gamma-\lambda & 0 & 0 \\ 0 & 0 & 0 & 0 \\ 0 & 0 & -\lambda & 0 \\ 0 & 0 & 0 & -2\gamma-\lambda \end{vmatrix} \\ \lambda(\gamma+\lambda)\begin{vmatrix} -\gamma+\tau(n-2)-\lambda & 0 & 0 \\ -2\tau(n-1) & -\lambda & 0 \\ 2\tau & 0 & -2\gamma-\lambda \end{vmatrix}+\lambda\tau\begin{vmatrix} 0 & 0 & 0 \\ 0 & -\lambda & 0 \\ 0 & 0 & -2\gamma-\lambda \end{vmatrix}-\tau(\gamma+\lambda)\begin{vmatrix} 0 & 0 & 0 \\ 0 & -\lambda & 0 \\ 0 & 0 & -2\gamma-\lambda \end{vmatrix} \\ \lambda(\gamma+\lambda)(-\gamma+\tau(n-2)-\lambda)\begin{vmatrix} -\lambda & 0 \\ 0 & -2\gamma-\lambda \end{vmatrix} \\ \lambda(\gamma+\lambda)(-\gamma+\tau(n-2)-\lambda)(-\lambda)(-2\gamma-\lambda), \end{eqnarray*} and setting the determinant of $(J_{df\_no\_clust}-\lambda I)$ equal to zero gives the 
It is straightforward to show that the eigenvalues are given by $\lambda_{1}=0$, $\lambda_{2}=-\gamma$, $\lambda_{3}=\tau(n-2)-\gamma$, $\lambda_{4}=0$ and $\lambda_{5}=-2\gamma$. The only eigenvalue that can be non-zero and non-negative is $\lambda_{3}=\tau(n-2)-\gamma$. Hence, we know that the disease-free steady state \eqref{eq:disease_free_st_state} is stable when $\lambda_{3}\leq 0$ and becomes unstable when $\lambda_{3}>0$. Thus, the epidemic threshold is given by $\lambda_{3}=0$ and this can be rearranged to give $\tau (n-2)/\gamma=1$. This is equivalent to the calculation based on determining the quasi-equilibrium of the fast variables. 

%%%%%%%%%%%%%%%%%%%%%%%%%%%%%%%%%%%%%%%%%%%%%%%%%%%%
\section{Standard linear-stability analysis for the case of the compact improved closure}
\label{sec:lin_stab_anal_comp_imp_clo}
%%%%%%%%%%%%%%%%%%%%%%%%%%%%%%%%%%%%%%%%%%%%%%%%%%%%

To determine an epidemic threshold, we consider conditions for stability of the disease-free steady state \eqref{eq:disease_free_st_state}. To do so, we compute the Jacobian matrix evaluated at the disease-free steady state as \begin{equation} J_{df_2}=\begin{pmatrix} 0 & 0 & -\tau & 0 & 0 \\ 0 & -\gamma & \tau & 0 & 0 \\ 0 & \frac{\partial\dot{[SI]}}{\partial[I]} & \frac{\partial\dot{[SI]}}{\partial[SI]} & 0 & \frac{\partial\dot{[SI]}}{\partial[II]} \\ 0 & \frac{\partial\dot{[SS]}}{\partial[I]} & \frac{\partial\dot{[SS]}}{\partial[SI]} & 0 & \frac{\partial\dot{[SS]}}{\partial[II]} \\ 0 & \frac{\partial\dot{[II]}}{\partial[I]} & \frac{\partial\dot{[II]}}{\partial[SI]} & 0 & \frac{\partial\dot{[II]}}{\partial[II]} \end{pmatrix} \end{equation} where $\frac{\partial\dot{[SI]}}{\partial[I]}=2\tau(n-1)\phi\alpha\delta\frac{n}{n^{2}+2n\delta+\delta^{2}}$, $\frac{\partial\dot{[SI]}}{\partial[SI]}=-(\tau+\gamma)+\tau(n-1)\left((1-\phi)+\phi\left(\frac{n-\delta}{n+\delta}\right)\right)$, $\frac{\partial\dot{[SI]}}{\partial[II]}=-2\tau(n-1)\left(\phi\alpha\frac{n}{n^{2}+2n\delta+\delta^{2}}\right)$, $\frac{\partial\dot{[SS]}}{\partial[I]}=-2\tau(n-1)\left(\phi\alpha\delta\frac{n}{n^{2}+2n\delta+\delta^{2}}\right)$, $\frac{\partial\dot{[SS]}}{\partial[SI]}=-2\tau(n-1)\left((1-\phi)+\phi\frac{n}{n+\delta}\right)$, $\frac{\partial\dot{[SS]}}{\partial[II]}=2\tau(n-1)\left(\phi\alpha\frac{n}{n^{2}+2n\delta+\delta^{2}}\right)$, $\frac{\partial\dot{[II]}}{\partial[I]}=-2\tau(n-1)\left(\phi\alpha\delta\frac{n}{n^{2}+2n\delta+\delta^{2}}\right)$, $\frac{\partial\dot{[II]}}{\partial[SI]}=2\tau+2\tau(n-1)\left(\phi\delta\frac{1}{n+\delta}\right)$ and $\frac{\partial\dot{[II]}}{\partial[II]}=-2\gamma+2\tau(n-1)\left(\phi\alpha\frac{n}{n^{2}+2n\delta+\delta^{2}}\right)$ cannot be fully evaluated as they contain products of the problematic variables $\alpha=\frac{[SI]}{[I]}$ and $\delta=\frac{[II]}{[I]}$. 

The Jacobian above becomes useful once  analytical expressions for $\alpha$ and $\delta$ are obtained (or it could be an asymptotic expansion or even numerical values). Plugging these in the Jacobian will allow to either numerically or analytically compute the threshold. We note that using the linear-stability analysis or focusing on the initial growth rate should lead to the same threshold value, as was already shown for the for the case of the system with the simple closure in Section~\ref{sec:lin_stab_anal_simp_clo}.

%\eqref{eq:disease_free_st_state} becomes unstable when \begin{eqnarray*} \tau(n-2)-\gamma>0 \\ \tau(n-2)>\gamma \\ n-2>\frac{\gamma}{\tau} \\ n>\frac{\gamma}{\tau}+2. \end{eqnarray*} 

% TODO - decide which additional appendices we want to use

% %%%%%%%%%%%%%%%%%%%%%%%%%%%%%%%%%%%%%%%%%%
% \subsection{Jacobian of the PW model for clustered networks with the simple closure}
% %%%%%%%%%%%%%%%%%%%%%%%%%%%%%%%%%%%%%%%%%%

% %%%%%%%%%%%%%%%%%%%%%%%%%%%%%%%%%%%%%%%%%%
% \subsection{Fast variables for clustered networks with the simple closure}
% %%%%%%%%%%%%%%%%%%%%%%%%%%%%%%%%%%%%%%%%%%

% %%%%%%%%%%%%%%%%%%%%%%%%%%%%%%%%
% \subsection{Leftovers}
% %%%%%%%%%%%%%%%%%%%%%%%%%%%%%%%%

% %%%%%%%%%%%%%%%%%%%%%%%%%%%%%%%%
% \subsection{The pairwise model for networks without clustering and with simple closure}
% %%%%%%%%%%%%%%%%%%%%%%%%%%%%%%%%

%%%%%%%%%%%%%%%%%%%%%%
%\section{Bibliography}
%%%%%%%%%%%%%%%%%%%%%%
\bibliography{sample}
\bibliographystyle{plain}

\end{document}